\documentclass[a4paper,onecolumn,11pt,unpublished]{quantumarticle}
\pdfoutput=1
\usepackage[utf8]{inputenc}
\usepackage[english]{babel}
\usepackage[T1]{fontenc}
\usepackage{amsmath}
\usepackage{hyperref}
\usepackage{xcolor}

\usepackage{algorithm}
\usepackage{algpseudocode}

\usepackage{tikz}
\usepackage{lipsum}
\usepackage{booktabs}
\usepackage{quantikz}
\usepackage{amsfonts}
\usepackage{multirow}
\usepackage{multicol}
\usepackage{silence}
\usepackage{subcaption}
\usepackage{amsthm}
\usepackage[numbers,sort&compress]{natbib}
\newtheorem{definition}{Definition}
\newtheorem{theorem}{Theorem}
\newtheorem{lemma}{Lemma}
\newtheorem{proposition}{Proposition}
\newtheorem{conjecture}{Conjecture}
\newtheorem{assumption}{Assumption}

\DeclareMathOperator*{\argmax}{arg\,max}
\newif\ifsupp
\supptrue 

\newif\ifarxivpreprint
\arxivpreprinttrue

\newcommand{\bref}[2]{
  \ifsupp
    \ref{#1}
  \else
    #2
  \fi
}

\begin{document}

\title{Understanding the Nature of Depth-1 Equivariant Quantum Circuit} 

\author{Jonathan TEO\textsuperscript{1}}
\email{jrteo.2022@smu.edu.sg}
\orcid{0009-0001-2132-9900}

\author{Xin Wei LEE\textsuperscript{1}}
\email{xwlee@smu.edu.sg}
\orcid{0000-0002-8509-8697}

\author{Hoong Chuin LAU\textsuperscript{1}}
\email{hclau@smu.edu.sg (Corresponding Author)}
\orcid{0000-0002-5326-411X}

\affiliation{\textsuperscript{1}School of Computing and Information Systems, Singapore Management University, Singapore 178902}

\maketitle

\begin{abstract}
    The Equivariant Quantum Circuit (EQC) for the Travelling Salesman Problem (TSP) has been shown to achieve near-optimal performance in solving small TSP problems ($\leq$ 20 nodes) using only two parameters at depth 1. However, extending EQCs to larger TSP problem sizes remains challenging due to the exponential time and memory for quantum circuit simulation, as well as increasing noise and decoherence when running on actual quantum hardware.
    In this work, we propose the \emph{Size-Invariant Grid Search} (SIGS), an efficient training optimization for Quantum Reinforcement Learning (QRL), and use it to simulate the outputs of a trained Depth-1 EQC up to 350-node TSP instances---well beyond previously tractable limits. 
    At TSP with 100 nodes, we reduce total simulation times by 96.4\%, when comparing to RL simulations with the analytical expression (151 minutes using RL to under 6 minutes using SIGS on TSP-100), while achieving a mean optimality gap within 0.005 of the RL trained model on the test set. 
    SIGS provides a practical benchmarking tool for the QRL community, allowing us to efficiently analyze the performance of QRL algorithms on larger problem sizes. 
    We provide a theoretical explanation for SIGS called the \emph{Size-Invariant Properties} that goes beyond the concept of equivariance discussed in prior literature. 
\end{abstract}
\section{Introduction} \label{sect:intro}

Quantum reinforcement learning (QRL) studies parameterized quantum circuits (PQCs) as function approximators within deep reinforcement learning (DRL) pipelines. A common instantiation replaces the neural network in a Deep Q-Network (DQN) with a PQC trained end to end by a classical optimizer, yielding a hybrid agent. Early demonstrations show that PQC-based DQNs can learn near-optimal policies on small benchmarks (e.g., FrozenLake, Cognitive Radio), establishing feasibility and hinting at alternative inductive biases relative to classical value approximators~\cite{chen-qrl-2020}.

This paper is concerned with solving combinatorial optimization problems with QRL. In the classical domain, \emph{Neural Combinatorial Optimization} (NCO) develops DRL/attention architectures that construct high-quality solutions to discrete problems, achieving near-optimal tours for symmetric TSP up to 100 nodes~\cite{Kool2019Attention}; more recent diffusion/backbone advances further improve the quality--speed trade-off~\cite{Li2024FastT2T}.

For QRL to match or even surpass NCO counterparts, two gaps remain. \textbf{Interpretability.} PQC architectures with far fewer parameters than classical networks can perform well, yet the role of each parameter is often opaque for concrete combinatorial tasks. \textbf{Scalability.} Prevalent encodings map one decision variable to one qubit and rely on highly entangling layers; both simulation and hardware execution grow costly with problem size, even with efficient tensor-network simulators. These gaps motivate the developments in this paper.

\subsection{Overview of the Literature}
\textbf{QRL for Combinatorial Optimization}. QRL involves replacing the value function approximator in \textit{Deep Reinforcement Learning} (DRL) with a \textit{Parameterized Quantum Circuit} (PQC). 
Skolik et al. (2023) marked the first systematic application of QRL to Combinatorial Optimization via the Equivariant Quantum Circuit (EQC) for TSP \cite{Skolik2022EquivariantQC}. Using a \textit{Symmetry Preserving Ansatz} (SPA), they demonstrated that respecting graph isomorphisms can drastically reduce the number of trainable parameters of the circuit down to two values, while producing near-optimal tours for instances of up to 20 nodes. 
This also serves as a starting point of our work. More recently, Kruse et al. (2024) generalized this framework to \textit{Quadratic Unconstrained Binary Optimization} (QUBO) formulations, with an \textit{sge-sgv} ansatz that closely related to the EQC (having also two parameters per layer), demonstrating the ability to learn near-optimal policies on other Combinatorial Optimization problems such as \textit{Weighted-MaxCut}, \textit{Knapsack} and the \textit{Unit Containment Problem} \cite{Kruse_2024_HEA}.  

\noindent\textbf{Ansatz design and trainability}.
An \textit{ansatz} specifies the \textit{a certain structure} of a PQC to prepare a variational state or wavefunction \cite{Tilly2022}. The structure of the ansatz can be chosen with or without prior knowledge of the problem. In QRL, the most popular ansatz used is the \textit{Hardware Efficient Ansatz} (HEA) \cite{chen-qrl-2020, Kruse_2024_HEA, dragan2022qrlHEA} due to their device-friendly layouts. However, due to its high expressivity and problem independence, these ansatz can trigger \textit{Barren Plateaus} (BP), making gradients vanish as problem sizes increase \cite{BP_Holmes_2022,BP_McClean_2018}. This has fueled a shift toward problem-tailoned ansatzes for QRL (to be symmetry-aware) in order to improve trainability at shallow depths.

We can interpret the work of Skolik et al. (2023) in three ways. Firstly, from a \textit{Geometric Quantum Machine Learning} (GQML) perspective, Schatzki et al. (2024) formalizes SPAs that are equivariant under the symmetry group $S_n$, and prove that $S_n$-equivariant ansatzes do not suffer from barren plateaus, quickly reach over-parameterization, and generalize well from small amounts of data \cite{Schatzki_2024, Larocca_2023}. Both the EQC and the \textit{sge-sgv} ansatzes are examples of $S_n$-equivariant \textit{Quantum Neural Networks} (QNNs). Secondly, Skolik et al. (2023) further empirically demonstrates that SPAs outperform their non-symmetry preserving counterparts through ablation simulations which gradually break the equivariance of the EQC \cite{Skolik2022EquivariantQC}. An alternate lens for interpreting the EQC is via the \textit{Quantum Approximate Optimization Algorithm} (QAOA). The structure of the EQC mirrors the structure of the QAOA ansatz: where the parameters $\gamma$ and $\beta$ play analogous roles to the cost and mixer Hamiltonians \cite{farhi2014quantumapproximateoptimizationalgorithm}. QAOA is well-known in solving combinatorial optimization problems, prominently on the MaxCut problem with proven bounds~\cite{Wurtz_2021,Marwaha_2021} and empirical benchmarks~\cite{leo2020,Lotshaw_2021}. Lastly, the idea of ensuring that embeddings of graphs to be invariant under node permutations closely follows to the design principles of classical NCO algorithms that respect graph structures through message passing and attention \cite{gillman2024combiningreinforcementlearningtensor, kaempfer2019learningmultipletravelingsalesmen, kool2019attentionlearnsolverouting}.

\subsection{Contributions}
Our work makes the following contributions:

\begin{itemize}
    \item \textbf{A size-invariant understanding of parameters:} We provide a new explanation for the roles of the trainable parameters ($\gamma$ and $\beta$) at Depth~1 of the EQC, extending beyond the concept of equivariance in \cite{Skolik2022EquivariantQC}. By analyzing the closed-form expressions of the Depth~1 EQC \cite{akshay2021parameter,analytical-qaoa,exp-qaoa-ising,qaoa-ferm}, we show how these parameters govern effective QRL policies for TSP. On one hand, while QAOA uses variational optimization to minimize Hamiltonian expectations, QRL (using Q-Learning) involves optimizing the Temporal-Difference Loss \cite{sutton-barto-2018}, with Q-values derived from $ZZ$ expectations and guided by a \textit{Markov Decision Process}(MDP). This fundamental methodological difference motivates our analysis of Depth~1 EQCs in reinforcement learning.
    
    \item \textbf{Size-invariant trainability:}. As a consequence of the size-invariant interpretation of parameters, we conclude that effective QRL policies with Depth~1 EQCs reside in a compact, problem size-invariant parameter region. Leveraging this, we propose the \textit{Size-Invariant Grid Search} (SIGS), which reduces training time for EQC-based QRL by 27 times (e.g., 151 minutes (using the Analytical Expression) to 5.43 minutes for TSP with 100 nodes), and enables efficient simulation of instances with up to 350 nodes (longest run completes under $140$ minutes). This allows us to analyze scalability of Depth~1 EQCs well beyond the tractable limits of classical quantum simulation.
\end{itemize}
\textbf{Limitation}. Our findings are restricted to Depth 1 EQCs. While, in theory, deeper circuits should improve the agent's performance, our empirical results at EQC Depths 2 to 4 only increases the time cost of simulation, with negligible difference in performance when compared to Depth 1 (see Section \ref{sect:discussion}). This negligible difference in performance is consistent with prior Depth-4 results in \cite{Skolik2022EquivariantQC}. 

The remainder of the manuscript proceeds as follows: Section \ref{sect:background} outlines the existing methodology in Skolik et al. (2023)\cite{Skolik2022EquivariantQC} which this project builds upon. Section \ref{sect:size-invariant-theory} then outlines the \textit{Size-Invariant Properties} of the EQC at Depth 1, showing that good QRL policies are found within a constrained search space of parameters. Section \ref{sect:size-invariant-method} presents the \emph{Size-Invariant Grid-Search} (SIGS) procedure. Finally, Section \ref{sect:size-invariant-exp} empirically validates our theorems through various simulation results.
\section{Background and Existing Methodology} \label{sect:background} \label{sect:background-methodology}
In this section, we review the relevant background and EQC-based methodology introduced in \cite{Skolik2022EquivariantQC}, providing the necessary context for the contributions presented in the subsequent sections.

\subsection{Markov Decision Process for TSP} \label{sect:background-methodology-mdp} The formulation of TSP as a Markov Decision Process (MDP) used in Skolik et al. (2023) \cite{Skolik2022EquivariantQC} is described here. We follow the same MDP in this work. 

\noindent Given a TSP size $|V|$, with $|E| = \frac{|V|(|V|-1)}{2}$ edges, each component of the MDP is defined as follows: 
\begin{enumerate}
    \item \textit{State Space}: Each (Reinforcement Learning) state is defined as the tuple $(s,d, e,t)$. $s \in \{0, \pi\}^{|V|}$ denotes a one-hot encoding for which nodes have been visited in a graph. The element $s_i$ is set to $\pi$ for unexplored nodes, and $0$ for explored nodes. $d \in \mathbb{R}^{|E|}$ denotes a vector of pairwise Euclidean Distance, and $e$ is a vector of scaled edge weights $e = \tan^{-1}(d) \in \mathbb{R}^{|E|}$, following the reference implementation in \cite{Skolik2022EquivariantQC}. Note that $f(x) = \tan^{-1}(x)$ is a monotonically increasing function in edge distance $x$. $t \in [1,|V|]$ denotes the current node where the agent resides.
    \item \textit{Action Space}: Each action $a$ is a node index (ie. $1\leq a \leq |V|$, $s_a =\pi$). We restrict possible values of $a$ to nodes that have not been explored. The environment then appends this node to the end of the tour.
    \item \textit{Reward}: At each step, the reward is the negative increase in the tour length, measured by Euclidean Distance. Suppose the tour length at step $t$ is $L_t$ and the tour length at step $t+1$ is $L_{t+1}$. Then the reward to the agent is $L_t - L_{t+1}$. We always assume that the last node of the partial tour connects to the first node, and factor in this edge weight accordingly.
    \item \textit{Transition Probabilities}: In a natural TSP formulation as an MDP, taking an action logically results in a deterministic next state: the node $a$ becomes explored, and $s_a \leftarrow  0$ with probability 1.  
    \item \textit{Episodes and Steps}: An \textit{episode} represents a complete TSP tour. The \textit{initial} step $t=0$ is the initial state the RL agent is in, denoted by $s_0$. Without loss of generality, our tours begin with the node at index $1$ (assuming one-indexed), where only node 1 is visited. The terminal step $T$ is the step $|V|-1$, which is when all nodes are visited, and $t$ being some permutation of the $|V|$ vertices. 
\end{enumerate}

\subsection{Parameterized Quantum Circuit} \label{sect:background-methodology-pqc} The \textit{Equivariant Quantum Circuit} (EQC) is an example of an PQC. A diagram of the EQC is provided in Figure \ref{fig:eqc-diagram} for a $n=4$ qubit system. In the formulation of the EQC, there is a one-to-one mapping between qubits and TSP nodes, which means an $n$ qubit system is only able to handle TSP instances of size $n$. There are two kinds of parameterized gates in the EQC as shown: a $ZZ$-gate and a $R_x$ gate. Note that $\gamma$ is shared across all $ZZ$-gates, while $\beta$ is shared across all $R_x$ gates. Measurements are recorded with the $ZZ = Z\bigotimes Z$ observable. The analytical expression for the output of the EQC in the basis of the observable $\langle Z_tZ_a\rangle$, where $t$ is the last node of the present tour and $a$ is an unexplored next node, is given in \cite{Skolik2022EquivariantQC}:
\begin{equation}
    \langle Z_tZ_a ; (\gamma,\beta)\rangle = \sin (\pi\beta)\sin(\gamma e_{ta})\prod_{k=1, k\neq t}^n \cos(\gamma e_{ak}).
    \label{eq:eqc-analytical}
\end{equation}
\begin{figure}[!t]
    \centering
    \input{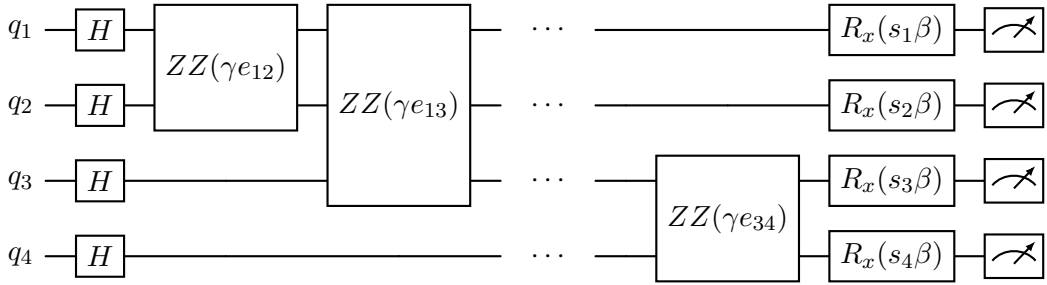}
    \caption{Equivariant Quantum Circuit (EQC) (Depth 1). $\gamma$ and $\beta$ are trainable parameters. Each qubit $q_i$ corresponds to the node $i$. So an $n=4$ qubit system can serve as a PQC for TSP-4 instances. The state $s_i$ is equal to $\pi$ if node $i$ is yet to be visited in the tour, otherwise $s_i = 0$.}
    \label{fig:eqc-diagram}
\end{figure}

\subsection{Quantum Deep Q Learning} \label{sect:background-methodology-quantum-dqn} In the EQC-based QRL framework, the RL state $(s,d,e,t)$ modifies the EQC's rotation parameters, where $s_i$ encodes the $R_x$-rotation angle for qubit $i$, and $e_{ij}$ is the $R_z$-rotation angle on the two-qubit $ZZ$-gate with control qubit $i$, target qubit $j$. 

\noindent The Q-value of moving to an unexplored node $a$ from the present state $(s,d,e,t)$ is given by \ref{eq:eqc-qval} (see equation 17 of \cite{Skolik2022EquivariantQC}):
\begin{equation}
    \begin{aligned}
        Q((s,d,e,t),a ;(\gamma, \beta)) &= d_{ta} \langle Z_tZ_a; (\gamma,\beta)\rangle\\
        &= d_{ta}\sin (\pi\beta)\sin(\gamma e_{ta})\prod_{k=1, k\neq t,a}^n \cos(\gamma e_{ak})
    \end{aligned}
    \label{eq:eqc-qval}
\end{equation}
where $\gamma$, $\beta$ denote the parameters of the \textit{behaviour} policy. The parameters of the behaviour policy are optimized by minimizing the temporal difference loss \cite{mnih2015dqn}.

\noindent Under the Q-Learning framework, the agent's policy at each step is the arg-max over actions:
\begin{equation}
    \pi((s,d,e,t);(\gamma,\beta))= \argmax_{a \cap \{i: s_i = \pi\}} Q((s,d,e,t), a;(\gamma,\beta))
\end{equation}

\subsection{Policy Evaluation Metrics} \label{sect:background-methodology-metrics} 
The authors in \cite{Skolik2022EquivariantQC} used a dataset of TSP instances of size 5, 10 and 20, with node locations uniformly generated in a unit square. We adopted a similar generation procedure for node locations for our simulations, with more specific details discussed in Supplementary Information \bref{sect:experimental-config}{B}.

We can evaluate a \textit{heuristic} for TSP by computing empirical averages of the \textit{optimality gap} on TSP instances in a test set. For one dataset, we evaluate the heuristic using the \textit{Mean Optimality Gap} (Mean Gap) and \textit{Worst Optimality Gap} (Worst Gap) (see Equation \ref{eq:mean-and-worst-gap}). In the below equations, let $x$ be an instance in a dataset $\cal D$, and $\text{QRL}(x)$  be the length of tour produced by the QRL algorithm for instance $x$, and $\text{OPT}(x)$ be the optimal tour length determined for instance $x$. 
\begin{equation}
    \text{Gap}(x) = \frac{\text{QRL}(x)}{\text{OPT}(x)}
    \label{eq:gap}
\end{equation}
\begin{equation}
    \text{Mean Gap} = \frac{1}{|{\cal D}|}\sum_{x \in {\cal D}}\text{Gap}(x), \quad \text{Worst Gap} =  \max_{x \in {\cal D}}\text{Gap}(x)
    \label{eq:mean-and-worst-gap}
\end{equation}
\section{Size-Invariant Properties of the Depth 1 EQC}\label{sect:size-invariant-theory}

In this section, we present Theorems \ref{thm:good-beta}, \ref{thm:gamma-greedy} and Proposition \ref{prop:gamma-range}, before supporting them with empirical results in Section \ref{sect:size-invariant-exp}. We use these theorems to show that the characteristics of the optimization landscape when using the EQC across all TSP sizes at depth 1 are similar and restricted to a small search space. 
Throughout this section, we use the formulation of the MDP in \cite{Skolik2022EquivariantQC}, outlined in Section \ref{sect:background-methodology-mdp}.

\begin{lemma} \label{lem:monotonic-edgeweights}
    \textit{(Relative Ordering of Q-values)} For a fixed, unexplored node $a$ and a set of candidate last nodes $T$, the Q-value of visiting node $a$ from node $t \in T$ is proportional to $d_{ta} \tan (\gamma e_{ta})$. 
\end{lemma}
\begin{proof}
    Starting from the form in Equation \ref{eq:eqc-qval}, we can show that:
    \begin{align}
        Q((s,d,e,t),a;(\gamma,\beta)) &= d_{ta}\sin(\pi\beta) \sin(\gamma e_{ta})\frac{\prod_{k=1, k\neq a}^{n} \cos(\gamma e_{ak})}{\cos(\gamma e_{ta})} \\
        &=d_{ta}\sin(\pi\beta)\tan(\gamma e_{ta})\prod_{k=1, k\neq a}^{n}\cos(\gamma e_{ak})\\
        &\propto d_{ta} \tan (\gamma e_{ta}) \ \text{for fixed node } a.
    \end{align}
\end{proof}

\begin{theorem} \label{thm:good-beta}
    \textit{(Optimality Condition for $\beta$)} When $\gamma$ satisfies $0 < \gamma e_{ij} < \pi/2$ for all edges $e_{ij}$, due to Lemma \ref{lem:monotonic-edgeweights}, good QRL policies are found only in regions where $\sin(\pi\beta)<0$.
\end{theorem}
\begin{proof}
    Our assumption on $\gamma$ implies that $\tan(\gamma e_{ta})>0$ (in Lemma \ref{lem:monotonic-edgeweights}).
    
    Suppose we have a set of candidate last nodes $t_1, t_2,t_3\in T$ and an unexplored node $a$ such that $e_{t_1a} < e_{t_2}a < e_{t_3}a$. Refer to Figure \ref{fig:vis-thm2} for a visualization. Also, under the Q-Learning framework, the RL policy select the maximum-valued action at any state: $\pi(s) = \max_{a} Q(s,a)$ \cite{sutton-barto-2018}. 
    \begin{itemize}
        \item When $\sin(\pi\beta) > 0$, by Lemma \ref{lem:monotonic-edgeweights}, 
        \begin{equation}
            Q((s, d, e, t_1),a; \cdot) < Q((s,d, e,t_2), a; \cdot) < Q((s,d,e,t_3), a ; \cdot).
        \end{equation} 
        \noindent This means that for a fixed node $a$, \textit{larger} edge weights to node $a$ will always have strictly \textit{larger} Q-values than smaller edge weights to node $a$, for any value of $\gamma$. 
        \item When $\sin(\pi\beta) < 0$, by Lemma \ref{lem:monotonic-edgeweights}, 
        \begin{equation}
            Q((s,d,e,t_1),a; \cdot) > Q((s,d,e,t_2),a; \cdot) > Q((s,d, e,t_3), a;\cdot).
        \end{equation}
        \noindent This means that for a fixed node $a$, \textit{larger} edge weights to node $a$ will always have strictly \textit{smaller} Q-values than smaller edge weights to node $a$, for any value of $\gamma$. 
    \end{itemize}
    When $\sin(\pi\beta) > 0$, when the agent selects the maximum valued action in terms of Q-value, it would select an edge from the largest few edge weights, instead of the smaller few edge weights. This results in the agent's expected discounted return being minimized. This is a contradiction to the goal of the agent's MDP.
    \begin{figure}[t!]
        \centering
        \includegraphics[width=0.7\linewidth]{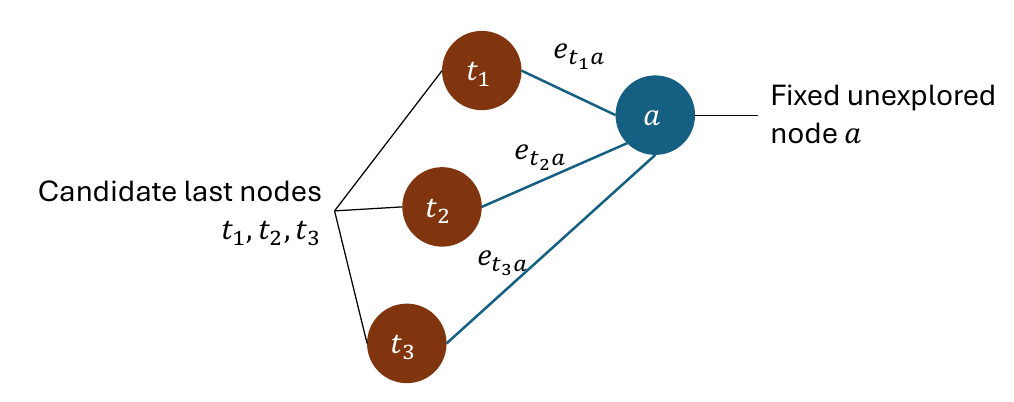}
        \caption{Visualization of $e_{t_1a}<e_{t_2a} < e_{t_3a}$ (Theorem 2)}
        \label{fig:vis-thm2}
    \end{figure}
\end{proof}


\begin{definition} \label{def:gamma-range}
    \textit{(Range for $\gamma$)} In Theorems \ref{thm:good-beta} and \ref{thm:gamma-greedy}, we assumed that $0 <\gamma e_{ij} < \pi/2$ for all edges $i,j$. This means that $\gamma$ must satisfy
    \begin{equation}
        \gamma \max_{(i,j)\in E} e_{ij} < \pi/2 \Leftrightarrow \gamma < \frac{\pi/2}{\max_{(i,j)\in E} e_{ij}}
    \end{equation}
    where $E$ denotes all scaled edge weights in the instance.

    \noindent From this point on, we define
    \begin{equation}
        \gamma_{\max} \coloneq \frac{\pi/2}{\max_{(i,j)\in E} e_{ij}}
    \end{equation}
\end{definition}
\begin{theorem} \label{thm:gamma-greedy}
    \textit{(Interpretation for $\gamma$)} When $\beta$ satisfies $\sin(\pi\beta)<0$  and $\gamma$ satisfies $ 0 < \gamma e_{ij} < \pi/2$ for all edges $e_{ij}$, the value of $\gamma$ controls how strongly the agent favors the nearest neighbor as opposed to further nodes.  
\end{theorem}
\begin{proof}
    When $\sin(\pi\beta)$ < 0, all Q-values are negative, so the maximal Q-value is the smallest absolute value of $\langle Z_tZ_a\rangle$. 
    \begin{itemize}
        \item When $\gamma \rightarrow 0$, then $\gamma e_{jk} \rightarrow 0$ and $\prod_{k=1,k\neq t}^n \cos(e_{ka}) \rightarrow 1$. As such, $Q((s,d,e,t),a;\cdot) \rightarrow d_{ta} \sin(\pi\beta)\sin(\gamma e_{ta})$. The function $f(e_{ta}) = d_{ta} \sin(\gamma e_{ta}) = \tan(e_{ta}) \sin(\gamma e_{ta})$ is an increasing function in $e_{ta}$ over the domain of $0 < e_{ta} < \tan^{-1} {\sqrt{2}}$. Therefore, the agent begins to award the nearest node with the highest Q-value most of the time.
        \item When $\gamma \rightarrow \frac{\pi/2}{\max_{(i,j)\in E} e_{ij}}$, then $\prod_{k=1, k\neq t}^n \cos(\gamma e_{ta}) \rightarrow 0$. At this stage, both a small edge weight $e_{ta}$ or a large distance from other nodes will maximize the Q-value. This results in the agent selecting from sometimes the nearest nodes and sometimes from the furthest few nodes. Note that the agent still selects the nearest nodes more often, due to the inherent ordering of the nodes discussed in Theorem \ref{thm:good-beta}. 
    \end{itemize}
\end{proof}
The remainder of this section intuitively explains our design choice to restrict optimization to $\gamma \in (0,\gamma_{\max})$, under the assumption that edges are drawn from a continuous distribution. We show that optimization is fragile in this domain, as we observe that the $\arg \max$ of the Q-values changes frequently due to sign flips in the cosine terms of the Q-value.

\begin{assumption} \label{assumption:non-resonance}
    (Continuous Distribution of Edge Weights) For each instance edge weights $(d,e)$, $d_{ij}$ between any distinct pairs of nodes $i, j$ where $i\neq j$ are sampled from a continuous distribution. With high probability, for all unordered node pairs $1\leq i,j \leq |V|, i\neq j$,
    \begin{enumerate}
        \item All $d_{ij}$ are distinct;
        \item Hence $e_{ij} = \tan^{-1}(d_{ij})$ (see Section \ref{sect:background-methodology-mdp}) are distinct;
        \item Consequently, solutions $\gamma$ to $\cos(\gamma e_{ij}) = 0$ are distinct.
    \end{enumerate}
\end{assumption}
\begin{definition} \label{def:cosine-zero-sets}
    (\textit{Cosine-Zero Sets}) Fix the state $(s,d,e,t)$. For an interval $I\subset (0, \infty)$ (we use $I = [\gamma_{\max} , 2\pi]$), define 
    \begin{equation}
        \mathcal{Z}_{(s,d,e,t)}(I)\coloneq \{\gamma \in I; \exists\  a\neq t, b \neq t, a\neq b; \cos(\gamma e_{ab}) = 0\}.
    \end{equation}
\end{definition}
\begin{lemma} \label{lem:safety-of-gamma-range}
    (Safety of the domain $\gamma \in (0, \gamma_{\max})$ ) . Let $\gamma_{\max} = \frac{\pi/2}{\max_{i<j}e_{ij}}$, then $\gamma \notin {\cal Z}_{(s,d,e,t)}((0,\infty))$ for every state $(s,d,e,t)$.  
\end{lemma}
\begin{proof}
    For all edges $(i,j)$, $0 < \gamma e_{ij} < \pi/2$, so $\cos(e_{ij} \gamma) > 0$ for all $e_{ij}$.
\end{proof}
\begin{definition} \label{def:policy-critical}
   (Policy Critical) Let $\delta >0$. We define a policy to be critical at $\gamma_0$ when $\gamma_0$ satisfies: 
   \begin{equation}
       \pi((s,d,e,t);(\gamma_0-\delta,\beta)) \neq \pi((s,d,e,t);(\gamma_0 +\delta, \beta))
   \end{equation}
\end{definition}
\begin{proposition} \label{prop:gamma-range}
    (Instability for $\gamma \in [\gamma_{\max}, 2\pi]$) Assume that edge weights are drawn from a continuous distribution (see Assumption \ref{assumption:non-resonance}). Fix a state $(s,d,e,t)$ and let $\gamma^\# \in {\cal Z}_{(s,d,e,t)}$. Then at $\gamma = \gamma^\#$, the relative ordering of two actions' Q-values from $t$ inverts.

    \noindent Let $\delta$ denote a small number, $\delta > 0$. If at $\gamma^\#-\delta$ or at $\gamma^\# +\delta$, either of these two actions has the top-1 Q-value, this inversion is {policy critical} at $\gamma^\#$ (see Definition \ref{def:policy-critical}).  
\end{proposition}
\noindent \textit{Proof.} See Supplementary Information \bref{appendix:thm-gamma-range-proof}{A}.

\noindent \textit{Remark. } If $\{e_{ab}\}$ are drawn from a continuous distribution the cosine zeros are dispersed throughout $[\gamma_{\max}, 2\pi]$, it is difficult to provide a closed-form count of policy-critical occurrences. Instead, we report an empirical Policy-Critical rate (the fraction of grid points where the greedy action changes). When $\delta = 0.01$, this Policy-Critical rate rises from 0.05\% in the safe $\gamma$ region to about 4\% beyond the safe $\gamma$ region (see Figure \ref{fig:prop-gamma-range-flip-rate}) - indicating frequent policy changes in this regime. This resulting brittleness motivates our design choice to restrict the search for SIGS to $\gamma\in (0, \gamma_{\max})$ in Section \ref{sect:size-invariant-method}.

To summarize, Theorem \ref{thm:good-beta} suggest that a good RL policy can only be found when $\sin(\pi\beta) < 0$. At Depth 1, the $\sin(\pi\beta)$ term merely acts as a multiplier of the Q-value, and does not affect the policy. Theorem \ref{thm:gamma-greedy} show that under the condition that $\sin(\pi\beta) <0$ and $0 < \gamma e_{ij} < \pi/2$, the value of $\gamma$ controls the greediness of the agent's policy. Theorem \ref{thm:good-beta} and \ref{thm:gamma-greedy} operates under the assumption that $0 < \gamma e_{ij} < \pi/2$ for all nodes $i,j$. In Proposition \ref{prop:gamma-range}, we discuss why training in the regime where $\gamma \in [\gamma_{\max}, 2\pi]$ gives rise to brittle policies, where optimization is unstable as individual cosine terms approach zero or turn negative. This gives rise to policies where the policy function to output different actions despite small perturbations in the value of $\gamma$.

\section{Size-Invariant Grid Search} \label{sect:size-invariant-method}
Performing quantum simulations on classical computers is challenging due to exponential time and memory required, and researchers in QRL for NCO have been limited to small problem sizes (TSP $\leq$ 20 nodes \cite{Skolik2022EquivariantQC}, $\leq$ 15 QUBO decision variables \cite{Kruse_2024_HEA}). To overcome this, we develop a simple \textit{Size-Invariant Grid Search} (SIGS) that restricts the search space of $\gamma$ and $\beta$ and employs a grid-search strategy. This approach exploits the size-invariant properties established in Sections \ref{sect:size-invariant-theory} and enables us to scale simulations well beyond tractable limits of reinforcement learning or quantum circuit based training.
We want to emphasize that, since the nature of our method is a grid search, and the optimal region of the search space is restricted by the size-invariant properties, it inherently avoids the \emph{barren plateaus} phenomenon that is well-known in variational quantum algorithms~\cite{bp-qnn,adjoint-bp,lie-bp}.

Simulation Configurations (Dataset generation procedure, QRL Configurations and compute resources) are outlined in Supplementary Information \bref{sect:experimental-config}{B}. We first formalize the training methodology and simulation configurations for SIGS \ref{sect:size-invariant-method-train-method}). We then (i) compare the performance of SIGS against the RL-based training procedure from Supplementary Information \bref{sect:experimental-config-qrl}{B.2} (Section \ref{sect:size-invariant-method-exp-comp-rl}), and (ii) apply it to study Depth-1 EQC performance on TSP instances of up to 350 nodes (Section \ref{sect:size-invariant-method-exp-scale-tsp350}).

\subsection{SIGS Training Methodology} \label{sect:size-invariant-method-train-method}
SIGS consists of selecting any $\beta$ satisfying $\sin(\pi \beta)<0$ and fixing $\beta$ as a constant (see Theorem \ref{thm:good-beta}) (we use $\beta = 1.01$ for all simulations). We then evaluate the model across 16 candidate $\gamma$ values $\gamma \in [0.1, 1.7)$ (see Theorem \ref{thm:gamma-greedy}, Proposition \ref{prop:gamma-range}), with step-size $\Delta\gamma = 0.1$. Note that no training dataset or gradient-based optimization is required - only a validation set to select $\gamma$ and a test set to obtain an unbiased estimate of performance after selection of $\gamma$.

 When we used a finer grid for $\gamma$ ($\Delta\gamma = 0.0001$), we found the difference in mean optimality gap between a $\Delta\gamma =0.1$ gridsearch and $\Delta\gamma = 0.0001$ gridsearch to be less than $0.02$ (1.65\% relative error) across TSP-10, TSP-100 and TSP-350. Note that this discretization error can be reduced by optionally doing a coarse-to-fine gridsearch (by first finding the best $\gamma$ using step-size $\Delta\gamma =0.1$, and then searching in the domain $\gamma \in [\gamma-\Delta\gamma, \gamma +\Delta\gamma]$ using a finer grid of $\gamma$. We leave this as future work. See Supplementary Information \bref{appendix:size-invariant-method-discretization-error}{F} for more information.

\subsection{Simulation 1: Comparison with RL} \label{sect:size-invariant-method-exp-comp-rl}
\begin{figure}[t]
    \centering
    \begin{subfigure}[b]{0.48\linewidth}
        \centering
        \includegraphics[width=\linewidth]{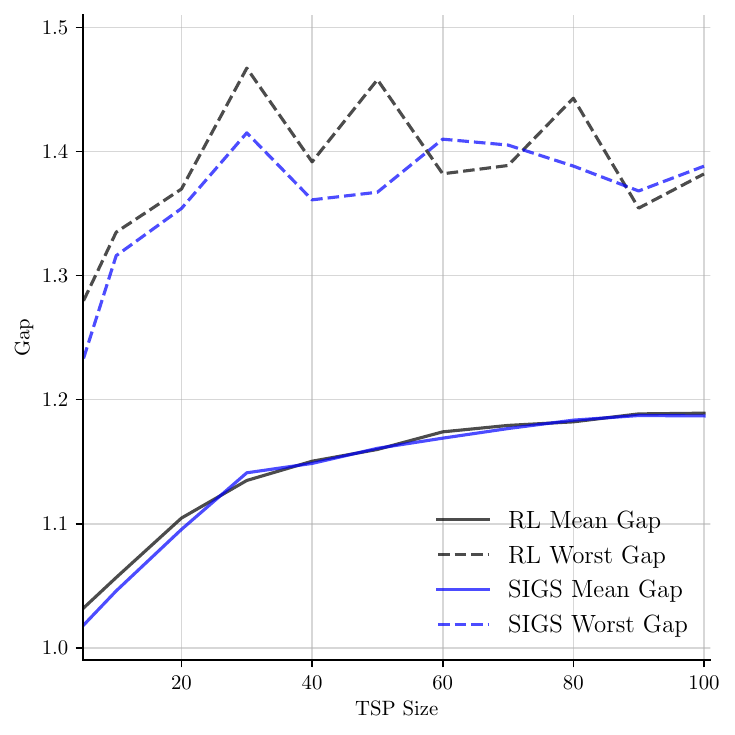}
        \caption{Mean and Worst Test Gap}
        \label{fig:rl-vs-brute-mean-gap}
    \end{subfigure}
    \begin{subfigure}[b]{0.48\linewidth}
        \centering
        \includegraphics[width=\linewidth]{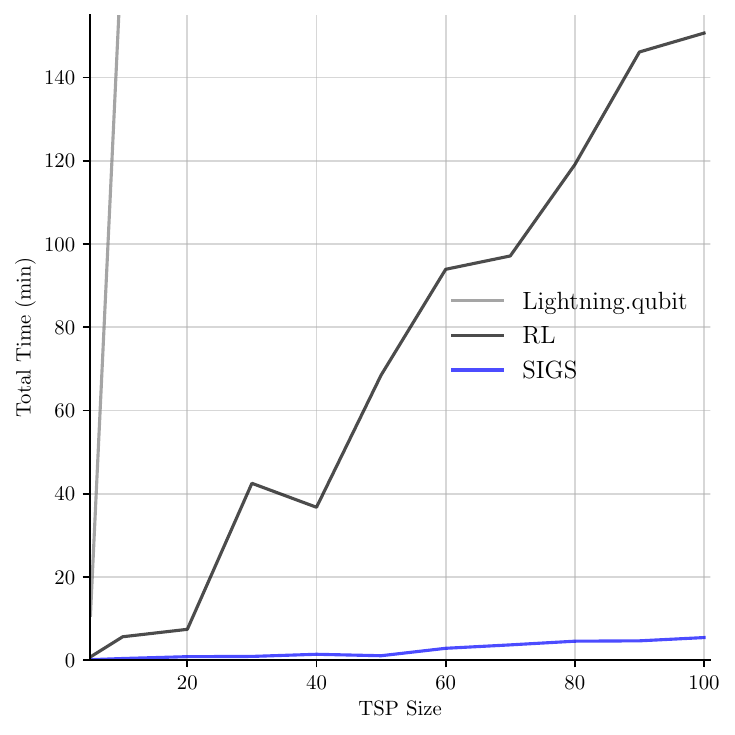}
        \caption{Total Time (minutes) used for training and testing}
        \label{fig:rl-vs-brute-wall-time}
    \end{subfigure}
    \caption{Comparison of Mean Gaps, Worst Gaps on the Test Set and Total Training Time between RL and SIGS. Details of \textbf{SIGS} is discussed in Section \ref{sect:size-invariant-method-train-method}. The methodology for "RL" is discussed in Supplementary Information \bref{sect:experimental-config-qrl}{B.2}. Tabulated results can be found in Supplementary Information \bref{appendix:size-invariant-tm-full-comp}{D}.}
    \label{fig:rl-vs-SIGS}
\end{figure}
We compare the performance of SIGS with RL-based training at Depth 1 (Figure \ref{fig:rl-vs-SIGS}). Across 11 TSP sizes (5-100), the mean difference in test gap is only $\bar \Delta = 0.005$ with a maximum of $\Delta_{\max} = 0.014$, demonstrating that SIGS achieves performance nearly identical to RL. Crucially, SIGS reduces training and evaluation time significantly on larger TSP sizes - for example, from 151 minutes (RL) to under 6 minutes (SIGS) on TSP-100, a 96.4\% reduction. This efficiency makes SIGS a practical tool for analyzing QRL performance on larger problem instances. A more detailed comparison, consisting of a total time breakdown and exact optimality gaps for Figures \ref{fig:rl-vs-brute-mean-gap} and \ref{fig:rl-vs-brute-wall-time}, can be found in Supplementary Information \bref{appendix:size-invariant-tm-full-comp}{D}.

Due to slight implementation and dataset differences, we find that the Mean Optimality Gaps of our RL implementation have small replication errors with that of Skolik et al. (2023) (see Figure 5b of \cite{Skolik2022EquivariantQC}). Specifically in terms of Mean Optimality Gap, TSP-5 has a replication error of 0.012 (reported 1.02, ours 1.032); TSP-10 has a replication error of 0.007 (reported 1.05, ours is 1.057); and TSP-20 has a replication error of 0.015 (reported 1.12, ours 1.105). 

\subsection{Simulation 2: Scale to TSP-350} \label{sect:size-invariant-method-exp-scale-tsp350}
\begin{figure}[t]
    \begin{subfigure}[b]{0.32\textwidth}
        \centering
        \includegraphics[width=\linewidth]{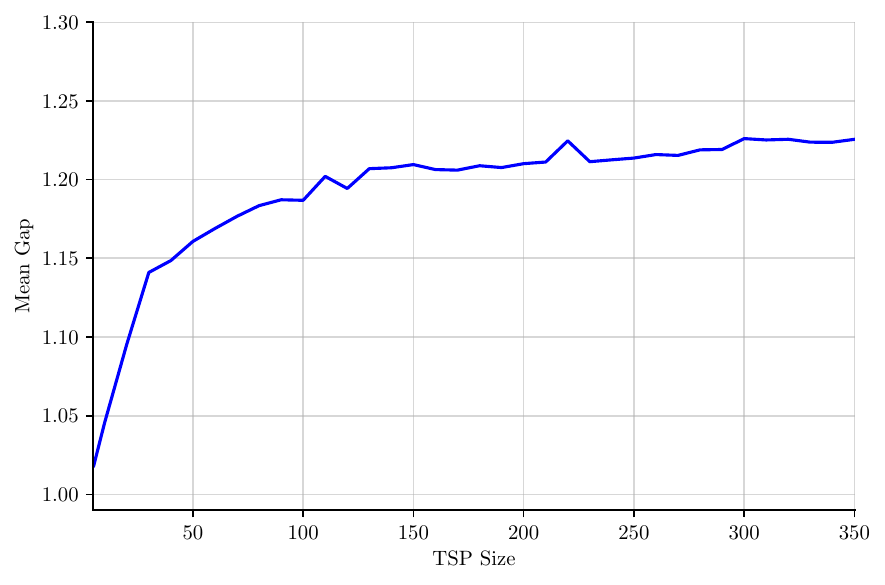}
        \caption{Mean Gap}
        \label{fig:gap-to-tsp350}
    \end{subfigure}
    \begin{subfigure}[b]{0.32\textwidth}
        \centering
        \includegraphics[width=\linewidth]{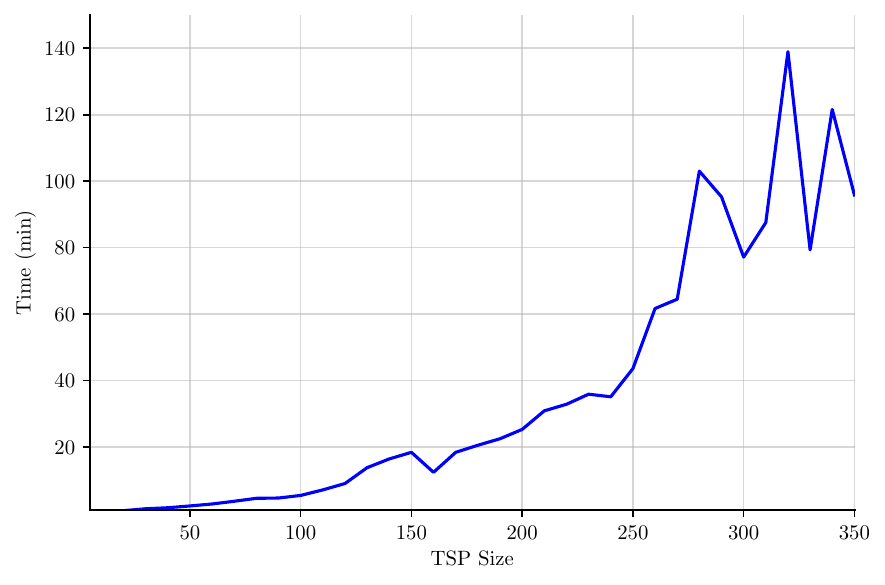}
        \caption{Total Time (min)}
        \label{fig:time-to-tsp350}
    \end{subfigure}
    \begin{subfigure}[b]{0.32\textwidth}
        \centering
        \includegraphics[width=\linewidth]{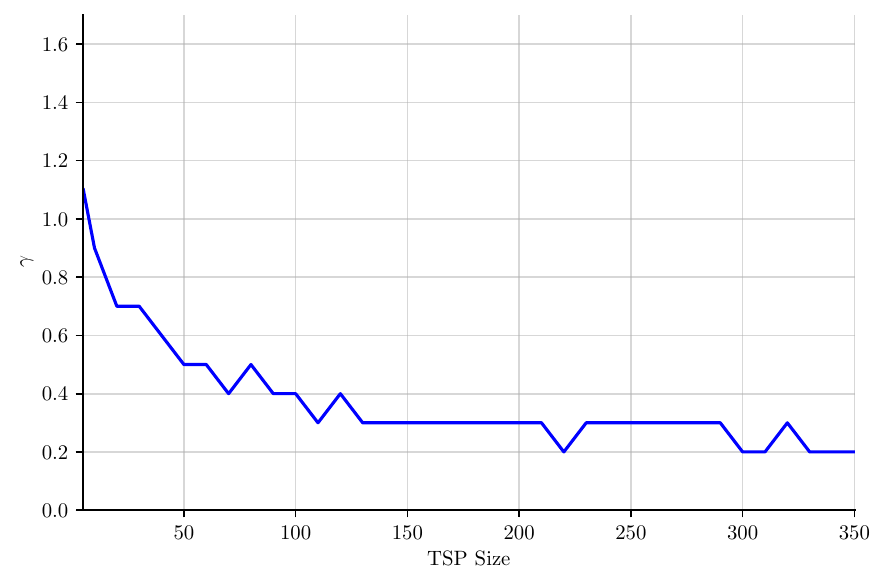}
        \caption{Value of $\gamma^*_{SIGS}$}
        \label{fig:gamma-to-tsp350}
    \end{subfigure}
    \caption{Mean Test Gaps, Total time and value of $\gamma^*_{SIGS}$ when scaling up to TSP-350, using $\Delta\gamma = 0.1$. Note that the time taken consists of the time taken to select parameter $\gamma$ and its evaluation on the test set}
    \label{fig:scale-to-tsp350}
\end{figure}
One of the advantages of the SIGS is that it enables simulations on problem sizes far beyond the limits of QRL-based training. RL optimization becomes intractable even at TSP-20, requiring days of training and large memory resources (see Supplementary Information \bref{appendix:size-invariant-tm-full-comp}{D}). In contrast, SIGS requires only a grid-search over $\gamma$ with a fixed $\beta$.

Figure \ref{fig:gap-to-tsp350} illustrates the mean test gaps obtained by SIGS across TSP instances up to 350 nodes, using $\Delta\gamma=0.1$ as justified in Supplementary Information \bref{appendix:size-invariant-method-discretization-error}{F}. Importantly, this is the first evaluation of a Depth-1 EQC on TSP instances exceeding 20 nodes. These results establish SIGS as a practical tool for benchmarking QRL in regimes that are inaccessible to RL-based training. We significantly reduce runtimes, as shown in Figure \ref{fig:time-to-tsp350}, the longest time in minutes to execute one particular run is 139 minutes (TSP-320). This may be because of other jobs currently executing on the GPU on that particular point in time.

Lastly, Figure \ref{fig:gamma-to-tsp350} highlights an important nuance of the size-invariant properties. Here, let $\gamma^*_{SIGS}$ denote the $\gamma$ selected by the SIGS method for a particular TSP size. Theorems \ref{thm:good-beta} and \ref{thm:gamma-greedy} do not guarantee that the value of $\gamma^*_{SIGS}$ is also constant across TSP size. If that were the case, our training strategy would just be to claim that a particular value of $\gamma$ would be near optimal for all TSP-sizes. Specifically, $\gamma^*_{SIGS}$ is decreasing as TSP size increases. We conjecture that $\gamma^*$ needs to decrease in order to preserve the tradeoff between making locally optimal and locally suboptimal moves, which is further explained in Conjecture \ref{conjecture:opt-gamma-decreases} below.

\begin{definition} \label{def:subgraph}
    (A subgraph $m$ nodes of an $n$ node TSP instance) Let $(d_n,e_n)$ be a TSP instance of size $n$, where $n>3$, and $d_n,e_n$ contain all edge distances between unordered pairs $1\leq i,j\leq n, i\neq j$. Without loss of generality, define a subgraph of size $m$ ($m< n$) as the induced subgraph on the first $m$ vertices of the TSP instance. Formally
    \begin{equation}
        d_m = \{d_{ij} = ||x_i - x_j||_2\ |\ \forall \ 1\leq i < j \leq m\}, \quad e_m \coloneq \tan^{-1}(d_m),
    \end{equation}
    \noindent where $x_i\in\mathbb{R}^2$ denotes the coordinates of node $i$. 
\end{definition}
\begin{lemma}
Let $d_n, e_n$ denote a TSP instance of size $n$ and $d_m,e_m$ denote the induced subgraph of size $m$ (as per Definition \ref{def:subgraph}). Let $s_n \in \{0,\pi\}^{n}, s_m\in \{0,\pi\}^{m}$ denote the respective one-hot encoding of visited nodes, and $t$ be the current node that an agent resides in that is common to both instances (so $t$ satisfies $1\leq t \leq m$). 
For fixed parameters $\gamma,\beta$,
$ Q((s_n,d_n,e_n,t), a;(\gamma,\beta)) \leq Q((s_m,d_m,e_m,t), a;(\gamma,\beta))$.  \label{lem:reducing-qvals}
\end{lemma}
\begin{proof}
    Without loss of generality, we let the nodes present in $s_m$. 
    Starting from the form of Equation \ref{eq:eqc-qval}, we have
    \begin{align}
        Q((s_n, d_n, e_n, t),a;\cdot) &= {d_n}_{ta}\sin(\pi\beta)\sin(\gamma {e_n}_{ta})\prod_{k=1,k\neq t}^n\cos(\gamma {e_n}_{ak})\\
        &=[{d_m}_{ta} \sin(\pi\beta) \sin(\gamma {e_m}_{ta})\prod_{k=1,k\neq t}^m\cos(\gamma {e_m}_{ak})]\prod_{k=m+1}^n\cos(\gamma {e_{n}}_{ak})\\
        &= Q((s_m, d_m, e_m,t),a; \cdot)\prod_{k=m+1}^n\cos(\gamma {e_{n}}_{ak})\\
        &\leq Q((s_m, d_m, e_m,t),a;\cdot).
    \end{align}
    This concludes the proof that $Q((s_n,d_n,e_n,t),a;\cdot) \leq Q((s_m,d_m,e_m,t), a;\cdot)$.
\end{proof}
Note that the equality only holds when $e_{ai}=0$ for all nodes $i\in [m+1, n]$, meaning every node introduced is placed exactly on top of node $a$ and its probability is effectively zero.

\begin{conjecture}
\textit{(Size Dependence of $\gamma^*_\text{SIGS}$)}
$\gamma^*_{SIGS}$ is monotonically decreasing as TSP-size increases. \label{conjecture:opt-gamma-decreases}
\end{conjecture}
From Theorem \ref{thm:gamma-greedy}, we can interpret $\gamma^*$ as the best possible tradeoff for the agent in terms of making locally optimal moves and locally sub-optimal moves. 

By Lemma \ref{lem:reducing-qvals}, when the number of TSP nodes increases, the magnitude of the Q-value generally decreases. We can view the Q-value as a tradeoff between the two terms $A(\gamma)$ and $B(\gamma)$:
\begin{equation}
    Q((s,d,e,t),a ;(\gamma, \beta)) = \sin(\pi\beta)\underbrace{d_{ta}\sin(\gamma e_{ta})}_{A(\gamma)} \underbrace{\prod_{k=1, k\neq t,a}^{n} \cos(\gamma e_{ak})}_{B(\gamma)}.
\end{equation}
Therefore, when we increase the number of nodes, we are effectively increasing the weight of $B(\gamma)$, thus affecting the best possible tradeoff of making locally optimal and locally suboptimal moves.

We provided a visualization in Supplementary Information \bref{appendix:gamma-constant}{E}, where we show that for the same value of $\gamma, \beta$ and when number of nodes increases, the nearest neighbor selected by the model's policy becomes increasingly non-locally optimal. 

\subsection{Simulation 3: Training beyond a Depth 1 EQC?} \label{sect:discussion}
\begin{table}[t]
\centering
\scriptsize
\caption{Mean test gap and total time (minutes) across TSP sizes (5 to 15) at EQC Depths (1 to 4). Best gaps using RL across Depths 1 to 4 are \underline{underlined}. Best gaps including the comparison with SIGS are in \textbf{bold}.}
\resizebox{\textwidth}{!}{%
\begin{tabular}{c|cc|cc|cc|cc|cc}
\toprule
EQC & \multicolumn{2}{c|}{TSP-5} & \multicolumn{2}{c|}{TSP-8} & 
   \multicolumn{2}{c|}{TSP-10} & \multicolumn{2}{c|}{TSP-13} & 
   \multicolumn{2}{c}{TSP-15} \\
Depth & Gap & Time(m) & Gap & Time(m) & Gap & Time(m) & Gap & Time(m) & Gap & Time(m) \\
\midrule
1 (RL) & \underline{1.032} & 0.296 & \underline{1.048} & 0.469 & 1.056 & 0.861 & \underline{1.069} & 0.755 & \textbf{\underline{1.080}} & 0.951 \\
2 & 1.033 & 18.6 & 1.050 & 91.8 & 1.060 & 303 & 1.072 & 808 & 1.090 & 3170 \\
3 & 1.030 & 50.7 & 1.055 & 251 & \underline{1.053} & 482 & \underline{1.069} & 1350 & 1.083 & 2300 \\
4 & 1.032 & 116 & 1.053 & 205 & 1.065 & 662 & 1.070 & 1190 & 1.089 & 2700 \\
\midrule
1 (SIGS) & \textbf{1.018} & 0.123 & \textbf{1.036} & 0.467 & \textbf{1.046} & 0.401 & \textbf{1.068} & 0.828 & 1.081 & 0.951 \\
\bottomrule
\end{tabular}}
\label{tab:eqc-depths1-to-4}
\end{table}

Table \ref{tab:eqc-depths1-to-4} shows the results of running EQCs at Depth 1 to Depth 4 on small TSP sizes, with training methodology discussed in Supplementary Information \bref{sect:experimental-config-qrl}{B.2}. While in theory, the quality of the agent's tours are supposed to improve with larger depth, we find that increasing depth in most cases only increases the cost of simulation, with a negligible difference in performance. For instance, at TSP 5,8,13 and 15, we find that the EQC at Depths 2 to 4 does not outperform the EQC at Depth 1, despite requiring significantly more time. It is only at TSP-10 where we find a drop in gap of $0.003$ (0.3\% improvement). For TSP 5,8,10 and 13, we find that the SIGS still produces the lowest gap on the test set (see bolded numbers in Table \ref{tab:eqc-depths1-to-4}), where RL using the Depth 1 EQC at TSP-15 slightly outperforms SIGS by a drop in gap of $0.001$. This negligible difference in small TSP sizes was also reflected in the validation performance between Depth 1 and Depth 4 of Skolik et al. (2023) \cite{Skolik2022EquivariantQC}. The mean validation performance of using an EQC at Depth 1 at TSP sizes 5, 10 and 20 were similar to using an EQC at Depth 4 (see Figure 5b and 5d of \cite{Skolik2022EquivariantQC}). We further analyze this phenomenon in Section \ref{sect:size-invariant-exp-deeper-eqcs}, showing that the Q-value against Distance distributions are similar across EQC depths from 1 to 4. 
\section{Discussion} \label{sect:size-invariant-exp}
This section contains empirical verifications of theorems presented in Section \ref{sect:size-invariant-theory}.

\begin{figure}[!t]
    \centering
        \begin{subfigure}[b]{0.47\textwidth}
        \includegraphics[width=\textwidth]{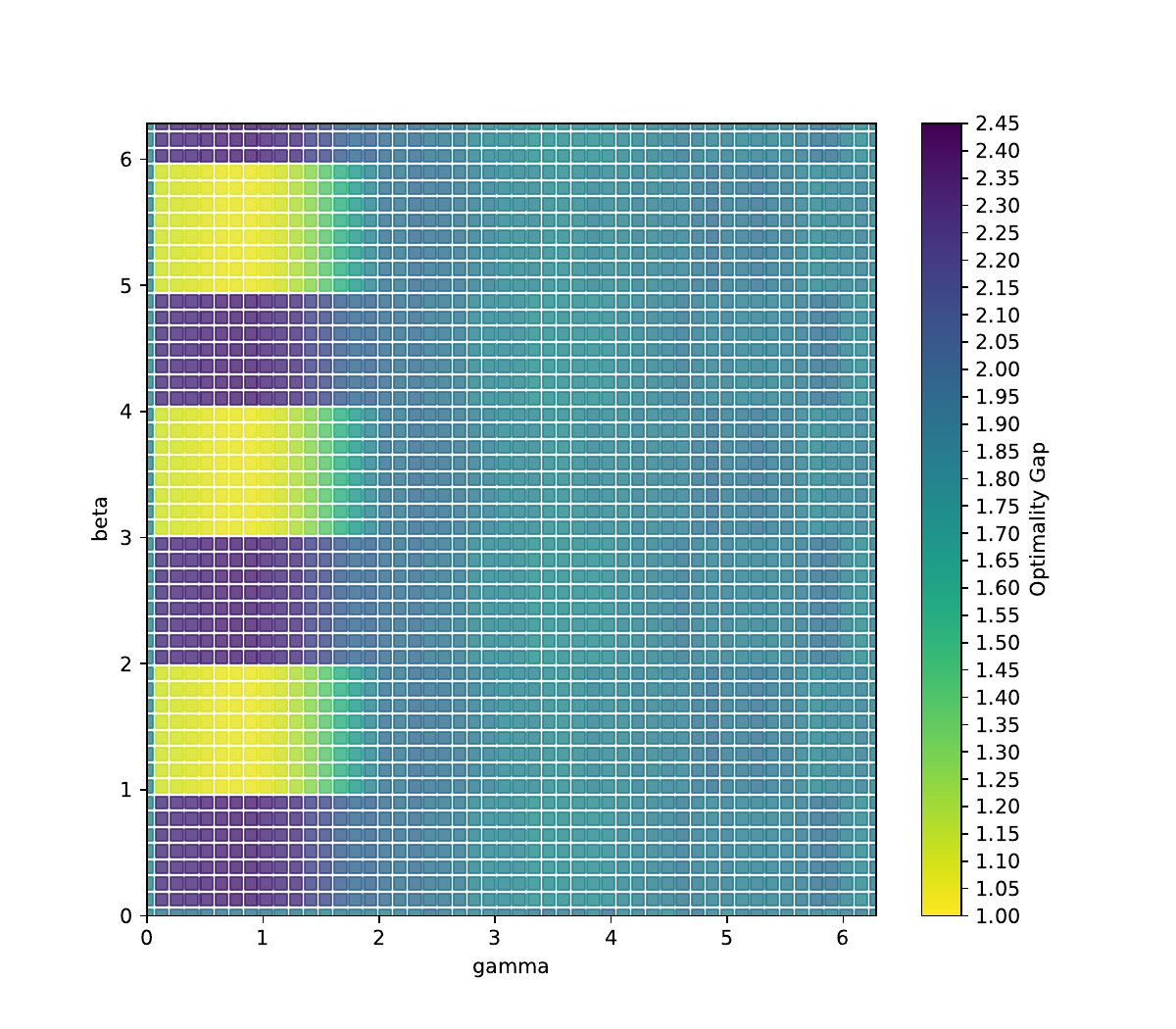}
        \caption{TSP-10}
        \label{fig:grid-search-heatmap-tsp10}
    \end{subfigure}
    \begin{subfigure}[b]{0.47\textwidth}
        \includegraphics[width=\textwidth]{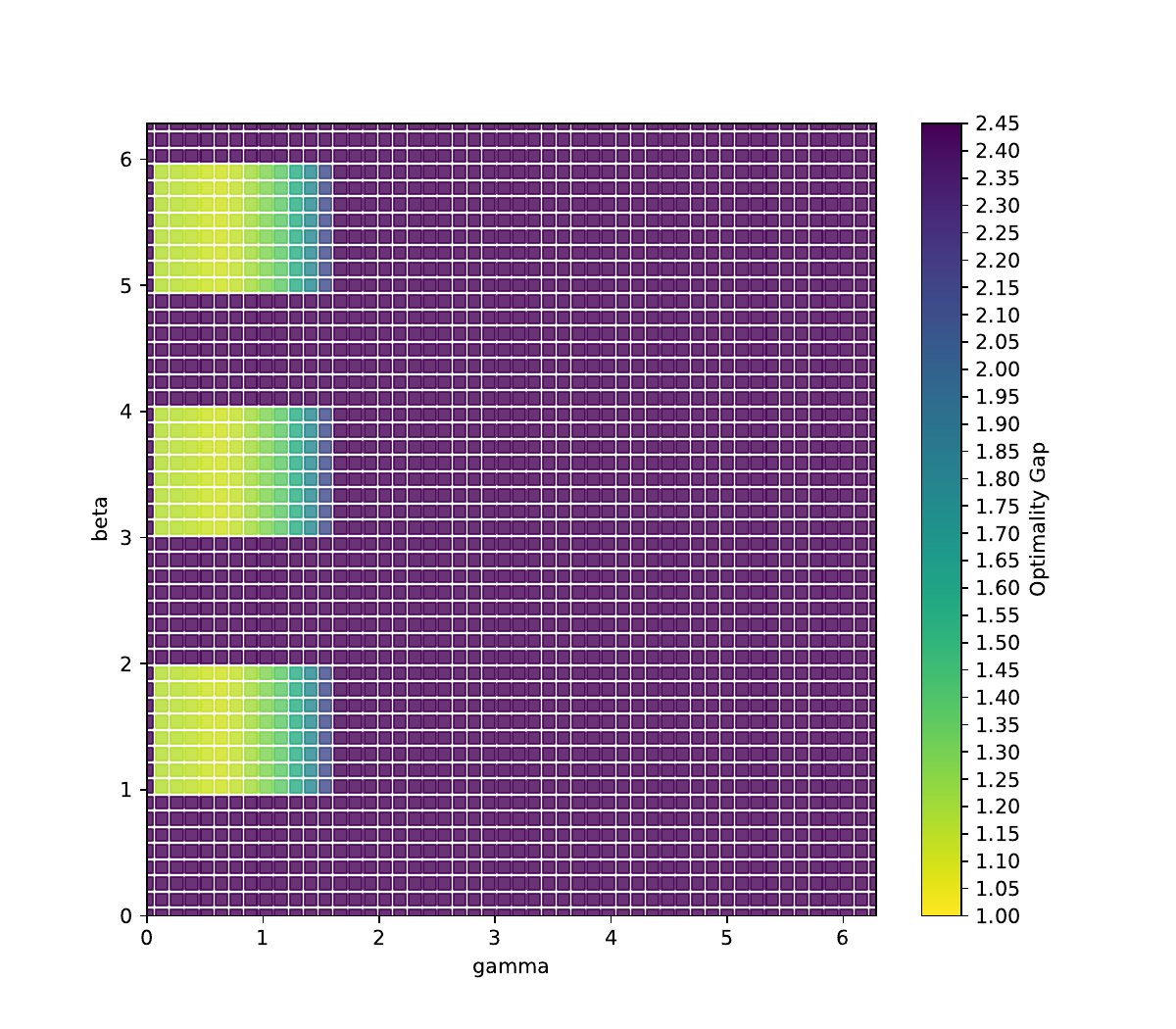}
        \caption{TSP-20}
        \label{fig:grid-search-heatmap-tsp20}
    \end{subfigure}
    \begin{subfigure}[b]{0.45\textwidth}
        \includegraphics[width=\textwidth]{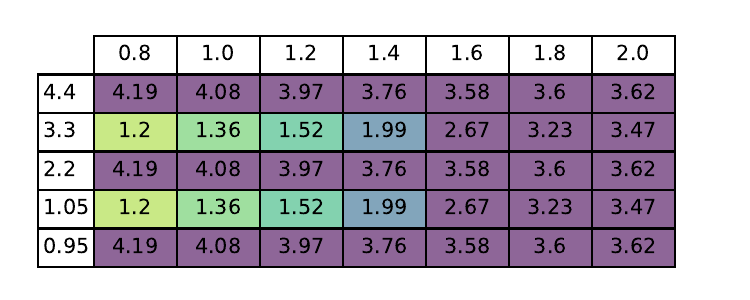}
        \caption{TSP-25}
        \label{fig:grid-search-heatmap-tsp25}
    \end{subfigure}
    \begin{subfigure}[b]{0.45\textwidth}
        \includegraphics[width=\textwidth]{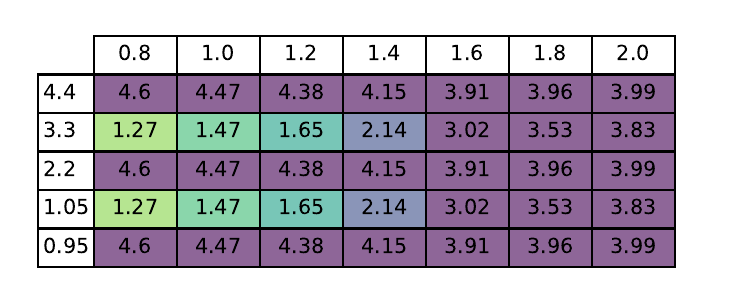}
        \caption{TSP-30}
        \label{fig:grid-search-heatmap-tsp30}
    \end{subfigure}
    \begin{subfigure}[b]{0.45\textwidth}
        \includegraphics[width=\textwidth]{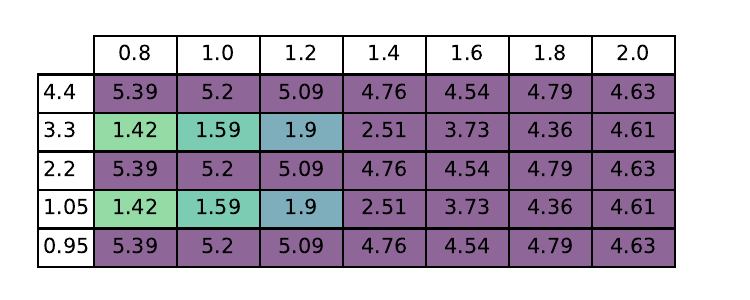}
        \caption{TSP-40}
        \label{fig:grid-search-heatmap-tsp40}
    \end{subfigure}
    \begin{subfigure}[b]{0.45\textwidth}
        \includegraphics[width=\textwidth]{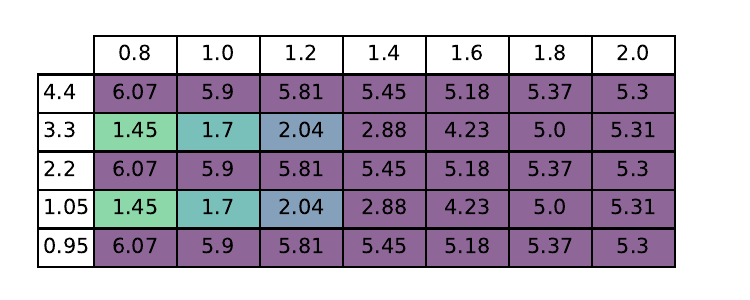}
        \caption{TSP-50}
        \label{fig:grid-search-heatmap-tsp50}
    \end{subfigure}
    \caption{Mean Optimality Gap Heatmap of Depth 1 EQC with respect to $\gamma,\beta$. These landscapes are Mean Gaps (see equation \ref{eq:mean-and-worst-gap}) over 50 instances for TSP-10 (Figure \ref{fig:grid-search-heatmap-tsp10}) and 30 instances for TSP-20 and above (Figure \ref{fig:grid-search-heatmap-tsp20} to \ref{fig:grid-search-heatmap-tsp50}). Smaller Optimality Gaps are color-coded as brighter colors, while larger optimality gaps are encoded with darker colors. TSP 10 (Figure \ref{fig:grid-search-heatmap-tsp10}) and 20 (Figure \ref{fig:grid-search-heatmap-tsp20}) heatmaps span the entire parameter space  $\gamma,\beta \in (0, 2\pi)$. TSP 25, 30, 40, and 50 (Figures \ref{fig:grid-search-heatmap-tsp25}-\ref{fig:grid-search-heatmap-tsp50}) are heatmaps of subsets of the parameter space to verify the size-invariant landscape. $\beta \in \{0.95,1.05, 2.2, 3.3,4.4\}$ is plotted on the vertical axis, while $\gamma \in \{0.8,1.0,\cdots, 2.0\}$ is plotted on the horizontal axis. As these heatmaps was created in the earlier phases of the project, we approximated the optimal tour to be the best of 30 runs of the simulated annealing implementation in \cite{dreo:hal-01341683} for TSP-20 to TSP-50. }
    \label{fig:grid-search-heatmap}
\end{figure}
\begin{figure}[t]
    \centering
    \begin{subfigure}[b]{0.48\textwidth}
        \includegraphics[width=\textwidth]{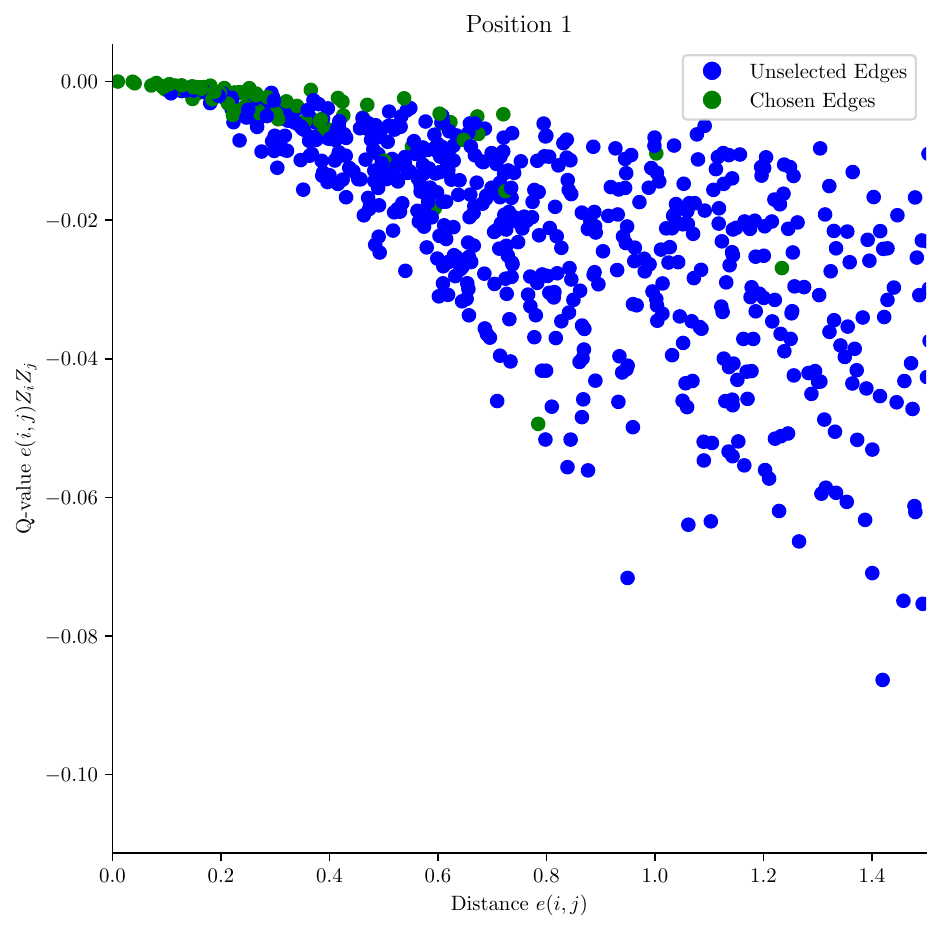}
        \caption{$\sin(\pi\beta)<0$}
        \label{fig:qval-beta-leq0}
    \end{subfigure}
    \begin{subfigure}[b]{0.48\textwidth}
        \includegraphics[width=\textwidth]{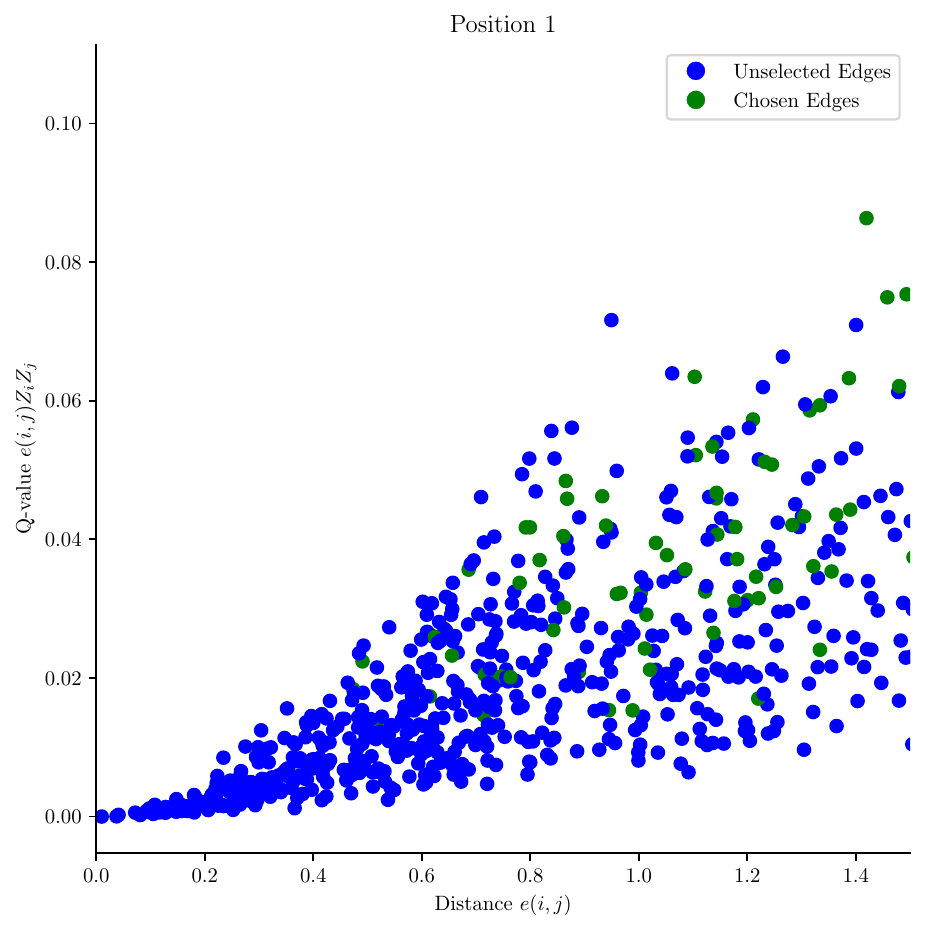}
        \caption{$\sin(\pi\beta)>0$}
        \label{fig:qval-beta-geq0}
    \end{subfigure}
    \caption{Q-value against Distance plots at Position 1 of the Tour (TSP-10). Based on 100 TSP-10 instances with uniformly random node placements in a unit square. Note that distances are scaled so the maximum length of an edge is $\sqrt 2$. We set $\gamma = 1$. Figures \ref{fig:qval-beta-leq0} and \ref{fig:nn-beta-leq0} is when $\sin(\pi\beta) < 0$ ($\beta = 1.1$), and Figures \ref{fig:qval-beta-geq0} and \ref{fig:nn-beta-geq0} is when $\sin(\pi\beta) > 0$ ($\beta = 0.9$). Since all TSP tours are initialized at node 1, Figure \ref{fig:qval-beta-leq0} and \ref{fig:qval-beta-geq0} show the distribution of Q-values when an agent visits an unexplored node $a$ against is Euclidean Distance from node 1. Green datapoints represent selected edges by the agent, while blue datapoints represent unselected edges. }
    \label{fig:beta-comp-qval}
\end{figure}
\begin{figure}[t]
    \centering
    \begin{subfigure}[b]{0.48\textwidth}
        \includegraphics[width=\textwidth]{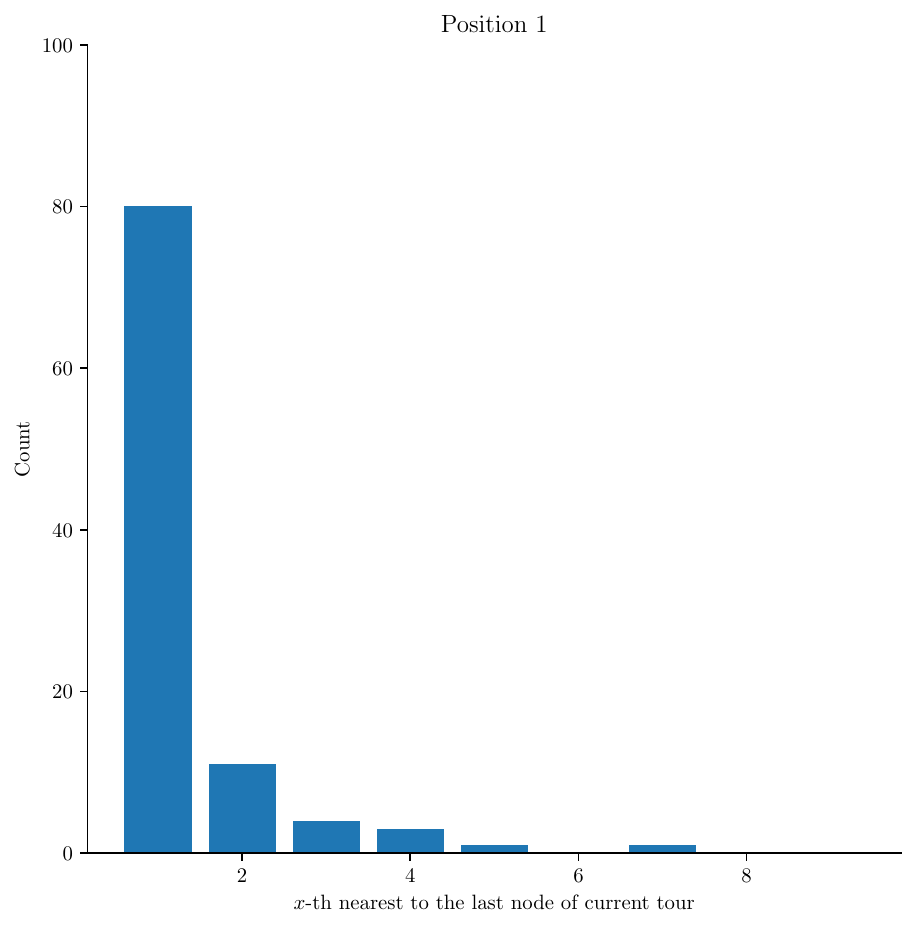}
        \caption{$\beta = 1.1,\sin(\pi\beta)<0$}
        \label{fig:nn-beta-leq0}
    \end{subfigure}
    \begin{subfigure}[b]{0.48\textwidth}
        \includegraphics[width=\textwidth]{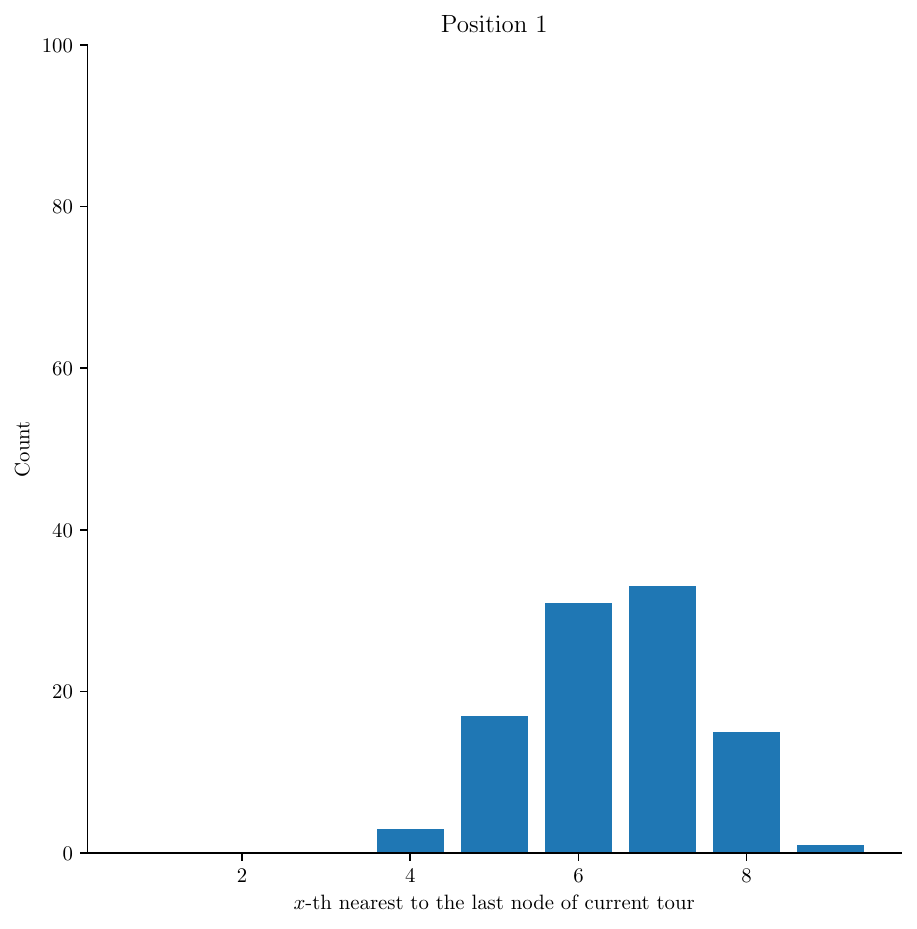}
        \caption{$\beta = 0.9, \sin(\pi\beta)>0$}
        \label{fig:nn-beta-geq0}
    \end{subfigure}
    \caption{Frequency of selecting the $x$-th Nearest Neighbor at Position 1 of the Tour (TSP-10). Based on 100 TSP-10 instances with uniformly random node placements in a unit square. Distances are scaled so the maximum length of an edge is $\sqrt 2$. We set $\gamma = 1$. Using the Figure \ref{fig:nn-beta-leq0} and \ref{fig:nn-beta-geq0} contains more local level information which stores the frequency of the number of times a node that is the $x$-th nearest neighbor to node 1 (in terms of Euclidean Distance) has the highest Q-value among the agent's options, and is therefore selected as part of its policy.}
    \label{fig:beta-comp-nn}
\end{figure}
\begin{figure}[!t]
    \centering
    \begin{subfigure}[b]{0.24\textwidth}
        \includegraphics[width=\textwidth]{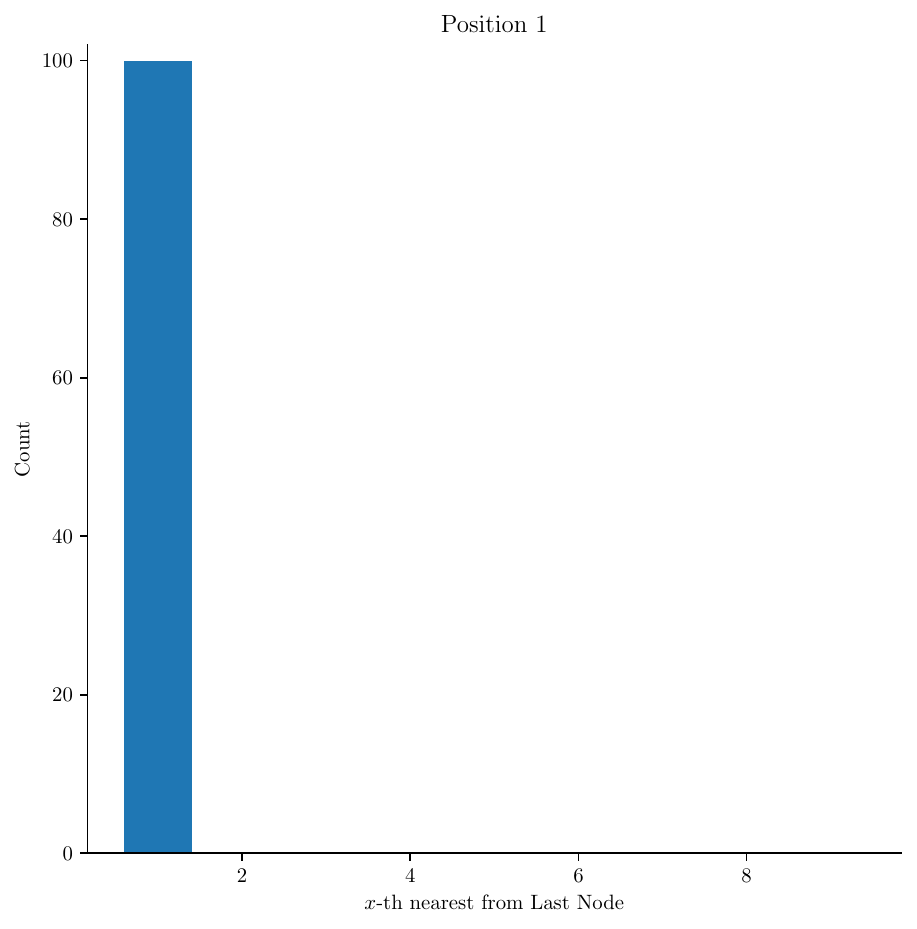}
        \caption{$\gamma = 0.001$}
        \label{fig:gamma-vis-g0.001}
    \end{subfigure}
    \begin{subfigure}[b]{0.24\textwidth}
        \includegraphics[width=\textwidth]{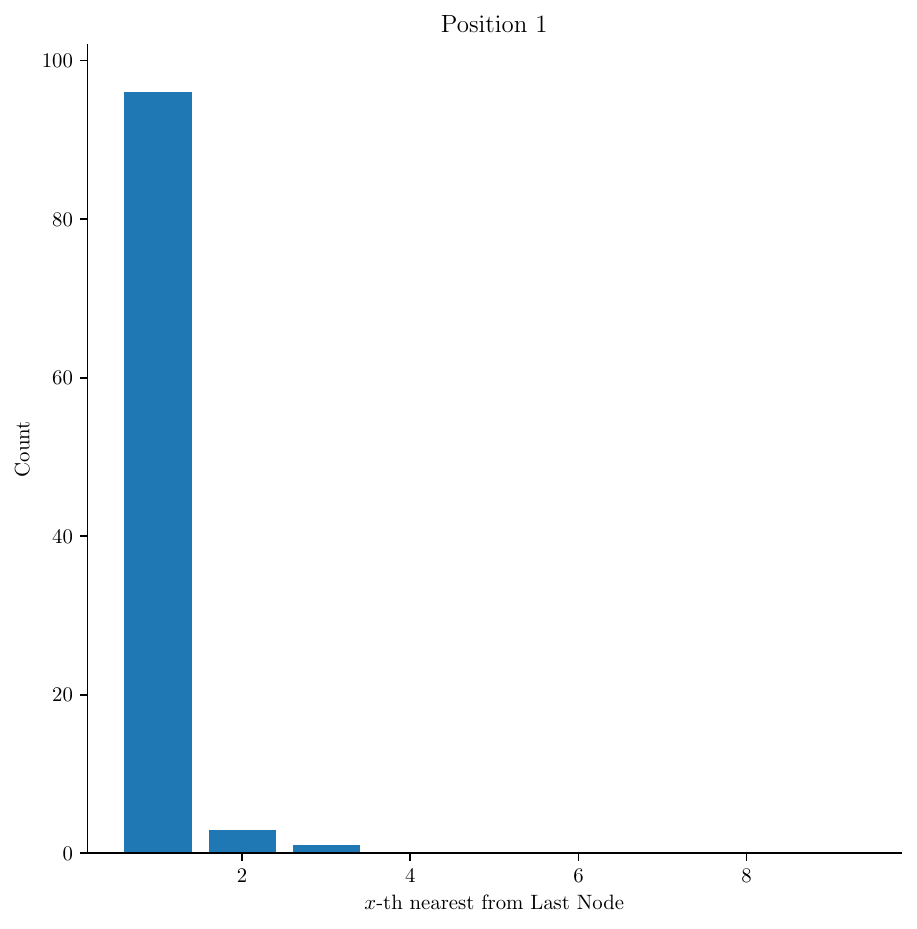}
        \caption{$\gamma = 0.5$}
        \label{fig:gamma-vis-g0.5}
    \end{subfigure}
    \begin{subfigure}[b]{0.24\textwidth}
        \includegraphics[width=\textwidth]{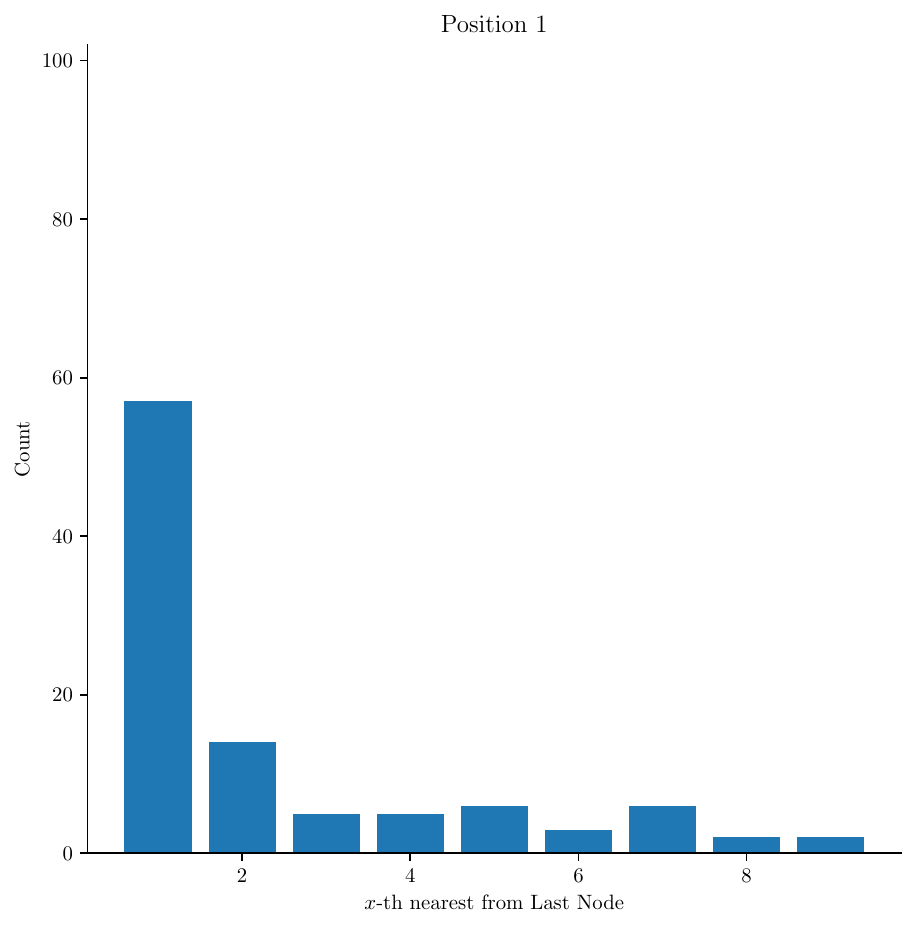}
        \caption{$\gamma = 1.3$}
        \label{fig:gamma-vis-g1.3}
    \end{subfigure}
    \begin{subfigure}[b]{0.24\textwidth}
        \includegraphics[width=\textwidth]{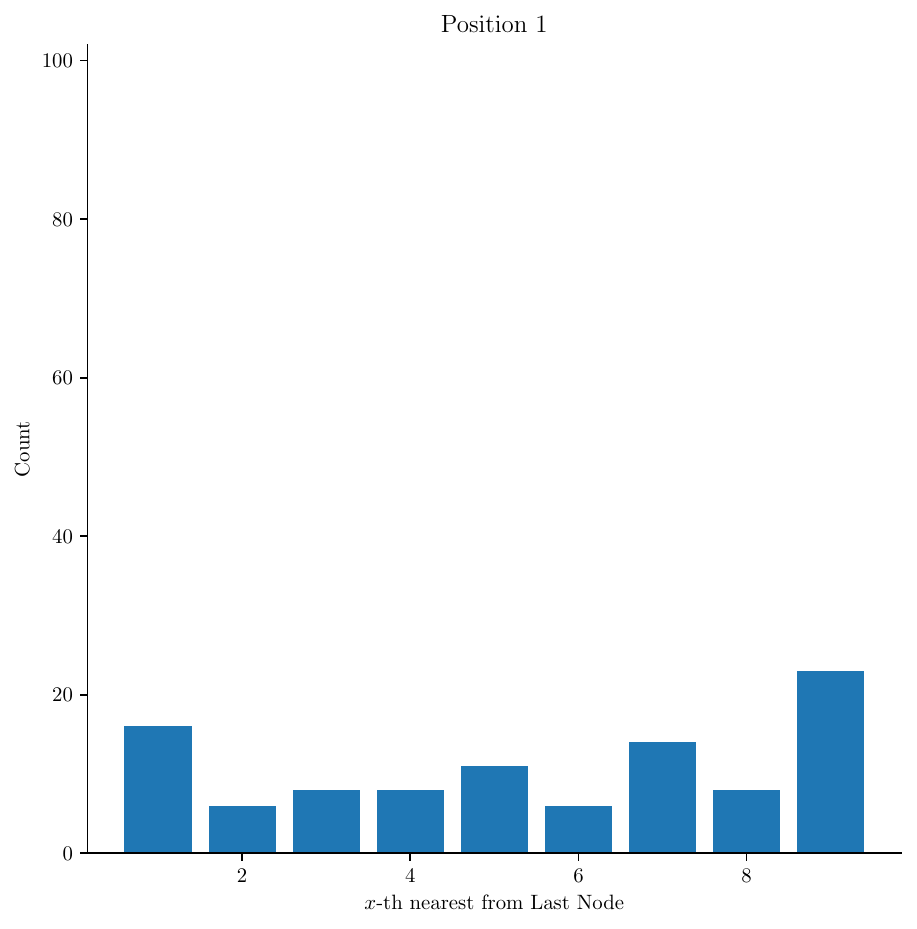}
        \caption{$\gamma = 1.63$}
        \label{fig:gamma-vis-g1.63}
    \end{subfigure}
    \caption{Effect of varying $\gamma$ on the Frequency of selecting the $x$-th Nearest neighbor at Tour Position 1 (TSP-10). Refer to the caption of \ref{fig:gamma-vis-qval} for more information.}
    \label{fig:gamma-ablation}
\end{figure}
\begin{figure}[!t]
    \centering
    \begin{subfigure}[b]{0.24\textwidth}
        \includegraphics[width=\textwidth]{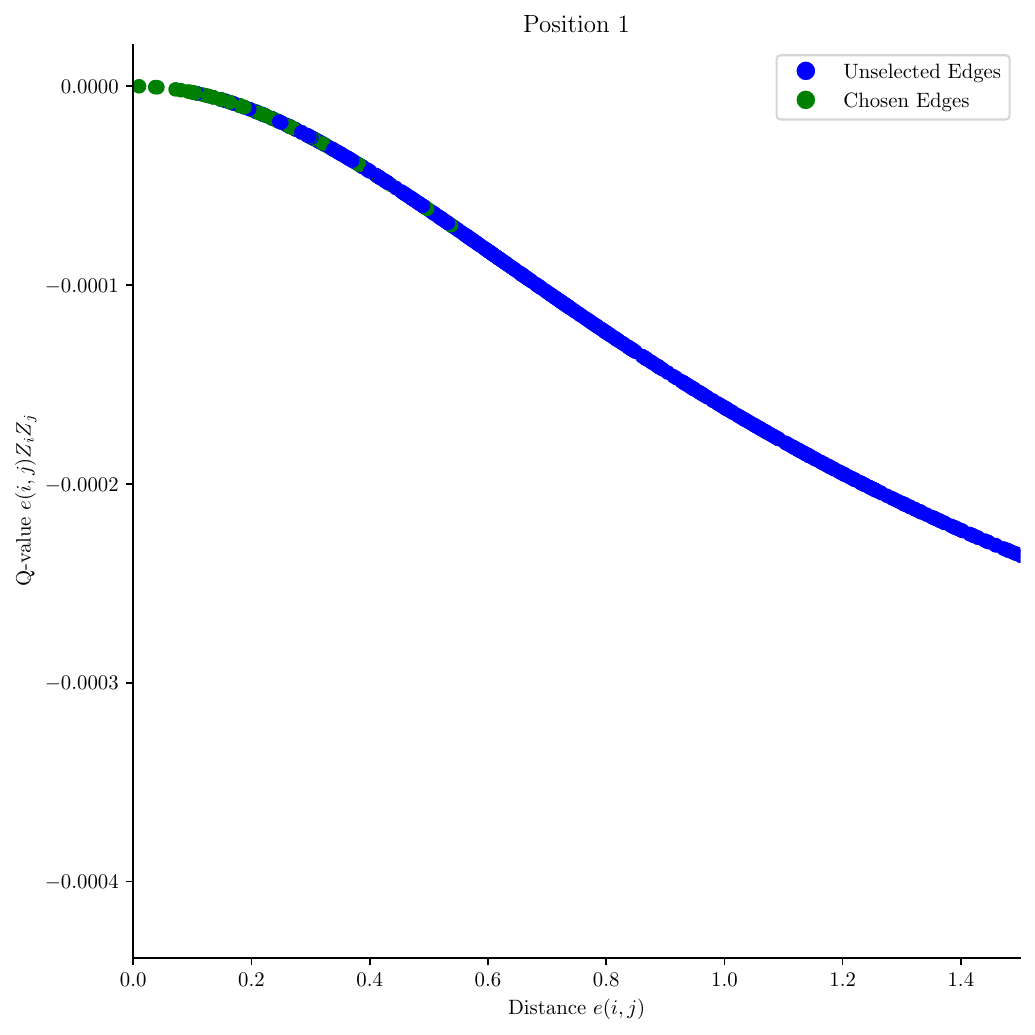}
        \caption{$\gamma = 0.001$}
        \label{fig:gamma-vis-qval-g0.001}
    \end{subfigure}
    \begin{subfigure}[b]{0.24\textwidth}
        \includegraphics[width=\textwidth]{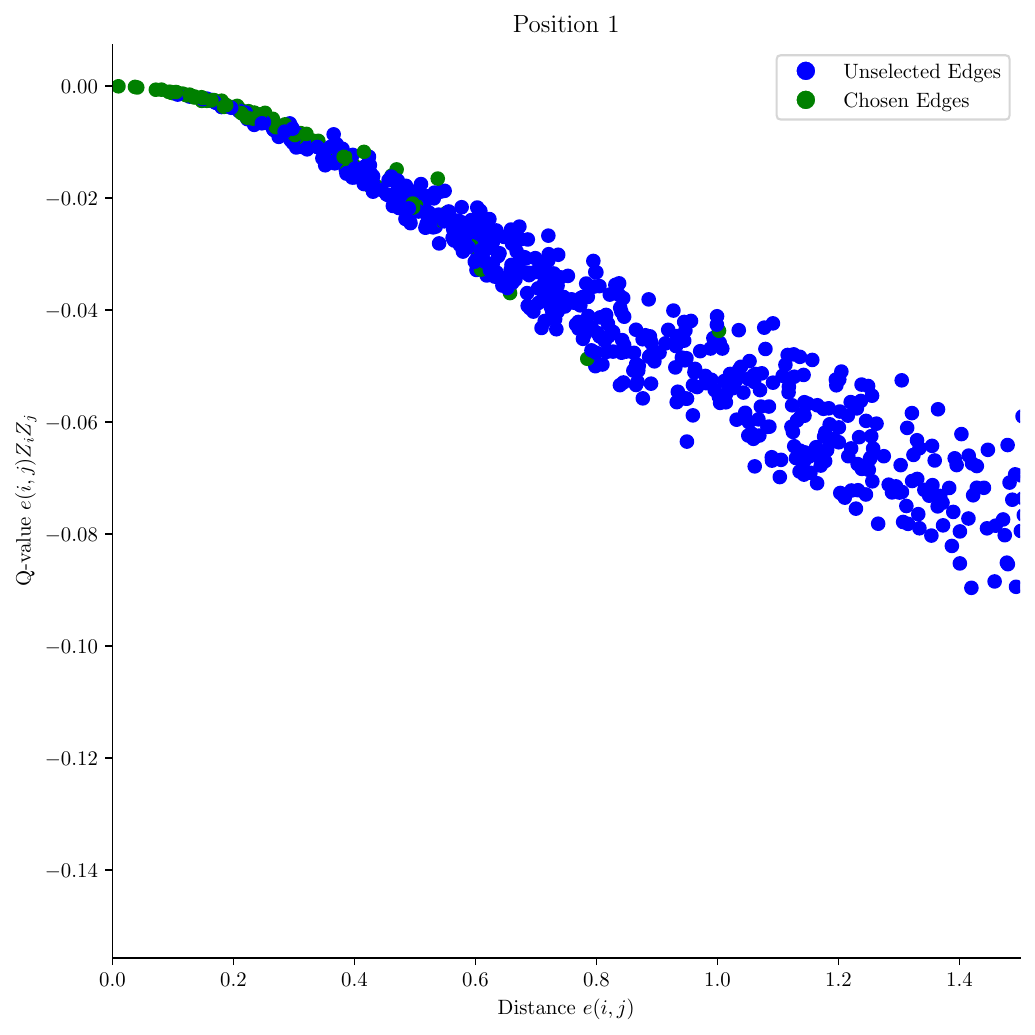}
        \caption{$\gamma = 0.5$}
        \label{fig:gamma-vis-qval-g0.5}
    \end{subfigure}
    \begin{subfigure}[b]{0.24\textwidth}
        \includegraphics[width=\textwidth]{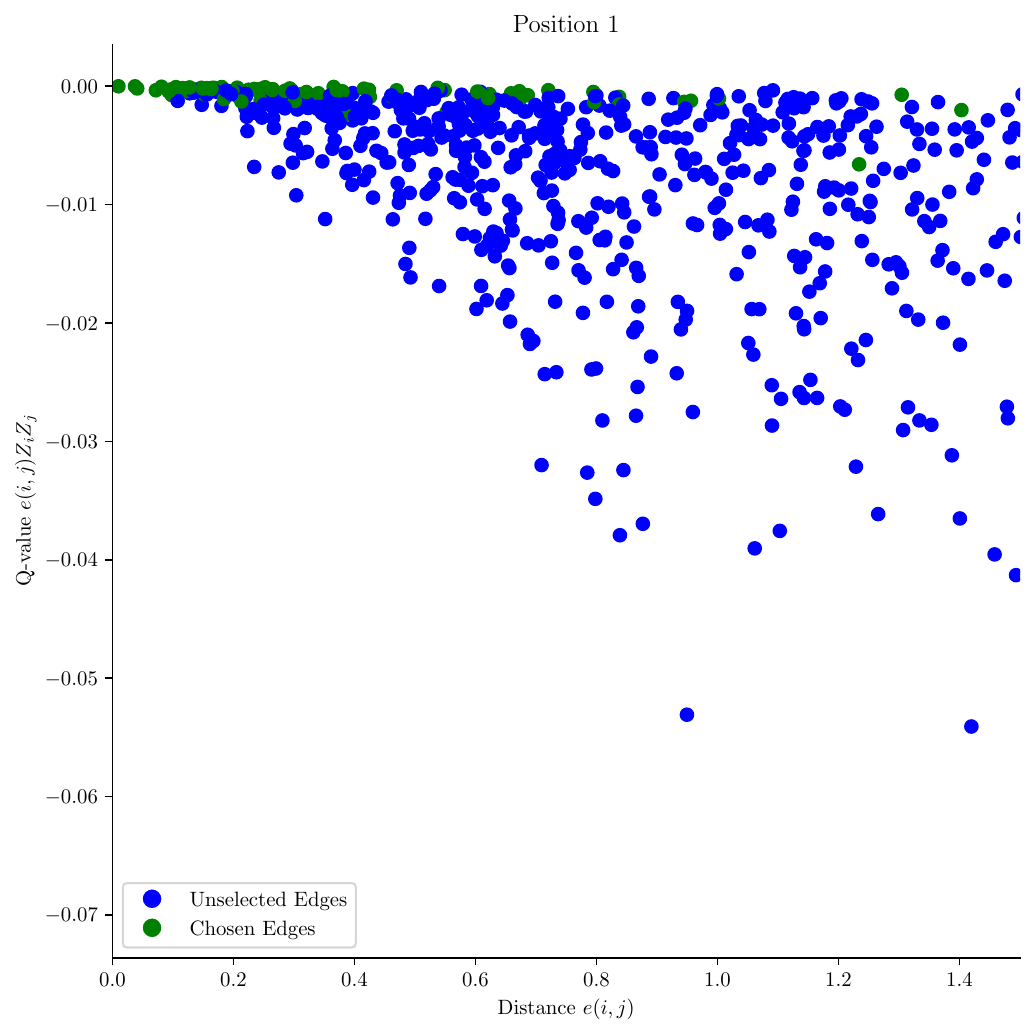}
        \caption{$\gamma = 1.3$}
        \label{fig:gamma-vis-qval-g1.3}
    \end{subfigure}
    \begin{subfigure}[b]{0.24\textwidth}
        \includegraphics[width=\textwidth]{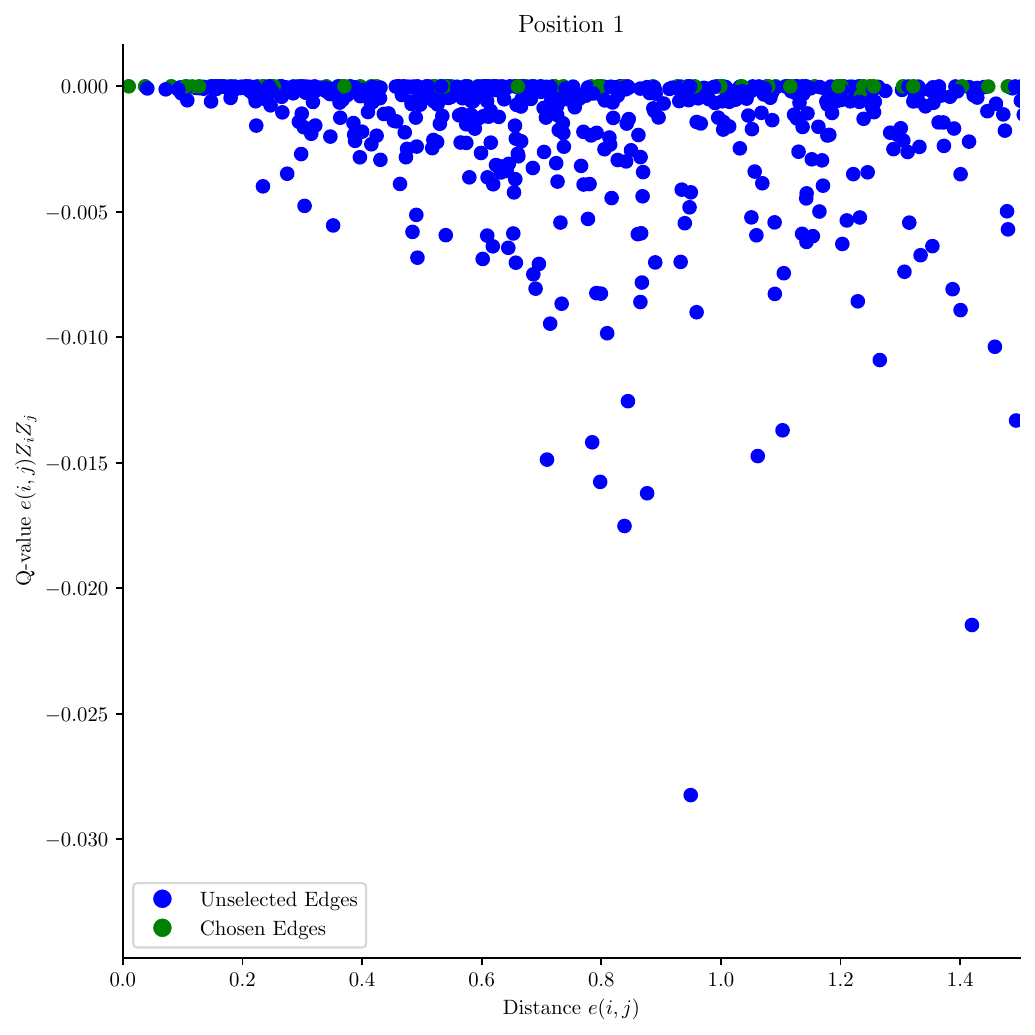}
        \caption{$\gamma = 1.63$}
        \label{fig:gamma-vis-qval-g1.63}
    \end{subfigure}
    \caption{Effect on varying $\gamma$ on the Q-value Distribution at Tour Position 1 (TSP-10). Figures \ref{fig:gamma-ablation} and \ref{fig:gamma-vis-qval} are based on the validation set for TSP-10. Distances are scaled so the maximum length of an edge is $\sqrt 2$ and we set $\beta = 1.1$, satisfying $\sin(\pi\beta)<0$. $\gamma = 1.63$ is close to the upper limit of $\gamma$ (ie. $\frac{\pi/2}{\tan^{-1}\sqrt2}\approx1.64$). Figure \ref{fig:gamma-ablation} and \ref{fig:gamma-vis-qval} are the same plots as Figure \ref{fig:beta-comp-nn} and \ref{fig:beta-comp-qval} respectively. Refer to the notes of Figures \ref{fig:beta-comp-nn} for more information.   }
    \label{fig:gamma-vis-qval}
\end{figure}

\subsection{Analysis for Theorem \ref{thm:good-beta}}
Figure \ref{fig:grid-search-heatmap} contains plots of Mean Gaps using a Depth 1 EQC with respect to parameters $\gamma$ and $\beta$. In all of the plots in Figure \ref{fig:grid-search-heatmap}, smaller optimality gaps are color-coded as brighter colors, while larger optimality gaps are encoded with darker colors. As we have only performed an exhaustive search of the optimization landscape between $0\leq \gamma,\beta \leq 2\pi$ for TSP 10 and 20, so for this analysis, we will focus on Figures \ref{fig:grid-search-heatmap-tsp10} and \ref{fig:grid-search-heatmap-tsp20}. 

The regions of good optimality gaps (ie. $\leq$ 1.05 (TSP-10) and $\leq 1.1$ (TSP-20))  correspond to the yellow regions of Figures \ref{fig:grid-search-heatmap-tsp10} and \ref{fig:grid-search-heatmap-tsp20}. These are precisely within the region of $1 < \beta < 2$, $3 < \beta < 4$ and $5 < \beta < 6$. These regions correspond to the values of $\beta$ where $\sin (\pi\beta) < 0$. There are also regions of poor optimality gaps (ie. $\geq$ 2.3) which correspond to the dark blue regions. These are between $0 < \beta < 1$, $2 < \beta < 3$, $4 < \beta < 5$ and so on. These regions correspond to the values of $\beta$ where $\sin(\pi\beta) > 0$. Additionally, \textit{within} each region of high optimality (or low optimality), the mean optimality gap is unchanged with respect to $\beta$. This is because $\sin(\pi\beta)$ at depth 1 is just a multiplier to the Q-value, and all Q-values are scaled proportionally, which does not affect the $\arg \max$ of the Q-value function. This is also empirical evidence that within $0< \gamma e_{ij} < \pi/2$ for all $i,j$, there is no optimal policy for $\sin(\pi\beta) >0$ (Theorem \ref{thm:good-beta}).

Figures \ref{fig:beta-comp-qval} and \ref{fig:beta-comp-nn} visualizes the Agent's Policy at Position 1 of the tour. To standardize all comparisons, we analyze the Agent's behavior at Position 1 of the tour. We have also provided analysis for other positions of the tour in Supplementary Information \bref{appendix:agent-vis-additional}{C}, but they are mostly similar and the analysis in this section also applies to the ones at other positions. When $\sin(\pi\beta) <0$ (see Figures \ref{fig:qval-beta-leq0} and \ref{fig:nn-beta-leq0})), the agent's policy is logical. Since all Q-values are negative, the agent is selecting among the smallest magnitude Q-values, which correspond to the nearest nodes of the tour. On the other hand, when $\sin(\pi\beta)>0$, the agent selects the \textit{further} neighbors most of the time. When $\sin(\pi\beta) > 0$ and all cosine terms in the product of \ref{eq:eqc-qval} are larger than zero, the largest edge weights have the largest Q-values, and are therefore part of the policy. 

Lastly, we also show that there is a positive relationship between the Q-value magnitude (ie. absolute value) and node distance. This has been explained in Lemma \ref{lem:monotonic-edgeweights}. However, on some occasions this positive relationship becomes less clear (such as Figure \ref{fig:gamma-vis-qval-g1.63}). This will be further clarified in the analysis for Theorem \ref{thm:gamma-greedy}. 

\subsection{Analysis for Theorem \ref{thm:gamma-greedy}} 
Figure \ref{fig:gamma-ablation} and \ref{fig:gamma-vis-qval} shows the effect of varying $\gamma$ on the Agent's policy. In this analysis, $\beta$ has been chosen such that $\sin(\pi\beta) < 0$. Similar to Figure \ref{fig:beta-comp-qval} and \ref{fig:beta-comp-nn}, all of these plots are recorded at position 1 of the tour. For an illustration of other tour positions, please refer to Supplementary Information \bref{appendix:agent-vis-additional}{C}. When $\gamma$ is near zero (but non-zero), the Q-value function models a sinusoidal curve (Figure \ref{fig:gamma-vis-qval-g0.001}) and the agent selects its nearest neighbor all of the time (Figure \ref{fig:gamma-vis-g0.001}). At $\gamma$ close to zero, the policy is essentially a nearest neighbor heuristic. When $\gamma$ increases, the agent begins to explore more of its nearest neighbors. When $\gamma = 1.63$ (Figure \ref{fig:gamma-vis-qval-g1.63}) (which is close to the maximum value of $\gamma$ as $\frac{\pi/2}{\tan^{-1} \sqrt2} \approx 1.64$, the agent begins to selects its furthest node more frequently than its nearest node. This validates the case where both small edge weights (nearest neighbor) or large edge weights (furthest neighbors) maximize the Q-value when $\gamma$ is large, resulting in the nearest and furthest neighbors becoming the top two most popular choices (Figure \ref{fig:gamma-vis-g1.63}) and a "V"-shaped distribution of Q-values (Figure \ref{fig:gamma-vis-qval-g1.63}), where smallest and largest edge weights generally have Q-values of smallest magnitude. These simulations hence validate Theorem \ref{thm:gamma-greedy}.

\subsection{Analysis for Proposition \ref{prop:gamma-range}} \label{sect:size-invariant-exp-gamma-range-justification}
\begin{figure}[!t]
    \centering
    \begin{subfigure}[b]{0.48\textwidth}
        \includegraphics[width=\textwidth]{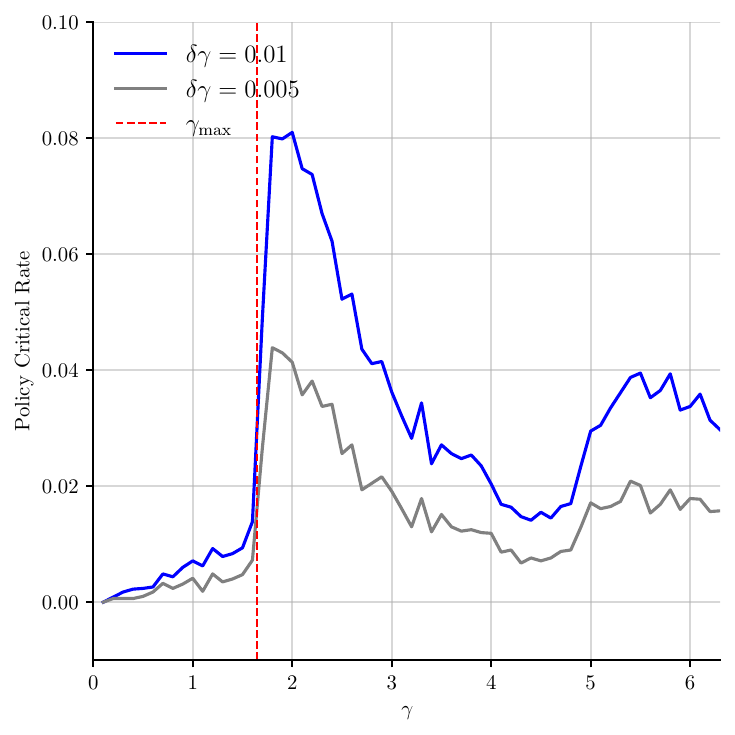}
        \caption{PC Rate}
        \label{fig:prop-gamma-range-flip-rate}
    \end{subfigure}
    \begin{subfigure}[b]{0.48\textwidth}
        \includegraphics[width=\textwidth]{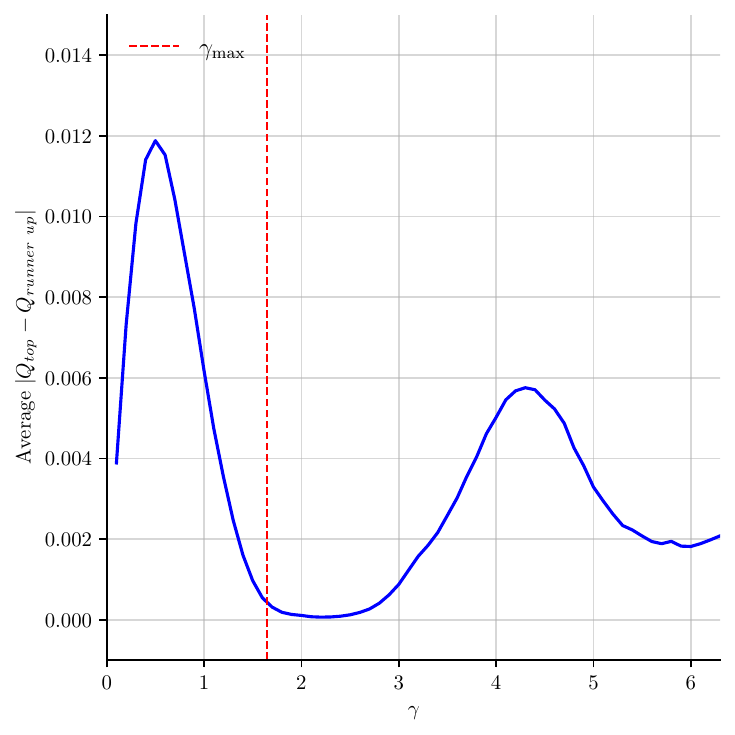}
        \caption{Average Margin}
        \label{fig:prop-gamma-range-avg-margin}
    \end{subfigure}
    \caption{Policy-Critical Flip Rates and Average Margin against $\gamma$, based on 500 TSP-10 instances with uniform node placements. The red line visualizes $\gamma = \gamma_{\max} \approx 1.64$, as we scaled the largest edge to have an Euclidean Distance $d_{\max} = \sqrt{2}$. Both of these visualizations were done on a grid from $\gamma \in [0.1, 6.3]$, with step size $\Delta \gamma = 0.1$, and $\beta = 1.1$. For computation of PC rate, we perturb $\gamma \in \Gamma$ by $\pm\Delta\gamma$, and estimate the flip rate by estimating the proportion of $\pi(\cdot;\gamma \pm \delta \gamma,\cdot) \neq \pi(\cdot;\gamma,\cdot)$ when there is at least two nodes unexplored in the graph (see Equation \ref{eq:flip-rate}). Average Margin is computed by the mean absolute difference of the top-two Q-values at a particular value of $\gamma$ (see Equation \ref{eq:avg-margin}). }
    \label{fig:prop-gamma-range}
\end{figure}
In order to visualize the training instability when $\gamma \in [\gamma_{\max}, 2\pi]$, we measure the Policy Critical Rate (PC Rate) and Average Margin on a dataset of 500 TSP-10 instances. PRC Rate (Equation \ref{eq:flip-rate}) shows the sensitivity of the $\arg \max$ to perturbations in $\gamma$ (higher value $\rightarrow$ more sensitive). Average Margin (Equation \ref{eq:avg-margin}) captures policy stability from a different angle: As the EQC is Lipschitz Continuous in $\gamma$, a larger margin implies that in most cases, a larger change in $\gamma$ is required for the policy function to change (smaller value $\rightarrow$ less stable).
\begin{equation}
    \text{PC Rate}(\gamma,\delta\gamma ) = \frac{1}{Z}\sum_{{\cal D}}\sum_{\text{tour pos}}1[\pi(\cdot;\gamma-\delta\gamma,\cdot) \neq \pi(\cdot;\gamma,\cdot)] + 1[\pi(\cdot;\gamma+\delta\gamma,\cdot) \neq \pi(\cdot;\gamma,\cdot)],
    \label{eq:flip-rate}
\end{equation}
where $[\cdot]$ denotes Boolean Function having value $1$ is the boolean condition is true, and $Z$ is a normalization constant.
\begin{equation}
    \text{Average Margin}(\gamma) = \frac{1}{Z}\sum_{{\cal D}}\sum_{\text{tour pos}}|Q(\cdot;\gamma,\cdot)_{(1)} - Q(\cdot;\gamma,\cdot)_{(2)}|,
    \label{eq:avg-margin}
\end{equation}
where $Z$ is a normalization constant and $Q(\cdot;\gamma,\cdot)_{(2)}$ denotes the first-largest and second-largest Q-values for a particular instance $(s,e)$ and last node $t$. 

In both metrics (Equations \ref{eq:flip-rate}-\ref{eq:avg-margin}), $Z$ denotes the total number of RL \textit{steps} (see Section \ref{sect:background-methodology-mdp}) when there was \textit{more than one} unexplored node.

Figure \ref{fig:prop-gamma-range} empirically validates Proposition \ref{prop:gamma-range} by showing that both the PC rate and the average margin exhibit strong dependence on $\gamma$. For small $\gamma < \gamma_{\max}$, both metrics remain stable: PC Rate remains near zero, indicating that small changes in the value of $\gamma$ do not affect the $\arg \max$ of Q-values, and the margin between the top two Q-values remains consistently higher than $0.004$. 

As $\gamma$ approaches $\gamma_{\max}$, the average margin begins to shrink gradually as the cosine term of the largest edge weight $\cos(\gamma e_{\max})$ is approaching zero. This is due to Theorem \ref{thm:gamma-greedy} discussed earlier, where we argued that both a small edge weight or a large distance from other nodes would be just as likely to maximize the Q-value. Therefore, it is expected that the average margin decreases as $\gamma$ approaches $\gamma_{\max}$ in the safe region. 

PC Rates peak at $\gamma \approx 1.7-1.8$ before decaying slowly. This is due to several $\cos(e_{ab}\gamma)$ factors approaching zero at close but distinct points, and the product terms in the Q-value to fluctuate rapidly. PC-Rate is high in areas where the Average Margin between the two Q-values is very small (ie. see regions $\gamma = 1.7-2.0$ in Figures \ref{fig:prop-gamma-range-flip-rate} and \ref{fig:prop-gamma-range-avg-margin}). Interestingly, we note that PC Rate decreases and the Average Margin becomes slightly more stable in the domain $\gamma \in [4,5]$. This may be because most of the cosine terms have changed sign and are negative in the domain $\pi/2 < \gamma e_{ab} < 3\pi/2$. Nevertheless, the policy is still comparatively less stable as the average margin is generally smaller in this region than in the safe $\gamma$ region of $\gamma \in (0, \gamma_{\max})$.

Overall, these results show that optimization in $[\gamma_{\max}, 2\pi]$ is fragile. Small perturbations in $\gamma$ can cause discontinuous changes in the greedy policy, validating our design decision to restrict the search to the \textit{safe domain} of $\gamma \in (0, \gamma_{\max})$.

\subsection{Policy Visualizations of Deeper EQCs} \label{sect:size-invariant-exp-deeper-eqcs}
\begin{figure}[!t]
    \centering
    \begin{subfigure}[b]{0.24\textwidth}
        \includegraphics[width=\textwidth]{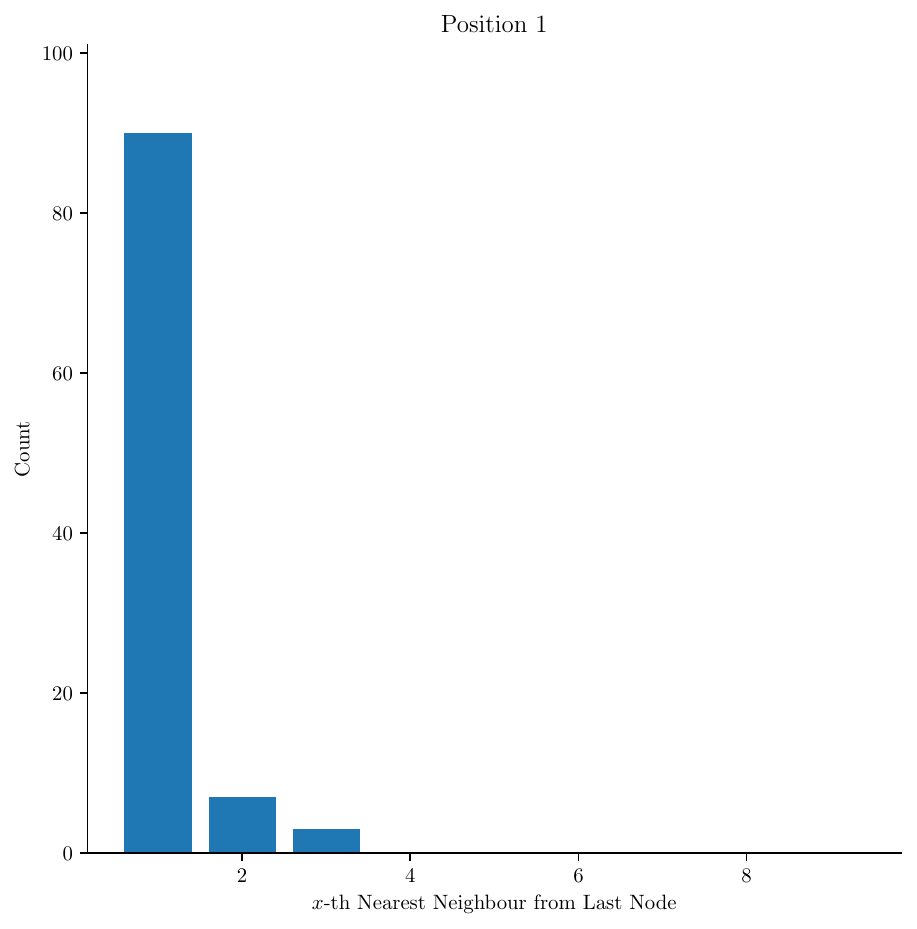}
        \caption{$d=1$}
        \label{fig:gamma-vis-d1}
    \end{subfigure}
    \begin{subfigure}[b]{0.24\textwidth}
        \includegraphics[width=\textwidth]{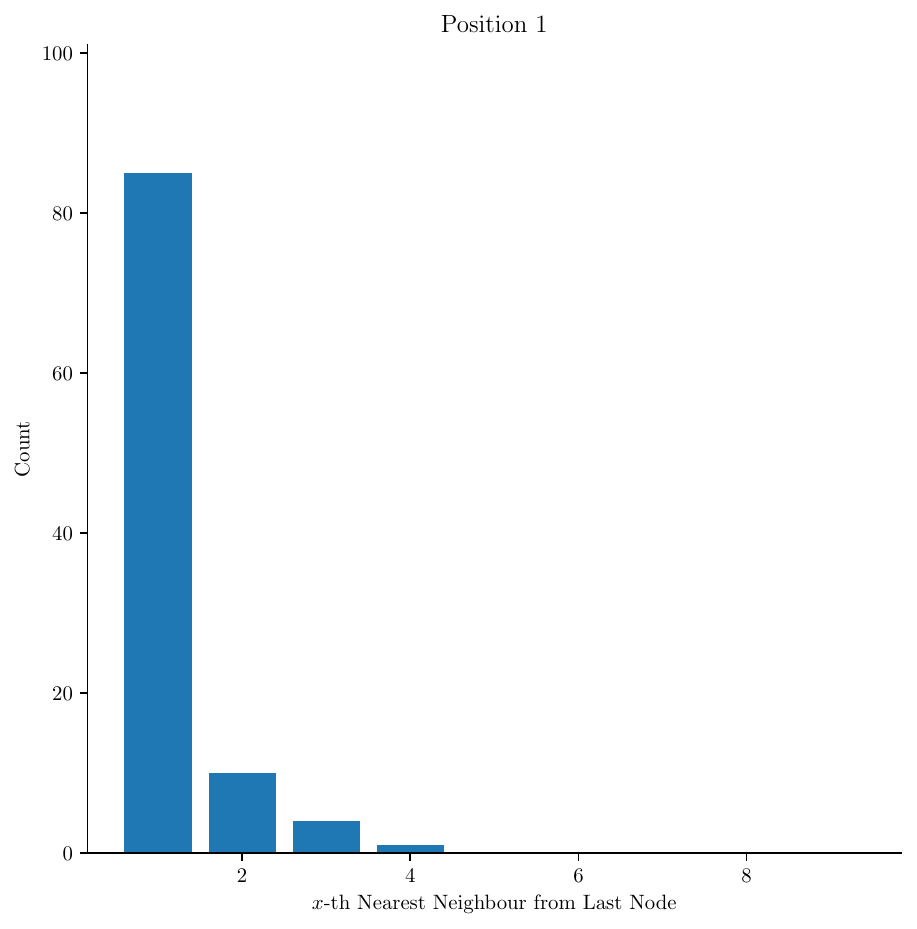}
        \caption{$d=2$}
        \label{fig:gamma-vis-d2}
    \end{subfigure}
    \begin{subfigure}[b]{0.24\textwidth}
        \includegraphics[width=\textwidth]{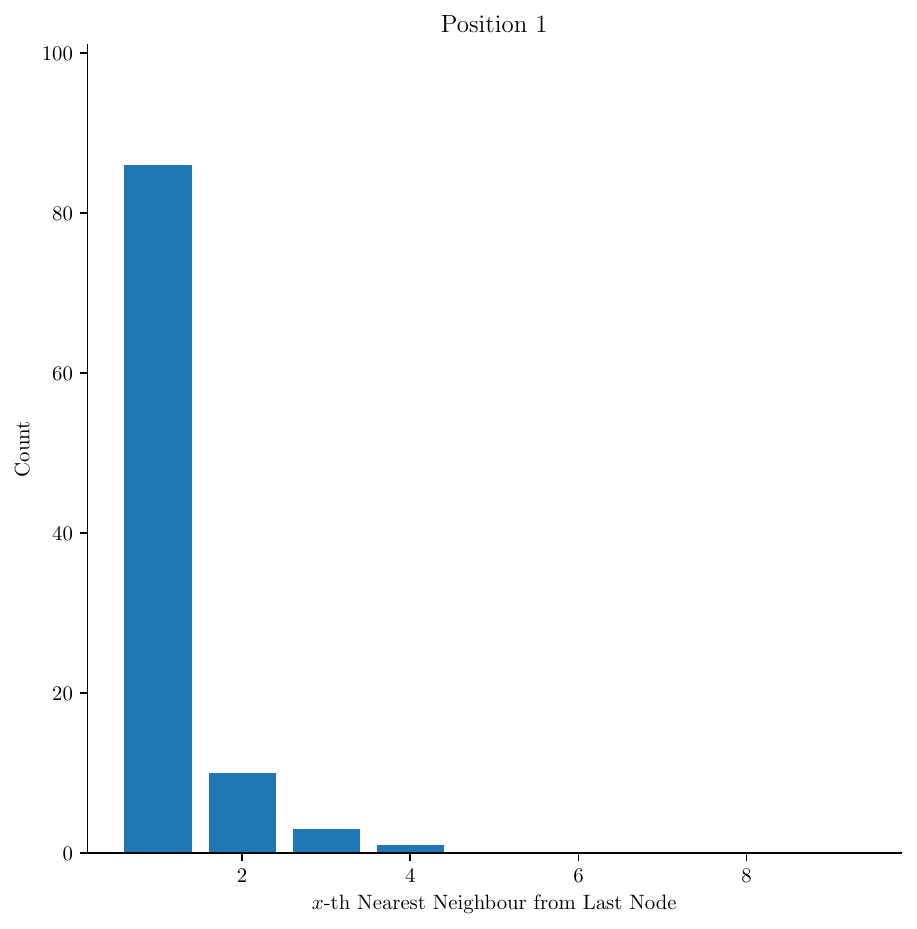}
        \caption{$d=3$}
        \label{fig:gamma-vis-d3}
    \end{subfigure}
    \begin{subfigure}[b]{0.24\textwidth}
        \includegraphics[width=\textwidth]{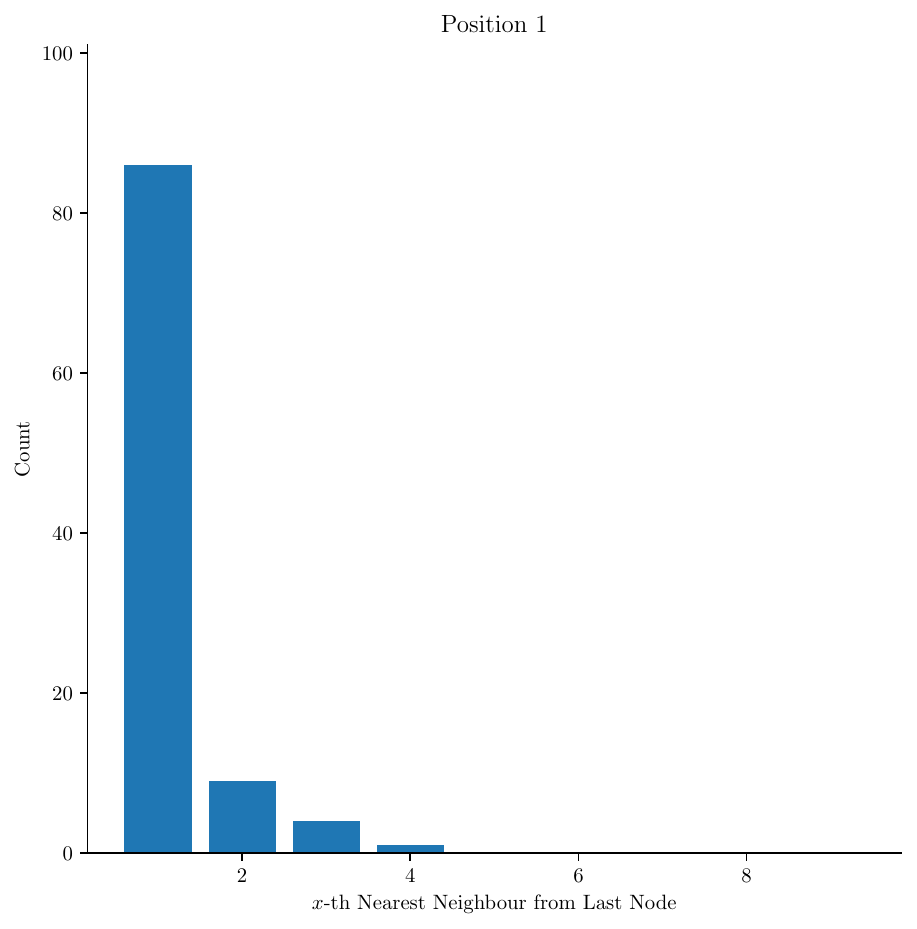}
        \caption{$d=4$}
        \label{fig:gamma-vis-d4}
    \end{subfigure}
    \caption{Frequency of the optimal policy found for EQC Depth $d$ selecting the $x$-th Nearest neighbor at Tour Position 1 (TSP-10). Visualizations are based on a randomly generated dataset of TSP-10 instances.}
    \label{fig:gamma-vis-d1-to-d4}
\end{figure}
\begin{figure}[!t]
    \centering
    \begin{subfigure}[b]{0.24\textwidth}
        \includegraphics[width=\textwidth]{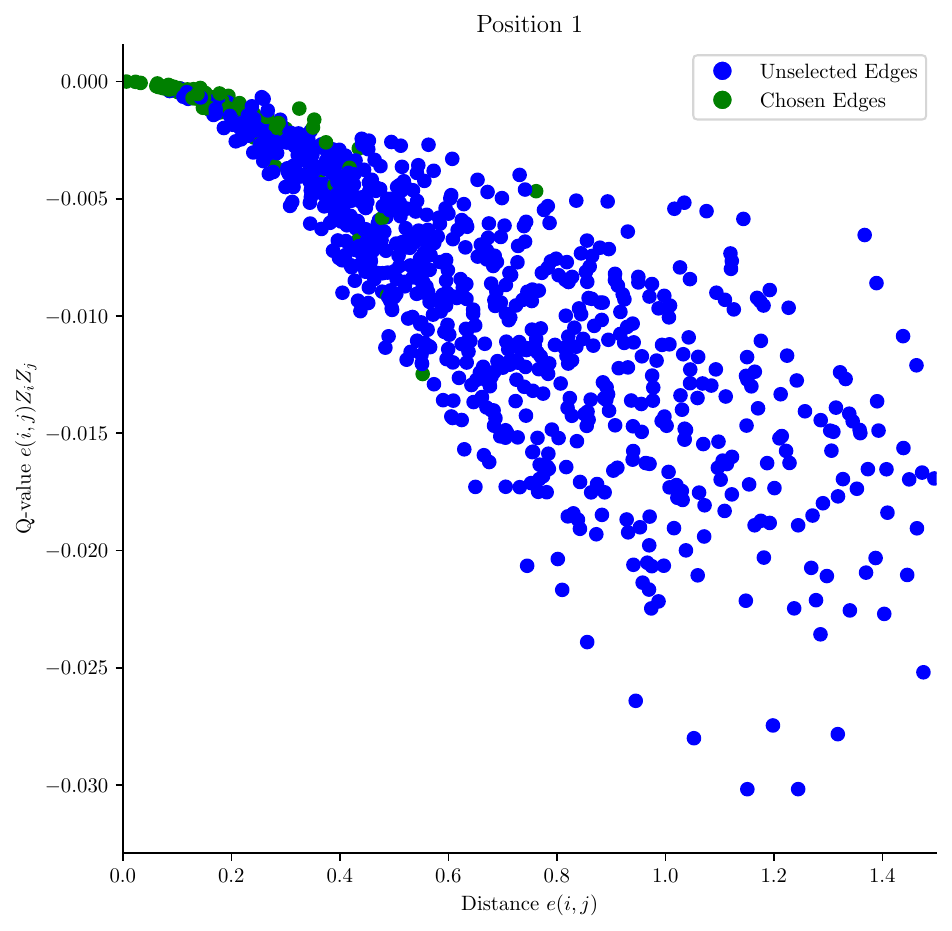}
        \caption{$d=1$}
        \label{fig:qval-d1}
    \end{subfigure}
    \begin{subfigure}[b]{0.24\textwidth}
        \includegraphics[width=\textwidth]{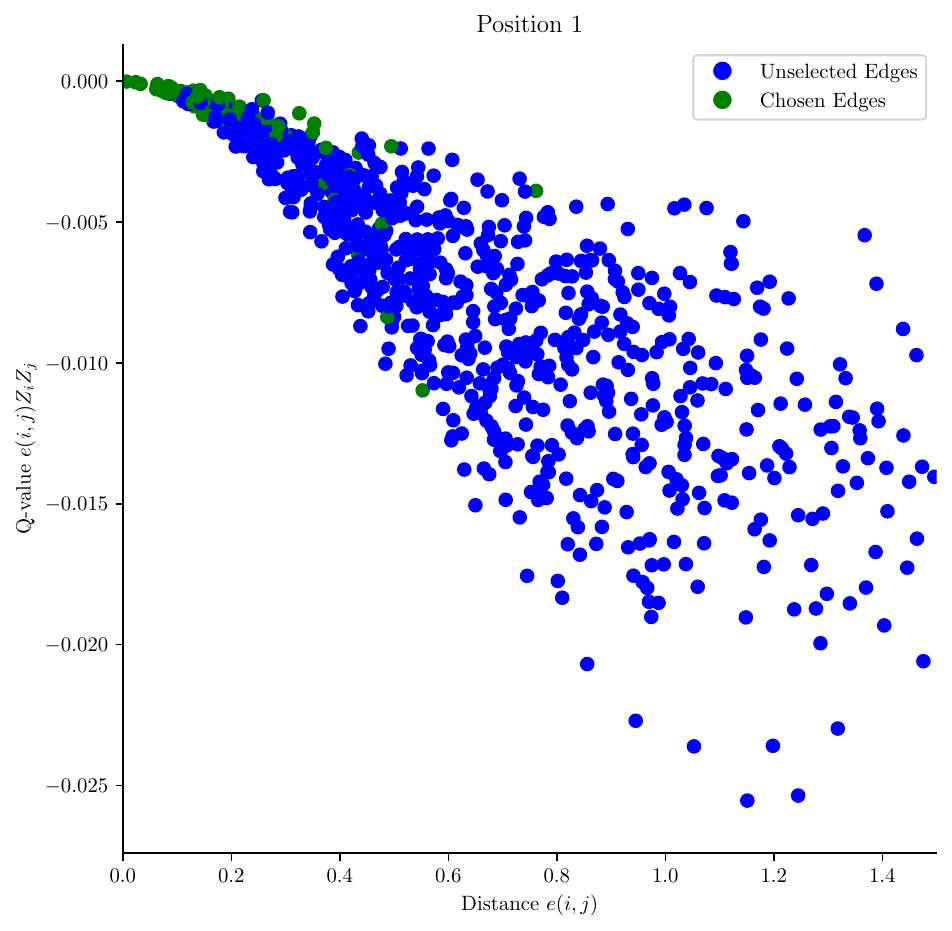}
        \caption{$d=2$}
        \label{fig:qval-d2}
    \end{subfigure}
    \begin{subfigure}[b]{0.24\textwidth}
        \includegraphics[width=\textwidth]{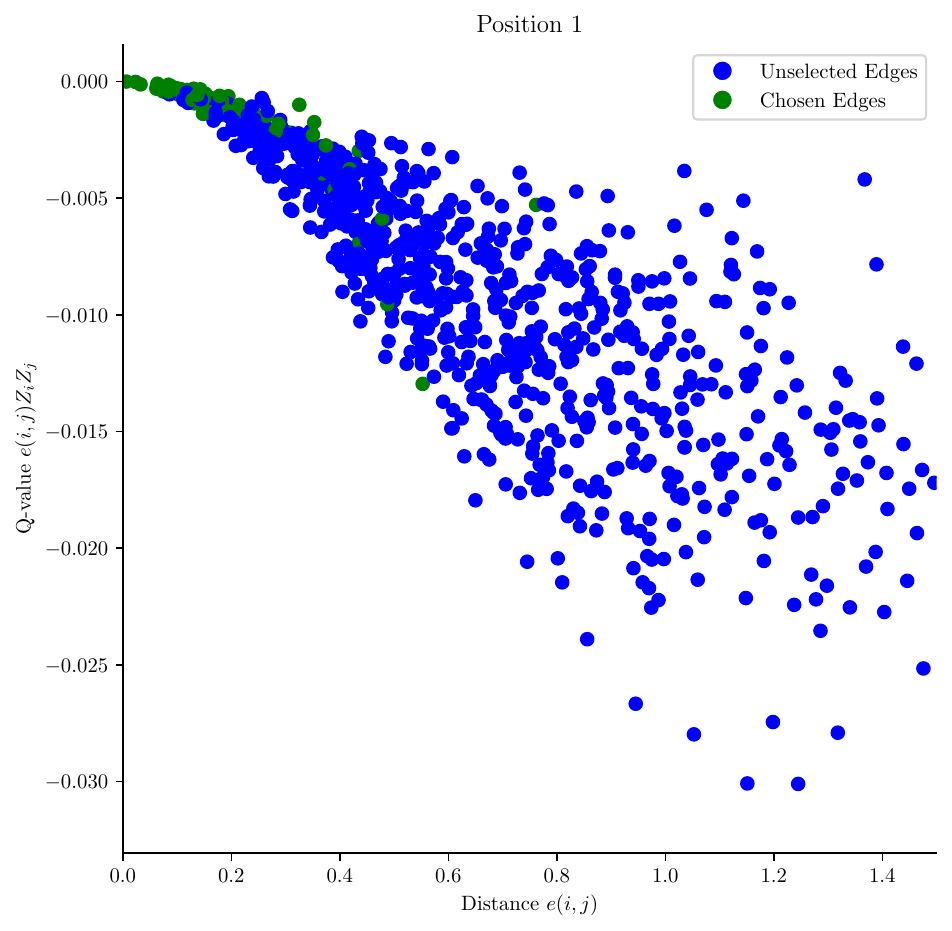}
        \caption{$d=3$}
        \label{fig:qval-d3}
    \end{subfigure}
    \begin{subfigure}[b]{0.24\textwidth}
        \includegraphics[width=\textwidth]{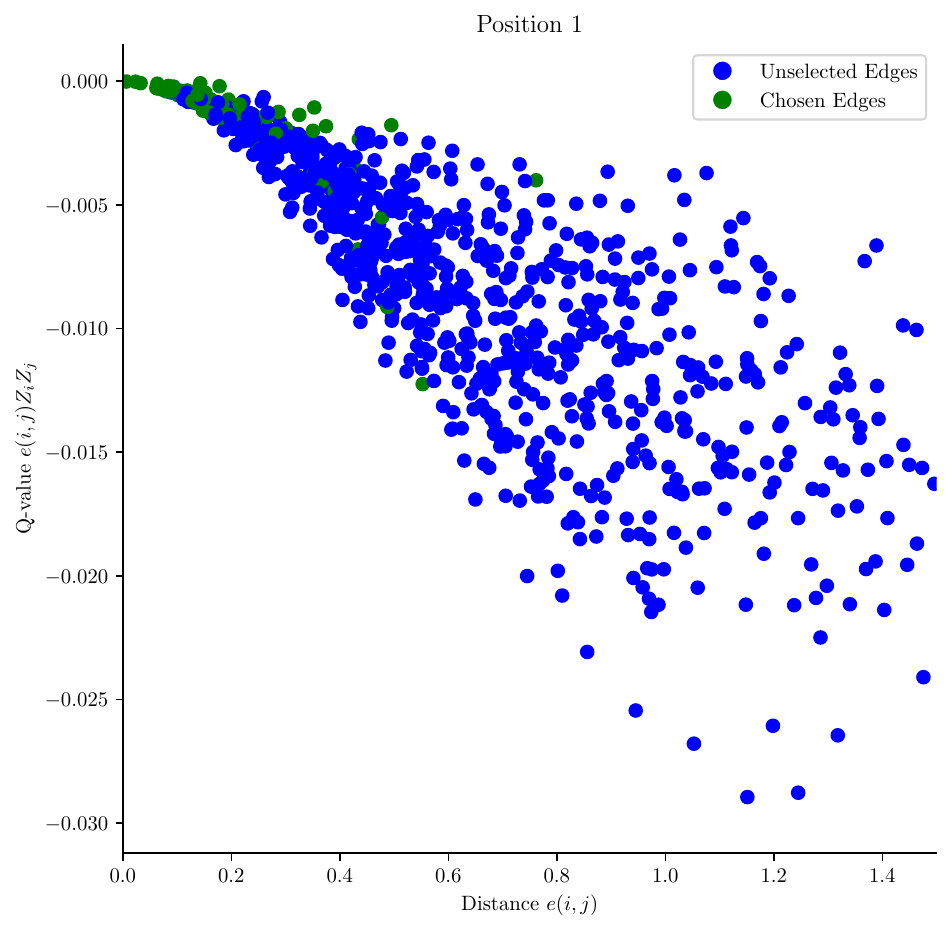}
        \caption{$d=4$}
        \label{fig:qval-d4}
    \end{subfigure}
    \caption{Q-value Distribution of the optimal policy found for EQC Depth $d$ at Tour Position 1 (TSP-10). Visualizations are created based on the TSP-10 Validation set.}
    \label{fig:qval-d1-to-d4}
\end{figure}

Figures \ref{fig:gamma-vis-d1-to-d4} and \ref{fig:qval-d1-to-d4} shows that the optimal policy at EQC Depth 1 is similar to the optimal policies at EQC Depths 2 to 4. We use the same policy visualization as per Section \ref{sect:size-invariant-exp} (but this time using an optimal EQC policy at Depths $d=1,2,3,4$). Surprisingly, we found that even at different depths, we found that similar trends exist in both the Frequency of nearest neighbor selection as well as Q-value distribution at depths larger than one. This suggests that the EQC at depths 2 to 4 might be converging to a similar policy as the Depth 1 EQC. Because the EQC at larger depths is still following a general trend of assigning shorter edges lower Q-values than longer edges (as illustrated in Figures \ref{fig:qval-d2}-\ref{fig:qval-d4}), we converge to a similar strategy of selecting the nearest neighbor most frequently, then the second nearest neighbor the next most frequently, and so on. This explains the small differences in Mean Test Gaps achieved by deeper EQCs as opposed to the EQCs at Depth 1.  

\section{Conclusion}\label{sect:conclusion}
In this work, we extended the present understanding of the EQC at Depth 1 beyond the concept of equivariance. We show that effective QRL policies reside within a small search space of the parameter landscape, and importantly, this region is \textit{size-invariant} - where different TSP-instance sizes would have a similar search space of parameters. Building on this insight, we introduced a \textit{Size-Invariant Grid Search} (SIGS), a lightweight training optimization that replicates the performance of QRL at Depth 1 while reducing runtime significantly. This allowed us to simulate Depth-1 EQCs on TSP instances up to 350 nodes, far beyond previously tractable limits of 20 nodes. We stress, however, that SIGS only applies to Depth-1 EQCs and not to deeper EQCs.

By pushing the evaluation to 350 nodes, we also revealed an inherent limitation: Solution quality degrades eventually, with mean gaps surpassing 20\% from optimality at TSP sizes larger than 150 nodes. While deepening the EQC might offer a remedy, our simulations demonstrated only minimal gains from Depth 2 to 4,  with learned policies appearing similar qualitatively. 

These findings have two key implications. First, the Depth 1 EQC should be regarded as an \textit{informative baseline and analysis tool} rather than a scalable solver in its own right. Second, SIGS provides a practical benchmarking tool for the QRL community, enabling an exploration of scaling behaviour in a noiseless setting at previously inaccessible sizes. 

Looking forward, we hope these results inspire future research in QRL for Combinatorial Optimization. Promising directions include extending size-invariant region analysis to deeper EQC architectures or alternative SPAs. Beyond TSP, our methodology invites exploration across a broader class of Combinatorial Optimization problems, testing whether size-invariant structures emerge in other domains. 

\ifarxivpreprint
The code and dataset generation scripts used in this work will be made publicly available upon full publication. 
\fi

\ifarxivpreprint
    

\else
    \section*{Data Availability}
    Dataset generation scripts of the TSP instances in this work as well as the computation of their optimal solutions can be found in the GitHub repository (\url{https://github.com/SMU-Quantum/nature-of-depth1-eqc})
    
    \section*{Code Availability}
    The full code used to generate numerical results and figures in this work can be found in the Github Repository (\url{https://github.com/SMU-Quantum/nature-of-depth1-eqc}). 
\fi

\section*{Acknowledgements} We gratefully acknowledge Prof.\ Nobuyoshi Asai from the University of Aizu for generously providing compute resources. We also would like to thank Mr. Monit Sharma - our ``walking PennyLane documentation'' since the project's infancy - for guidance on implementation choices on backend simulators in PennyLane. We would also like to thank Prof.\ Paul Griffin for insightful feedback that strengthened the early drafts of the paper. 

\ifarxivpreprint\else
\section*{Funding} This work is funded by various internal grants from the Singapore Management University. 

\section*{Author Contributions}
All authors conceived the idea for this work. J.T. conducted the theoretical analysis and performed the numerical simulations and, together with X.W. and H.C., validated and finetuned the results. X.W. provided guidance for J.T. in understanding quantum computing concepts, and provided valuable discussion and conceptual input that supported the refinement of the size-invariant properties as well as the interpretation of the experimental results. H.C. acquired funding and provided supervision for this project. J.T. wrote the first draft of the paper, and all authors contributed to the review and the editing of the final version.

\section*{Competing Interests}
The authors declare no competing interests.
\fi

\bibliographystyle{naturemag}
\bibliography{citations}

\ifsupp
    \pagebreak
    \appendix
\section{Detailed Derivation of Proposition 1} \label{appendix:thm-gamma-range-proof}
\noindent\textit{Restatement of Proposition: (Instability for $\gamma \in [\gamma_{\max}, 2\pi]$) Assume that edge weights are drawn from a continuous distribution (see Assumption \bref{assumption:non-resonance}{1}). Fix a state $(s,d,e,t)$ and let $\gamma^\# \in {\cal Z}_{(s,d,e,t)}$. Then at $\gamma = \gamma^\#$, the relative ordering of two actions' Q-values from $t$ inverts.}

\noindent \textit{Let $\delta$ denote a small number, $\delta > 0$. If at $\gamma^\#-\delta$ or at $\gamma^\# +\delta$, either of these two actions has the top-1 Q-value, this inversion is {policy critical} at $\gamma^\#$ (see Definition \bref{def:policy-critical}{3}). }

\begin{proof}

    Let our current last node be some node $t$ and we take action $a$ (ie. visit node $a$) using edge $t\rightarrow a$. Let $b$ be a node where $b \neq t,a$ which has a scaled edge weight $e_{ab}$ from node $a$, and we define $\gamma_{ab}^\#$ such that it satisfies $\cos(\gamma_{ab}^\# e_{ab}) = 0$. We can see that $\gamma_{ab}^\# \in {\cal Z}_{(s,e),t}$.

    \noindent Under Assumption 1, with high probability, all elements in the cosine-zero set ${\cal Z}_{(s,d,e,t)}$ are unique.

    \noindent We define $\delta > 0$ such that in the range of $\gamma \in (\gamma_{ab}^\# - \delta, \gamma_{ab}^\#  + \delta)$, no other terms in the cosine-zero set for ${\cal Z}_{(s,d,e,t)}$ changes sign. We show that the \textbf{relative ordering} of the Q-values is reversed before and after $\gamma = \gamma_{ab}^\#$.
    
    \noindent The Q-values affected by the zero-crossing of $\cos(\gamma_{ab}^\# e_{ab}) = 0$ are $Q((s,e,t), a;\cdot)$ and $Q((s,e,t),b; \cdot)$. We show that the resultant ordering of Q-values before and after $\gamma_{ab}^\#$ is changed. With sufficiently small $\delta > 0$,
    \begin{equation}
        \frac{|Q((s,e,t),a;\gamma_{ab}^{\#}-\delta ,\beta)|}{|Q((s,e,t),b;\gamma_{ab}^{\#}-\delta, \beta)|}\approx \frac{|Q((s,e,t),a;\gamma_{ab}^\# + \delta,\beta)|}{|Q((s,e,t),b;\gamma_{ab}^\# + \delta ,\beta)|},
    \end{equation}
    which implies an inversion in the ordering of the Q-values $Q((s,d,e,t),a;\cdot)$ and $Q((s,d,e,t),b;\cdot)$. Without loss of generality, if 
    \begin{equation}
        Q((s,d,e,t),a;(\gamma^\#_{ab} - \delta, \beta)) < Q((s,d,e,t),b;(\gamma^\#_{ab} - \delta, \beta)),
    \end{equation} 
    then after the zero-crossing of $\cos(\gamma_{ab}^\# e_{ab}) =0$, we have:
    \begin{equation}
        Q((s,d,e,t),a;(\gamma^\#_{ab} + \delta, \beta)) > Q((s,d,e,t),b;(\gamma^\#_{ab} + \delta, \beta))
    \end{equation}
    
    \noindent If at $\gamma =\gamma_{ab}^\# - \delta$, if either
    \begin{equation}
        \pi((s,d,e,t); (\gamma^\#_{ab}-\delta , \beta)) = a\ \text{or} \ \pi((s,d,e,t), (\gamma^\#_{ab}-\delta , \beta)) = b
    \end{equation}
    is satisfied, then due to the change in relative ordering of Q-values, 
    \begin{equation}
        \pi((s,d,e,t);(\gamma^\#_{ab} - \delta, \beta)) \neq \pi((s,d,e,t); (\gamma^\#_{ab} + \delta, \beta)),
    \end{equation}
    which is policy critical at $\gamma = \gamma_{ab}^\#$. 
\end{proof}
\pagebreak
\section{Simulation Configurations} \label{sect:experimental-config}
This section outlines the simulation configurations discussed in Sections \bref{sect:size-invariant-method}{4} to \bref{sect:size-invariant-exp}{5}.

\begin{table}[H]
\centering
\scriptsize
\caption{
Default simulation configurations for QRL training and dataset splits. When simulating with RL at Depth 1, we use the Analytical Expression (see Equation \bref{eq:eqc-qval}{2}) by implementing it classically in PyTorch \cite{torch2017}. This results in a more efficient classical simulation, as there is no need to compute the quantum state vector. If a PennyLane simulator is used, then it would be the "lightning.qubit" simulator. The second part of the table consists of essential compute specifications for our simulations. We used a GPU for simulations with a Depth 1 EQC (written classically), but a CPU with large memory for simulations with a EQC Depth $>1$. 
}
\label{tab:rl-config}
\begin{tabular}{ll|ll}
\hline
\multicolumn{2}{c}{\textbf{Configuration}}  & \multicolumn{2}{c}{\textbf{Configuration}}   \\
\hline
Max. Episodes & 20000 & Learning Rate & $10^{-3}$ \\
Initial $\gamma_1\cdots \gamma_p$, $\beta_1\cdots \beta_p$ & 1.0 & Memory Buffer Size & 10000 \\
Early stopping patience & 50 updates (500 eps) & PennyLane Simulator & Lightning.Qubit* \\
Early Stopping Min Decrease & -0.001 & Discount Rate & 1 (No Discounting) \\
Train step interval & 10 episodes & $\epsilon$-greedy initialization & 1.00  \\
Target Network update & 30 episodes & $\epsilon$-greedy decay factor & 0.999 \\
Batch Size & 10 & $\epsilon$-greedy minimum & 0.01  \\
\midrule
EQC Depth 1 CPU & Intel Core i9-13000KF & System RAM &  64 GB \\
EQC Depth 1 GPU &  NVIDIA-RTX 4090 & vRAM & 32 GB\\
EQC Depth $>1$ CPU & Intel Xeon Gold 5222 & System RAM & 754 GB \\ 
EQC Depth $>1$ GPU & Not Used & vRAM & Not Applicable \\
\bottomrule
\hline
\end{tabular}
\end{table}
\subsection{Datasets}\label{sect:experimental-config-datasets}
The datasets consist of two-dimensional Euclidean TSP instances, where node coordinates are sampled uniformly from $[0,1]^2$. For TSP sizes below 15, optimal tours are obtained using the Held-Karp Dynamic Programming algorithm \cite{5a5ec1f2-057b-31e3-9132-41b3d9faf61c}; for larger sizes, we adopt the best of 20 runs of the Python implementation \cite{elkai-2023} of LKH-3\cite{lkh3-2017}. 

Each dataset (train, validation, test) is generated from independent random seeds and contains instances of a single fixed TSP size. Unless otherwise stated, the training set comprises 500 instances; validation set contains 100 instances for sizes 5-15, 50 instances for sizes 16-30 and 30 instances for sizes 31-350; and the test set comprises of 1000 instances.


\subsection{QRL Training}\label{sect:experimental-config-qrl}
Unless otherwise specified, QRL simulations follow the training procedure of Skolik et al. (2023)\cite{Skolik2022EquivariantQC}, with the addition of a validation procedure every parameter update. Every 10 episodes, a training step is ran by randomly sampling a minibatch 10 transitions from the replay buffer with replacement. This updates the parameters of the behaviour policy. We set a hard limit of 20000 episodes, but training is early stopped where the mean vaidation gap fails to decrease for 50 consecutive parameter updates (equivalent to 500 episodes, since one training step is performed every 10 episodes). 

This training strategy is used whenever we refer to "RL" results in this paper, and it is also applied in all simulations with EQC having a depth greater than one. The full simulation configurations can be found in Table \ref{tab:rl-config}. 

\pagebreak
\section{Agent's Policy Visualization at More Positions in the RL Tour} \label{appendix:agent-vis-additional}
Figure \ref{fig:tsp10-visualization-additional} and \ref{fig:tsp100-visualization-additional} shows the policy visualizations of TSP-10 and TSP-100 at Position 1 (initial position) and at positions where 25\%, 50\% and 75\% of the tour has occurred. Across different positions, both the Qvalue-Distance relationship (top row) and the frequency of selecting the $x$-th nearest neighbor (bottom row) exhibit a consistent distribution over different tour positions.
\begin{figure}[H]
    \centering
    \begin{subfigure}[b]{0.24\linewidth}
        \centering
        \includegraphics[width=\linewidth]{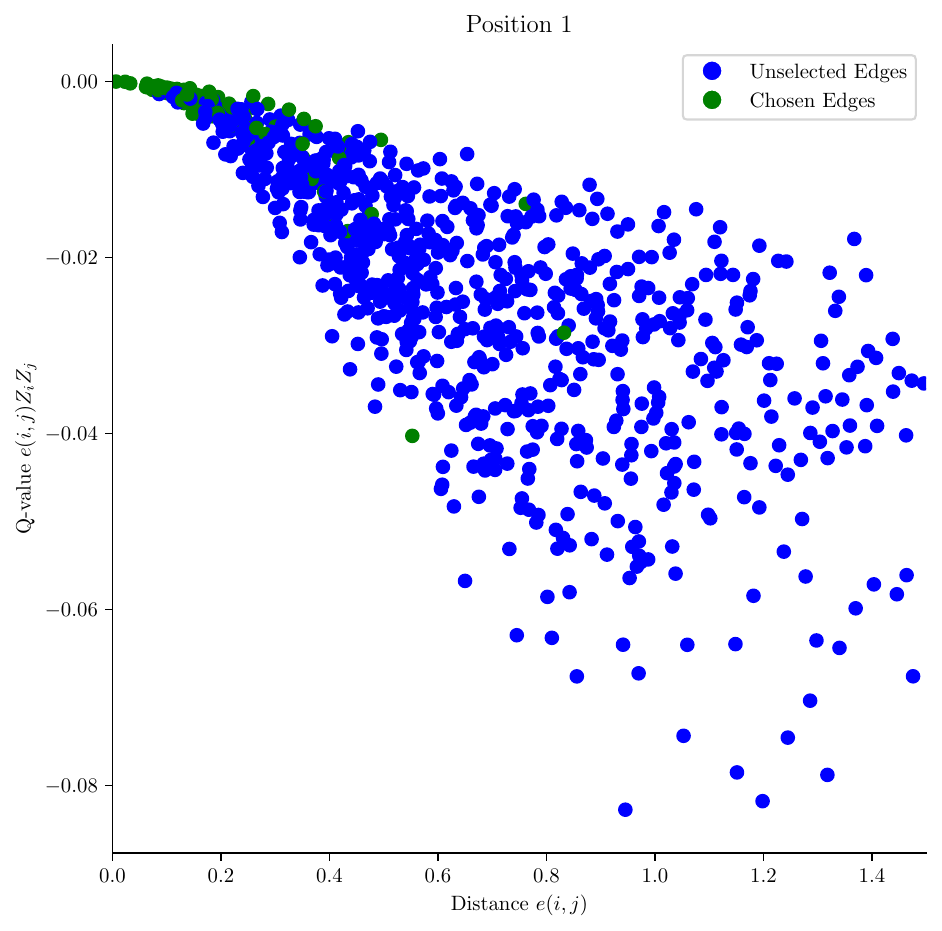}
    \end{subfigure}
    \begin{subfigure}[b]{0.24\linewidth}
        \centering
        \includegraphics[width=\linewidth]{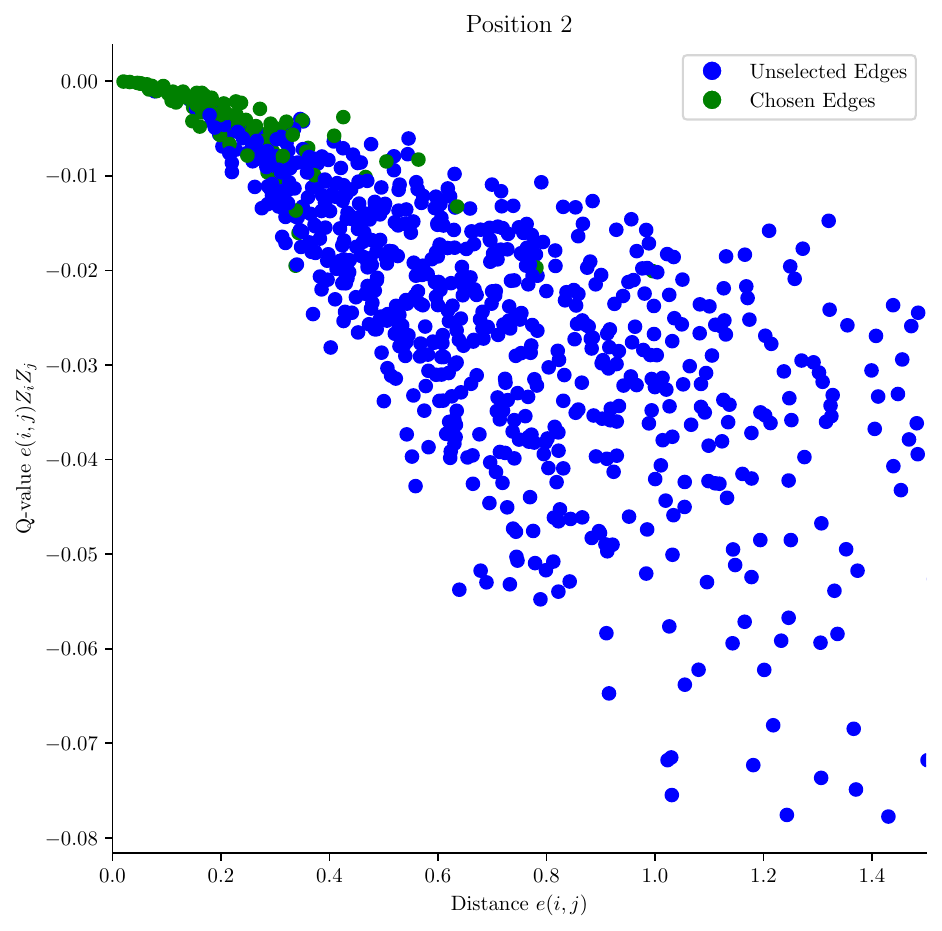}
    \end{subfigure}
    \begin{subfigure}[b]{0.24\linewidth}
        \centering
        \includegraphics[width=\linewidth]{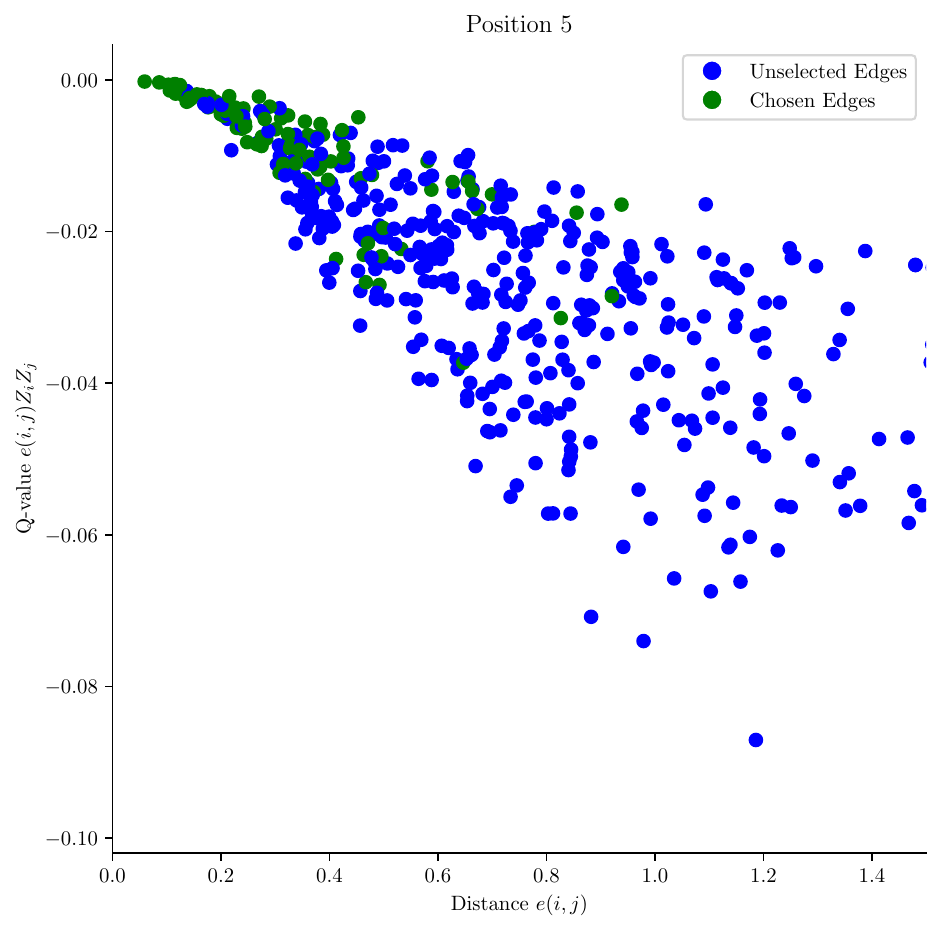}
    \end{subfigure}
    \begin{subfigure}[b]{0.24\linewidth}
        \centering
        \includegraphics[width=\linewidth]{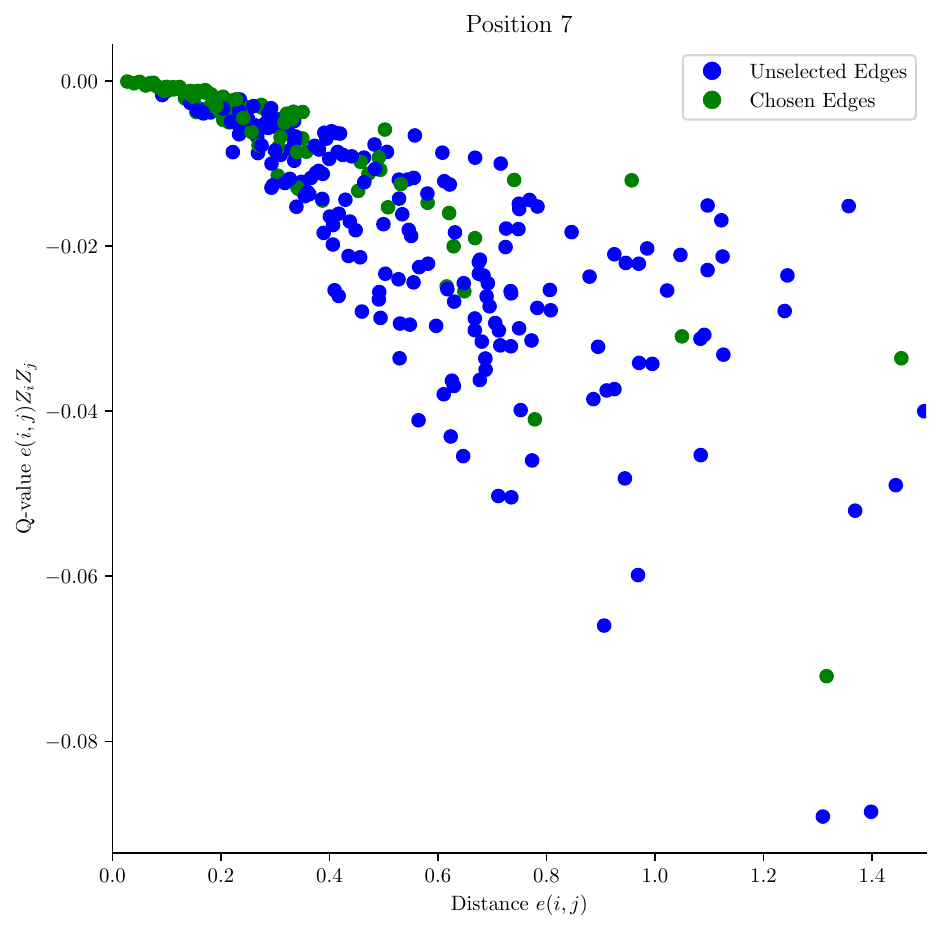}
    \end{subfigure}
    \hfill
    \begin{subfigure}[b]{0.24\linewidth}
        \centering
        \includegraphics[width=\linewidth]{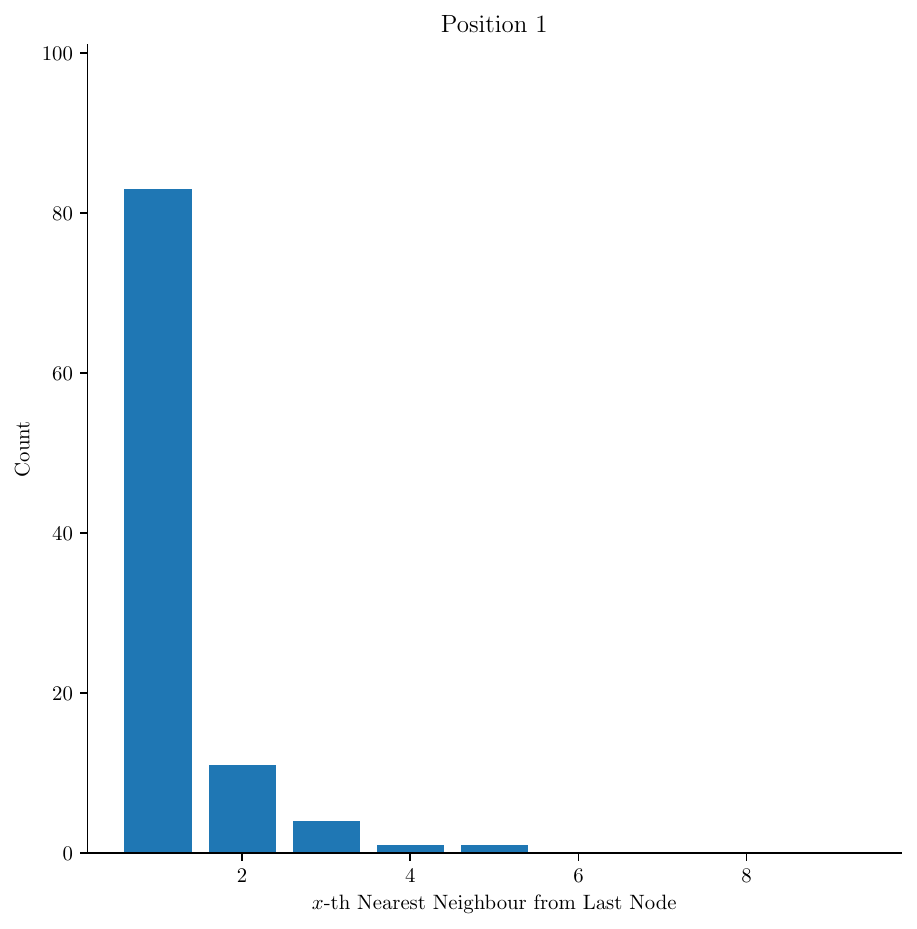}
        \caption{Tour Position 1}
    \end{subfigure}
    \begin{subfigure}[b]{0.24\linewidth}
        \centering
        \includegraphics[width=\linewidth]{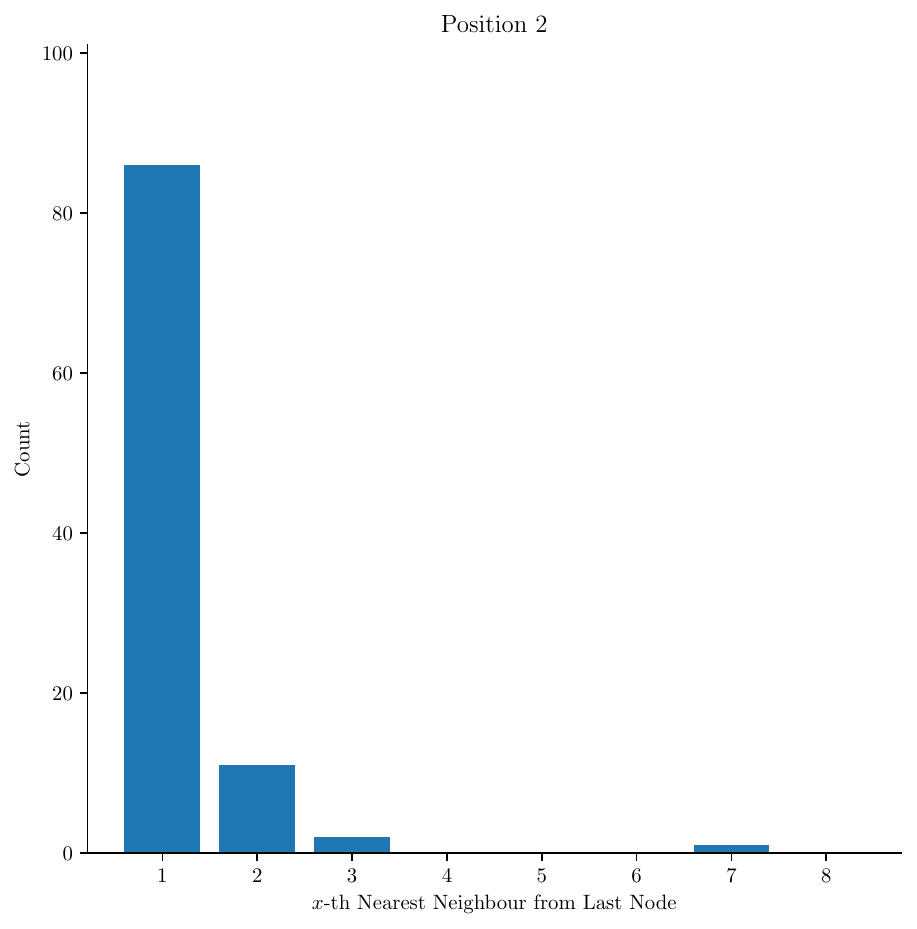}
        \caption{Tour Position 2}
    \end{subfigure}
    \begin{subfigure}[b]{0.24\linewidth}
        \centering
        \includegraphics[width=\linewidth]{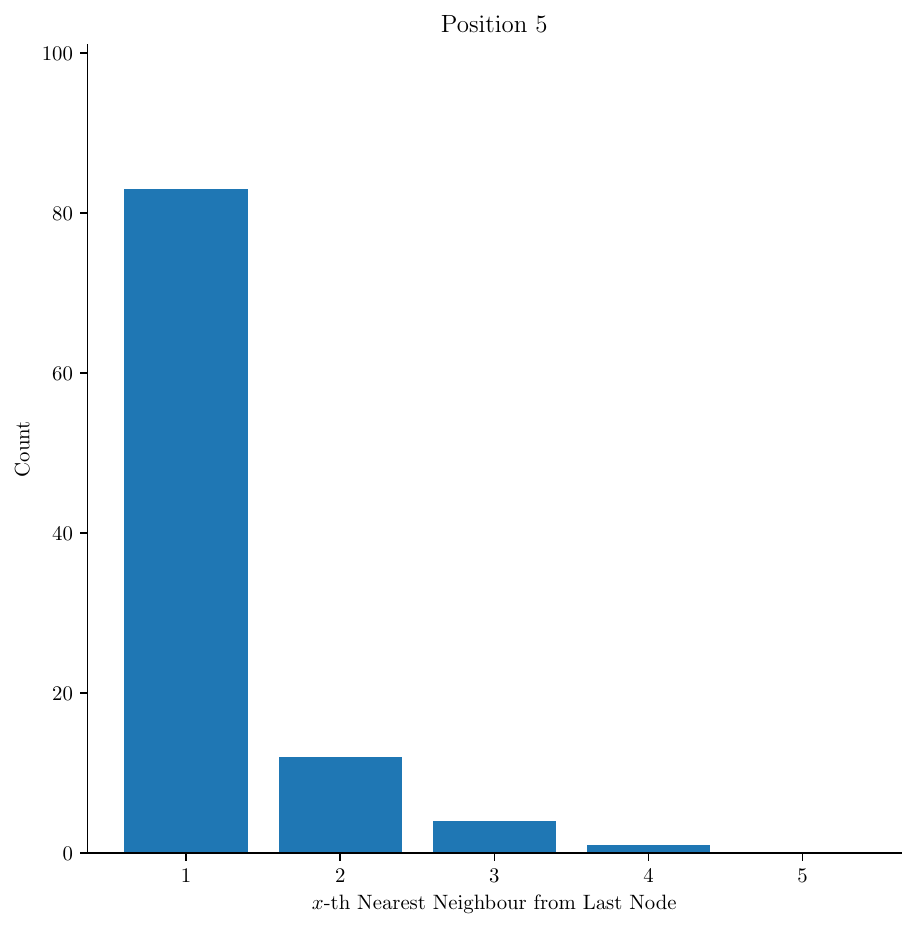}
        \caption{Tour Position 5}
    \end{subfigure}
    \begin{subfigure}[b]{0.24\linewidth}
        \centering
        \includegraphics[width=\linewidth]{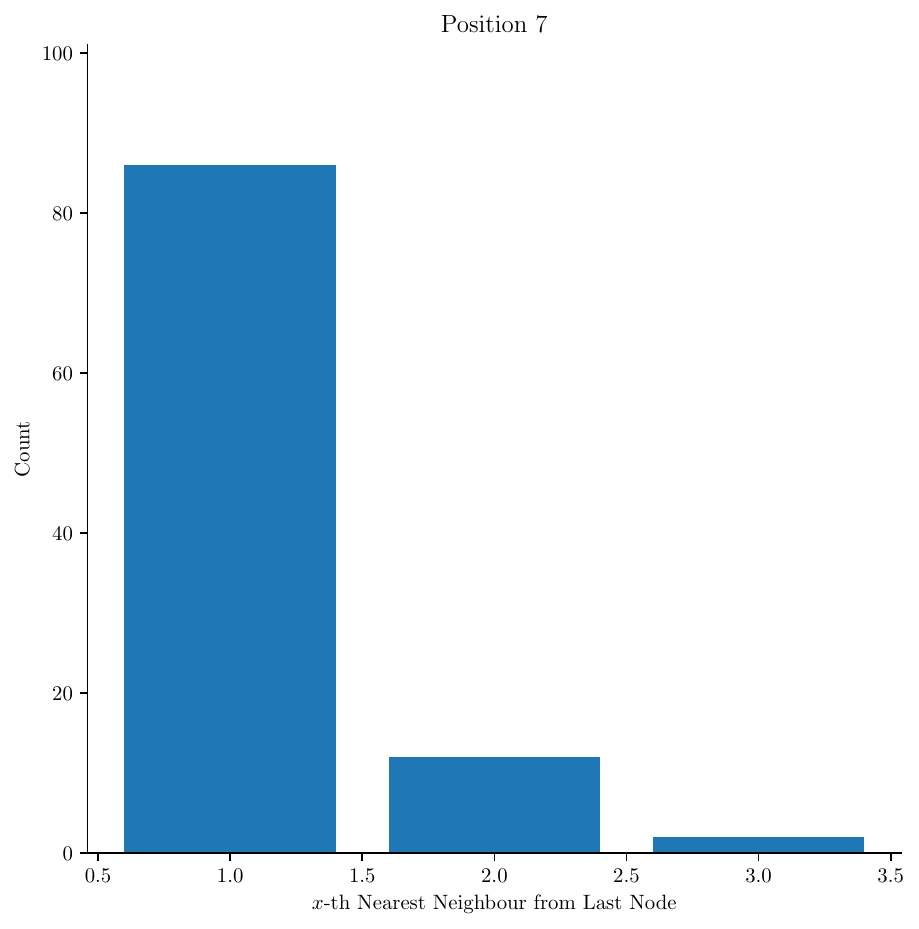}
        \caption{Tour Position 7}
    \end{subfigure}
    \caption{Visualization of TSP-10 policies using Q-value against Distance plots (top row) and Frequency of selecting the $x$-th Nearest Neighbour (bottom) at different positions of the tour. These plots are generated using the $\gamma$ found using the SIGS method for TSP-10, $\gamma= 0.9$, and $\beta = 1.1$.}
    \label{fig:tsp10-visualization-additional}
\end{figure}
\begin{figure}[H]
    \centering
    \begin{subfigure}[b]{0.24\linewidth}
        \centering
        \includegraphics[width=\linewidth]{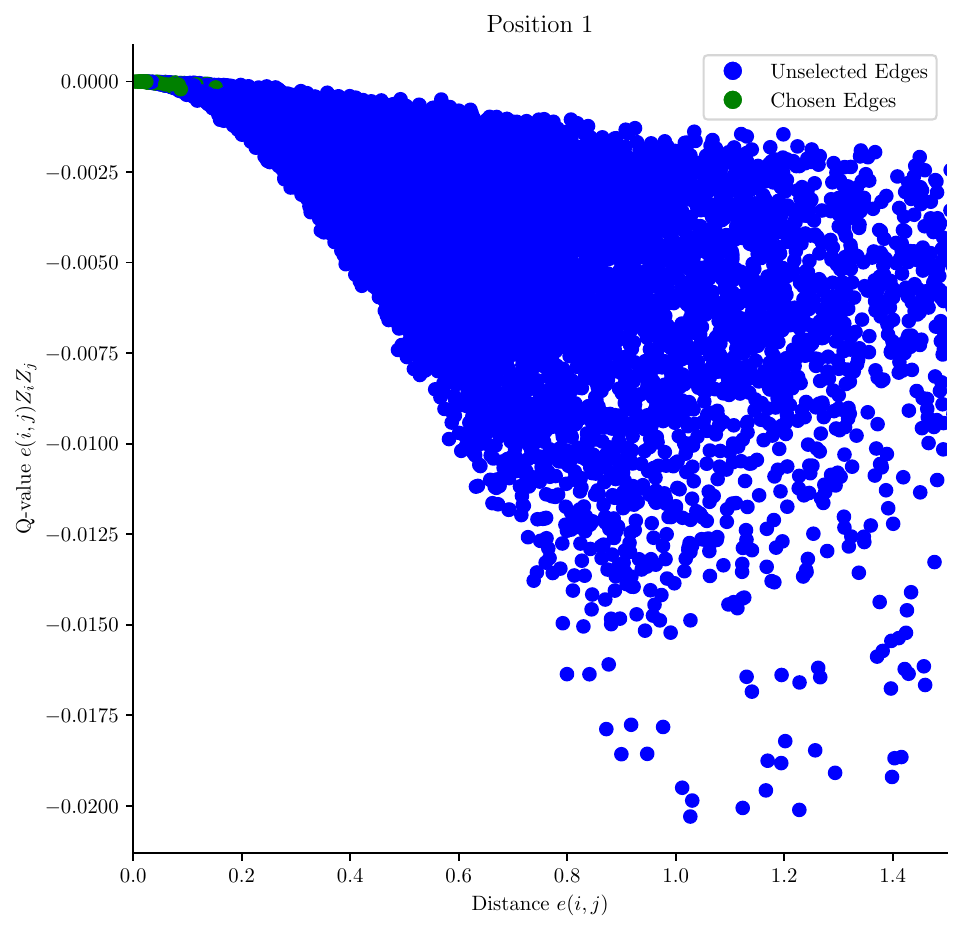}
    \end{subfigure}
    \begin{subfigure}[b]{0.24\linewidth}
        \centering
        \includegraphics[width=\linewidth]{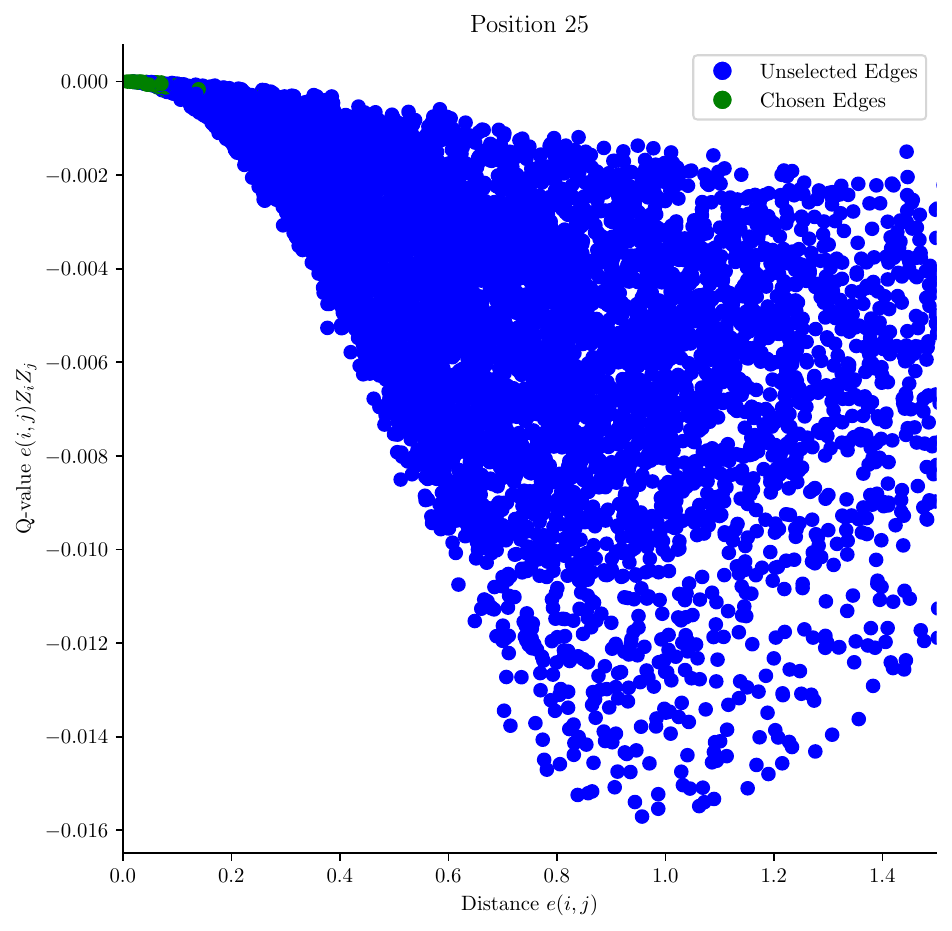}
    \end{subfigure}
    \begin{subfigure}[b]{0.24\linewidth}
        \centering
        \includegraphics[width=\linewidth]{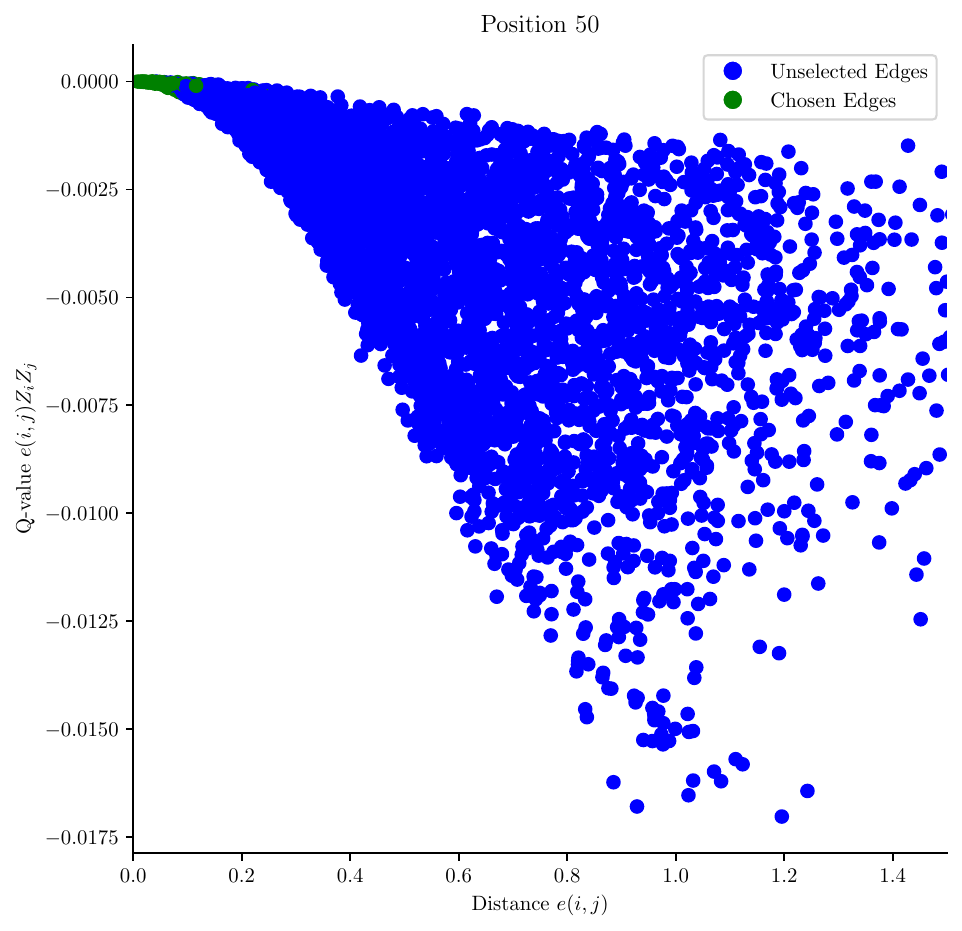}
    \end{subfigure}
    \begin{subfigure}[b]{0.24\linewidth}
        \centering
        \includegraphics[width=\linewidth]{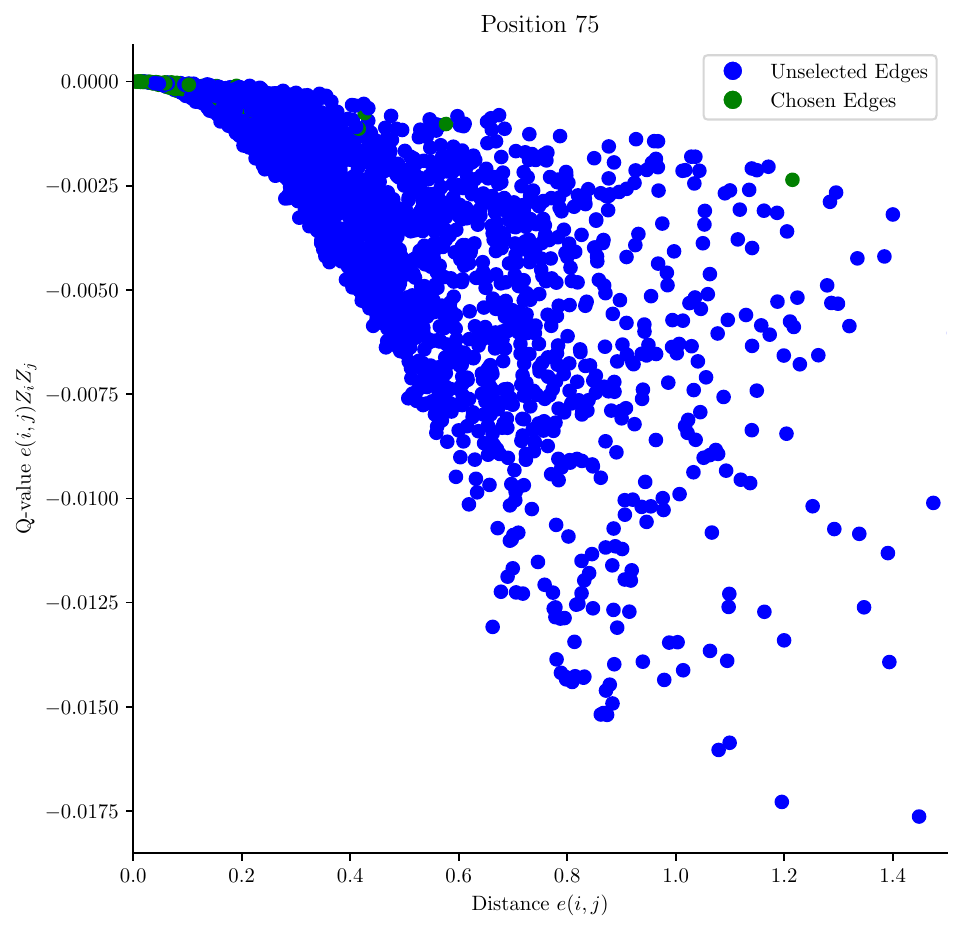}
    \end{subfigure}
    \hfill
    \begin{subfigure}[b]{0.24\linewidth}
        \centering
        \includegraphics[width=\linewidth]{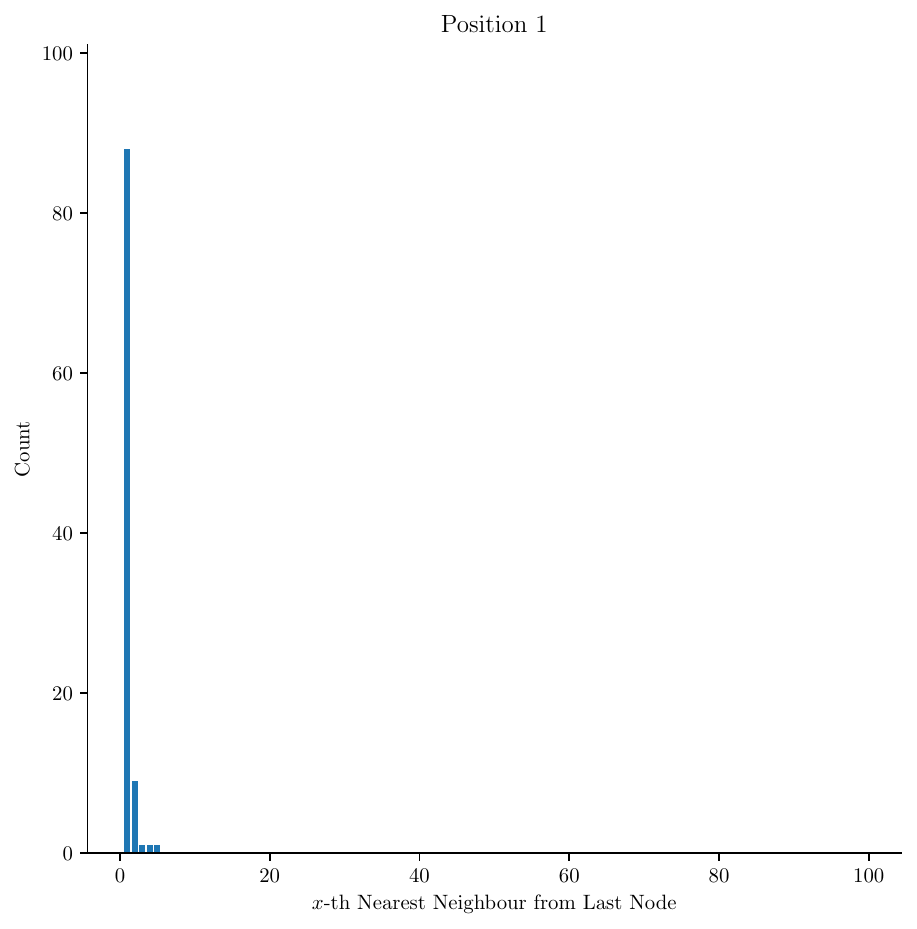}
        \caption{Tour Position 1}
    \end{subfigure}
    \begin{subfigure}[b]{0.24\linewidth}
        \centering
        \includegraphics[width=\linewidth]{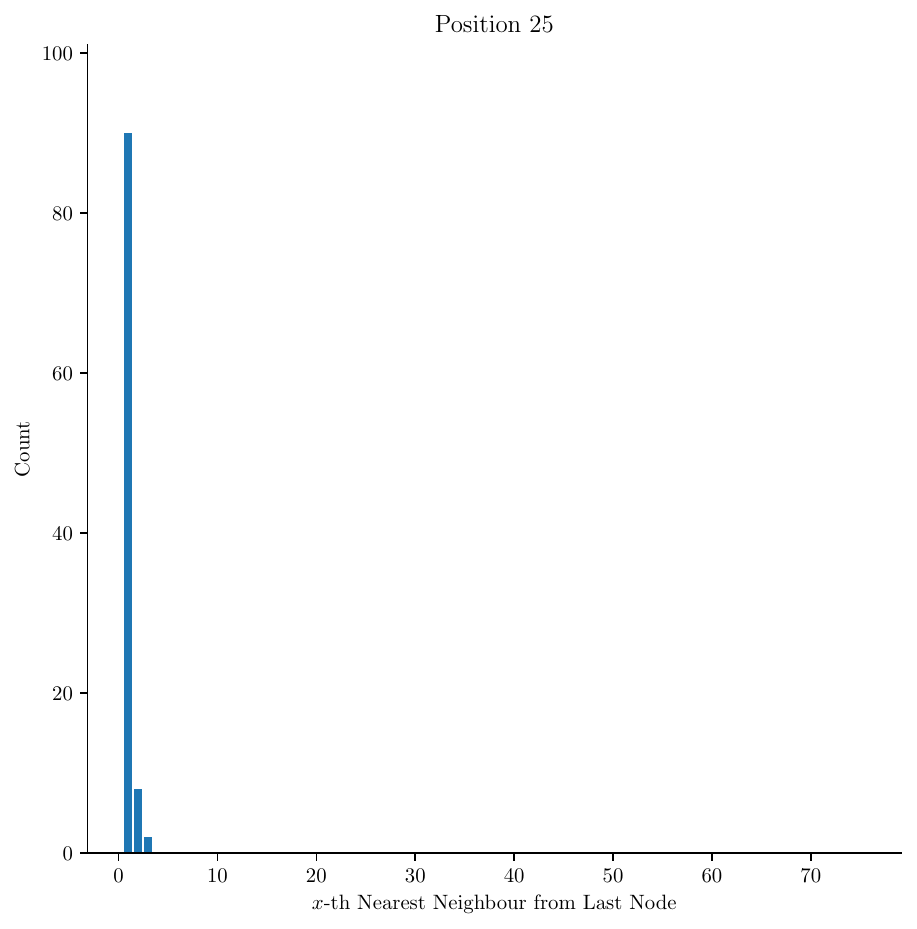}
        \caption{Tour Position 25}
    \end{subfigure}
    \begin{subfigure}[b]{0.24\linewidth}
        \centering
        \includegraphics[width=\linewidth]{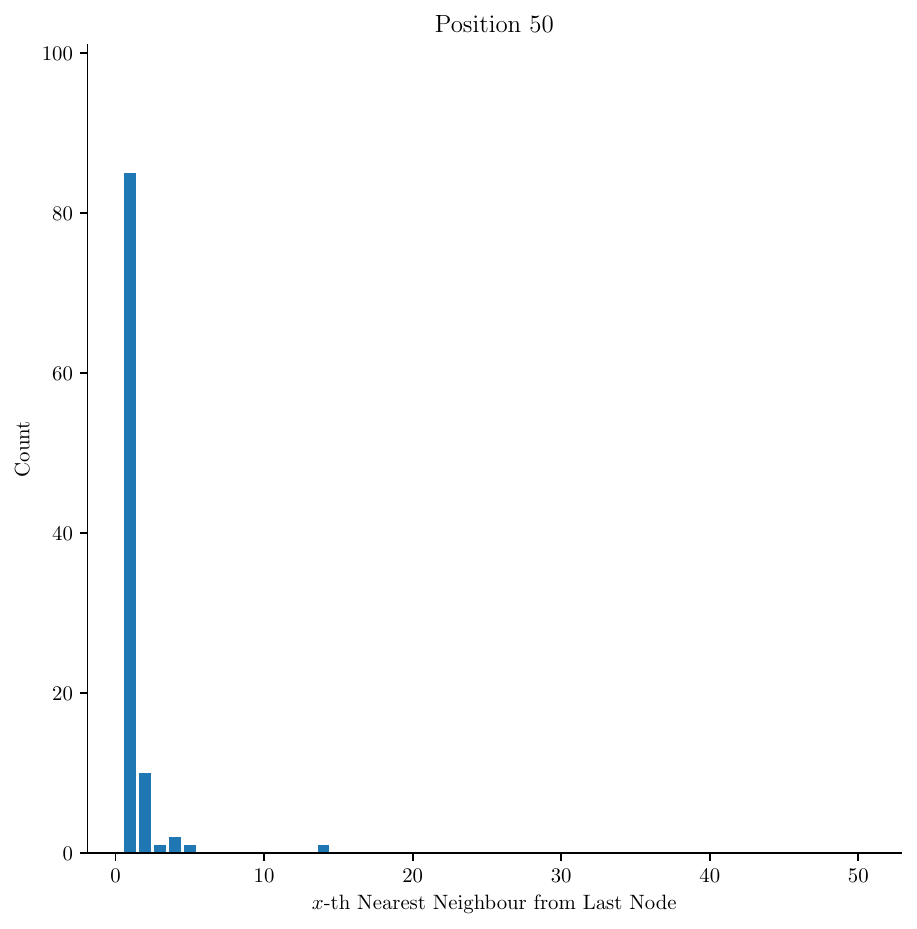}
        \caption{Tour Position 50}
    \end{subfigure}
    \begin{subfigure}[b]{0.24\linewidth}
        \centering
        \includegraphics[width=\linewidth]{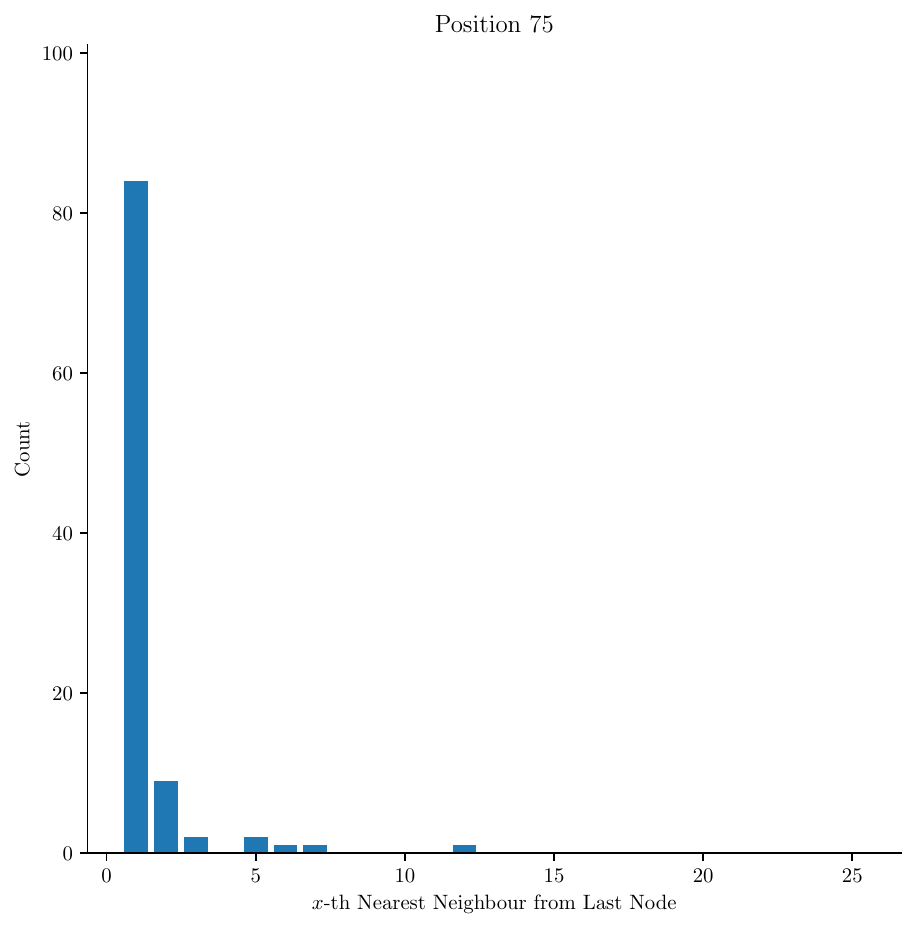}
        \caption{Tour Position 75}
    \end{subfigure}
    \caption{Visualization of TSP-100 policies using Q-value against Distance plots (top row) and Frequency of selecting the $x$-th Nearest Neighbour (bottom) at different positions of the tour. These plots are generated using the $\gamma$ found using the SIGS method for TSP-100, $\gamma= 0.4$, and $\beta = 1.1$.}
    \label{fig:tsp100-visualization-additional}
\end{figure}
\section{Full Comparison of RL and SIGS } \label{appendix:size-invariant-tm-full-comp}
\begin{table}[H]
\centering
\caption{Full comparison between RL at EQC Depth 1 and SIGS. Gaps reported are based on a held out test set for each TSP size, and additionally include the Worst Gap (Equation \bref{eq:mean-and-worst-gap}{5}), Best Gap (Lowest gap achieved on a single instance in a Dataset) and the number of episodes required to reach early stopping (Episodes). Additionally, we provide a breakdown of total time (that was reflected in Figure \bref{fig:rl-vs-SIGS}{3}). Episode Collection Time (Ep Time) refers to the total amount of time the agent spends in collecting transitions for train steps. This is dependent on the number of episodes the agent has ran for before early stopping. Train Step time (TS Time) refers to the total amount of time we spent on computing gradient updates, which is dependent on the number of episodes before early stopping, since $\frac{\text{episodes}}{10}$ is the number of train steps (as outlined in Supplementary Information \ref{sect:experimental-config-qrl}). Validation Time (Val Time) refers to the total amount of time we spent on validating the agent's performance (after every train step the agent is evaluated). Test Time refers to the total amount of time spent on computing the performance of the chosen model on the test set. For SIGS, there is no episode collection and train step calculations, therefore the total time breakdown only consists of the time for validation and testing. }
\scriptsize
\begin{tabular}{lccccccccccc}
\toprule
\textbf{TSP Size} & \textbf{5} & \textbf{10} & \textbf{20} & \textbf{30} & \textbf{40} & \textbf{50} & \textbf{60} & \textbf{70} & \textbf{80} & \textbf{90} &\textbf{100} \\
\midrule
\multicolumn{12}{l}{\textbf{RL} (EQC Depth 1 - PennyLane Simulator)}\\
\multicolumn{12}{l}{\textit{Test Gaps}}\\
Mean & 1.032 & 1.056 & 1.104 & - & - & - & - & - & - & - & -\\
Worst & 1.280 & 1.335 & 1.370 & - & - & - & - & - & - & - & -\\
Best & 1.000 & 1.000 & 1.000 & - & - & - & - & - & - & - & -\\
Episodes & 500 & 820 & 770 & - & - & - & - & - & - & - & -\\
\multicolumn{12}{l}{\textit{Wall Times (min)}}\\
\underline{Total} & 10.8 & 175 & 8980 & - & - & - & - & - & - & - & -\\
Ep & 0.568 & 5.08 & 3720 & - & - & - & - & - & - & - & -\\
TS & 4.78 & 99.8 & 2720 & - & - & - & - & - & - & - & -\\
Val & 4.39 & 61.7 & 2020 & - & - & - & - & - & - & - & -\\
Test & 1.00 & 8.59 & 517 & - & - & - & - & - & - & - & -\\
\midrule
\multicolumn{12}{l}{\textbf{RL} (EQC Depth 1 - Analytical Expression)}\\
\multicolumn{12}{l}{\textit{Test Gaps}}\\
Mean $G_1$ & 1.032 & 1.057 & 1.105 & \underline{1.135} & 1.150 & \underline{1.160} & 1.174 & 1.179 & \underline{1.182} & 1.188 & 1.189 \\
Worst & 1.280 & 1.335 & 1.370 & 1.467 & 1.392 & 1.458 & 1.382 & 1.389 & 1.443 & 1.354 & 1.382 \\
Best & 1.000 & 1.000 & 1.000 & 1.000 & 1.000 & 1.016 & 1.005 & 1.021 & 1.039 & 1.024 & 1.064 \\
Episodes & 500 & 810 & 770 & 3050 & 2690 & 4320 & 5000 & 4290 & 4700 & 5200 & 4670 \\
\multicolumn{12}{l}{\textit{Wall Times (min)}}\\
\underline{Total} & 0.788 & 5.63 & 7.42 & 42.5 & 36.7 & 68.5 & 93.9 & 97.1 & 119 & 146 & 151 \\
Ep  & 0.0370 & 0.298 & 0.630 & 7.24 & 8.54 & 18.7 & 26.6 & 26.6 & 33.7 & 42.3 & 42.7 \\
TS & 0.0530 & 0.111 & 0.117 & 0.504 & 0.458 & 0.701 & 0.801 & 0.687 & 0.732 & 0.871 & 0.792 \\
Val  & 0.537 & 4.36 & 5.06 & 32.8 & 24.6 & 45.9 & 62.8 & 65.7 & 80.2 & 98.0 & 102 \\
Test & 0.143 & 0.830 & 1.58 & 1.90 & 3.00 & 3.03 & 3.49 & 3.97 & 4.24 & 4.78 & 5.06 \\
\midrule
\multicolumn{12}{l}{\textbf{SIGS} (EQC Depth 1)} \\
\multicolumn{12}{l}{\textit{Test Gaps}}\\
Mean, $G_2$ & \underline{1.018} & \underline{1.046} & \underline{1.095} & 1.141 & \underline{1.149} & 1.161 & \underline{1.169} & \underline{1.177} & 1.183 & \underline{1.187} & \underline{1.187} \\
Worst & 1.233 & 1.316 & 1.354 & 1.415 & 1.361 & 1.367 & 1.410 & 1.405 & 1.388 & 1.368 & 1.388 \\
Best & 1.000 & 1.000 & 1.000 & 1.001 & 1.000 & 1.016 & 1.005 & 1.021 & 1.030 & 1.024 & 1.061 \\
\multicolumn{12}{l}{\textit{Wall Times (min)}}\\
\underline{Total} & 0.123 & 0.401 & 0.847 & 0.903 & 1.419 & 1.070 & 2.855 & 3.676 & 4.562 & 4.648 & 5.433 \\
Val  & 0.080 & 0.108 & 0.407 & 0.415 & 0.669 & 0.362 & 0.956 & 1.236 & 1.583 & 1.654 & 1.817 \\
Test & 0.041 & 0.063 & 0.440 & 0.488 & 0.750 & 0.708 & 1.896 & 2.433 & 2.977 & 2.991 & 3.613 \\
\midrule
\multicolumn{12}{l}{$\Delta = |G_1-G_2|$}\\
$\Delta$ & 0.014 & 0.011 & 0.009 & 0.006 & 0.002 & 0.001 & 0.005 & 0.003 & 0.001 & 0.001 & 0.002 \\
\bottomrule
\end{tabular}
\label{tab:rl-vs-SIGS-full}
\end{table}
We observe that \textit{Train Step} (TS) time dominates runtime as TSP-size grows, when using the PennyLane simulator for QRL. This is due to the quantum circuit simulator having to compute, and take gradients of, the entire quantum state vector $|\psi\rangle\in \mathbb{C}^{2^n}$ and unitary ${\hat U}\in \mathbb{C}^{2^n \times 2^n}$ (where $n$ is the number of nodes). This can be avoided if we simulate the EQC using Analytical Expression. This way the largest matrix we use in our calculations is $\mathbb{R}^{|E|\times|E|}$, where $|E|$ is the number of edges, which is $\mathbb{R}^{{\cal O}(n^4)}$, which is much more efficient in terms of computation time and memory. The following time reduction is also reflected in the TS Time for TSP-20 (2720 min - on PennyLane Simulator to 0.111 min using the Analytical Expression, where both runs ran for 770 episodes). 

We also observe that the \textit{Worst gap on the test set} (Worst) is approximately $20\%-30\%$ higher than the \textit{Mean gap on the test set} (Mean). We suspect that they may be a certain type of node placements where the optimal EQC policies are poor at. Further exploration into this is beyond the scope of this work.
\section{Additional visualizations showing the effect of increasing TSP size at the same value of \texorpdfstring{$\gamma$}{gamma}}\label{appendix:gamma-constant}
\begin{figure}[H]
    \centering
    \begin{subfigure}[b]{0.32\textwidth}
        \centering
        \includegraphics[width=\textwidth]{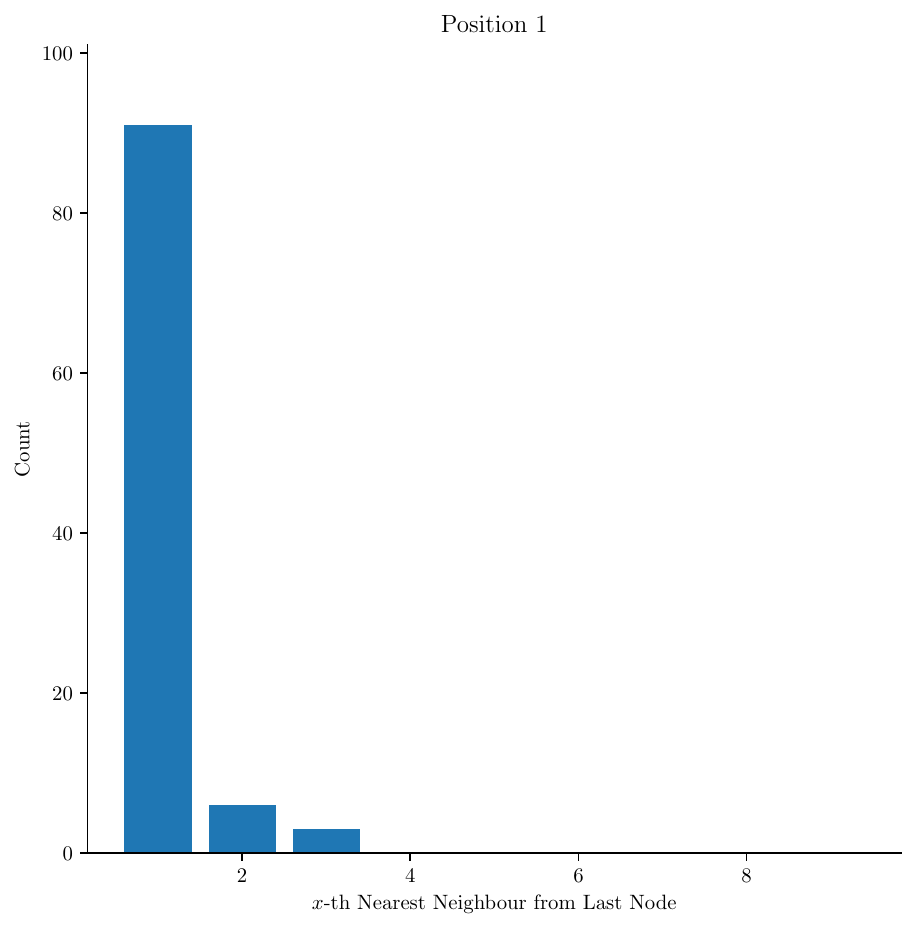}
        \caption{TSP-10, using $\gamma^*_{\text{TSP-10}} = 0.9$}
        \label{fig:gamma-comp-multi-size-tsp10}
    \end{subfigure}
    \begin{subfigure}[b]{0.32\textwidth}
        \centering
        \includegraphics[width=\textwidth]{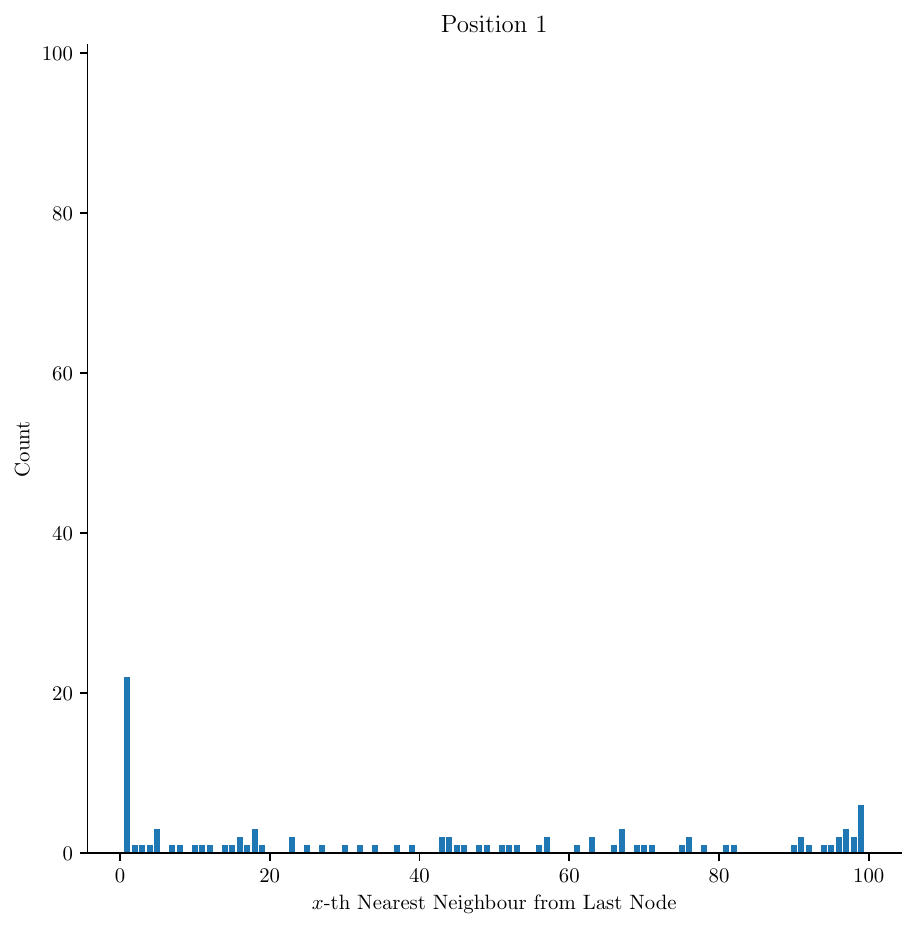}
        \caption{TSP-100, using  $\gamma^*_{\text{TSP-10}} = 0.9$}
        \label{fig:gamma-comp-multi-size-tsp100-g0.9}
    \end{subfigure}
    \begin{subfigure}[b]{0.32\textwidth}
        \centering
        \includegraphics[width=\textwidth]{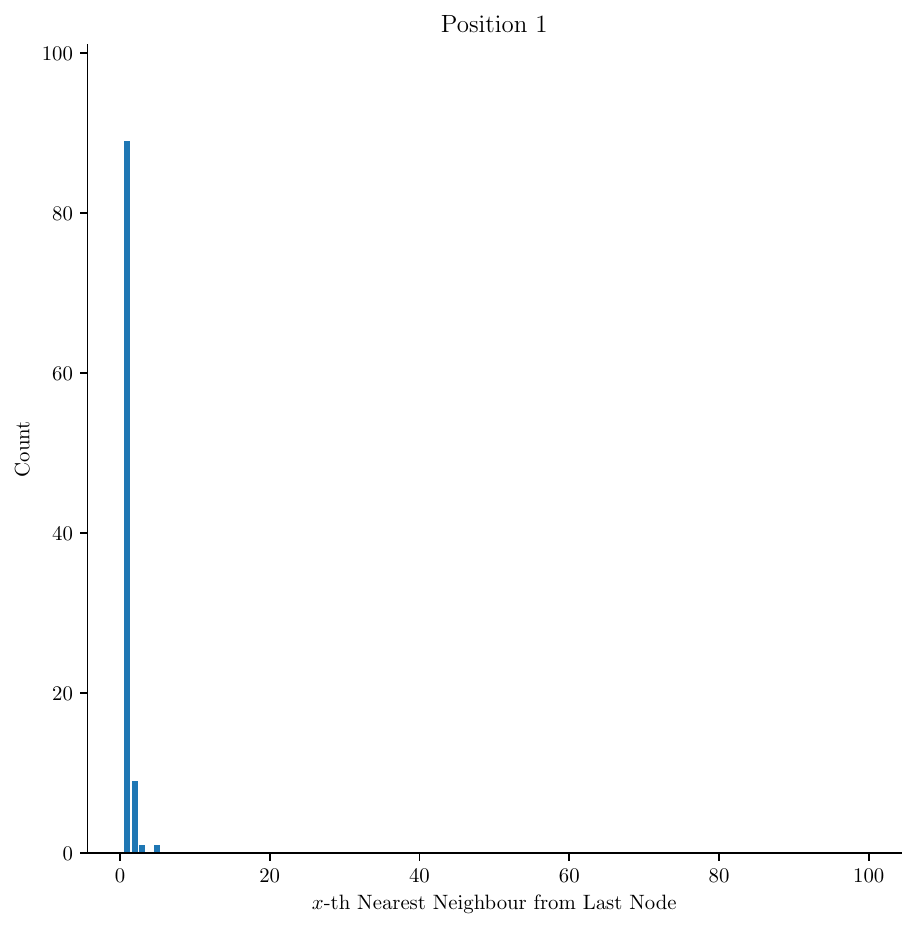}
        \caption{TSP-100 using $\gamma^*_{\text{TSP-100}} = 0.4$}
        \label{fig:gamma-comp-multi-size-tsp100-g0.4}
    \end{subfigure}
    \caption{Effect of using the best $\gamma$ found for TSP-10 to evaluate on TSP-100 instances on the Frequency of selecting the $x$-th Nearest Neighbour at Position 1 of the Tour. A detailed description of these histogram can be found in the description of Figure 7. Let $\gamma^*_{\text{TSP-}x}$ denote the best $\gamma$ found via the SIGS procedure on TSP instances of $x$ nodes. These plots were generated on the validation set for TSP-10 and a randomly generated dataset of 100 uniform TSP instances for TSP-100 (for TSP-100 we generated a new dataset due to the small size of the validation set used for training).}
    \label{fig:gamma-comp-multi-size}
\end{figure}
Figure \ref{fig:gamma-comp-multi-size} shows the effect of using the best $\gamma$ found for TSP-10 ($\gamma^*_{\text{TSP-10}} = 0.9$) to evaluate TSP-100 instances. Figure \ref{fig:gamma-comp-multi-size-tsp10} demonstrates that for TSP-10, the best policy selects the nearest neighbour most of the time, followed by the second, and third nearest neighbours at Position 1 of the tour. When the same $\gamma^*_{\text{TSP-10}}$ is used on TSP-100 instances (Figure \ref{fig:gamma-comp-multi-size-tsp100-g0.9}), the agent selects the nearest neighbour only about 20\% of the time and instead chooses many other neighbours. This behavior arises from the increased weight of the additional cosine terms as TSP size grows (see Conjecture \bref{conjecture:opt-gamma-decreases}{1}). 

For comparison, Figure \ref{fig:gamma-comp-multi-size-tsp100-g0.4} shows the behavior under the best $\gamma$ found for TSP-100 ($\gamma^*_{\text{TSP-100}}=0.4$). Here the agent's policy once again focuses only on the nearest neighbors, producing a similar distribution to TSP-10. This illustrates why the optimal $\gamma^*$ decreases with problem size, as first noted in Section \bref{sect:size-invariant-method-exp-scale-tsp350}{4.3}).  

\section{SIGS: Error Bound on Gridsearch Discretization Error} \label{appendix:size-invariant-method-discretization-error}
\begin{definition}
    \textit{(Mean Gap Function)} Let $f(\gamma)$ denote the mean gap achieved by a Depth 1 EQC with a parameter setting $\gamma \in (0, \gamma_{\max})$.
    \label{def:mean-gap-f}
\end{definition}
Note that $\gamma_{\max}$ was defined as $\gamma_{\max} \coloneq \frac{\pi/2}{\max_{i< j}e_{ij}}$ (see Lemma \bref{lem:safety-of-gamma-range}{2}).
\begin{definition}
    \textit{(True Minimizer of $f$)} Let $\gamma^*$ denote the true minimizer of $f$, where $\gamma^*\in(0, \gamma_{\max})$. 
    \label{def:true-minimizer}
\end{definition}
\begin{definition}
     \textit{(Gridsearch Boundary)} Let $\Gamma$ denote the a uniform grid over the interval $\gamma \in (0, \frac{\pi/2}{\tan^{-1}\sqrt{2}})$ with a spacing of $\Delta\gamma$. Define $\gamma_-$ and $\gamma_+$ such that
     \begin{align}
     \gamma_- & = \max\{ \gamma \in \Gamma : \gamma \leq \gamma^*\}\\
     \gamma_+ & = \gamma_- + \Delta\gamma,
     \end{align}
     which ensures that $\gamma^*\in[\gamma_-,\gamma_+]$. 
    \label{def:gamma-plus-minus}
\end{definition}
\begin{definition}
    (Discretization Error) Let $\gamma^*_{SIGS}$ be the value of $\gamma$ chosen using SIGS, and let $\gamma^*$ be defined as per Definition \ref{def:true-minimizer}. Then define the discretization error from running SIGS with a step size $\Delta\gamma$ be defined as 
    \begin{equation}
        \varepsilon(\Delta\gamma) = f(\gamma^*_{SIGS} ) - f(\gamma^*)
    \end{equation}
    \label{def:discretization-error}
\end{definition}
\begin{proposition}
    \textit{(Error bound for the grid-search approximation)} The Discretization Error (see Definition \ref{def:discretization-error}) is no worse than the better of its two adjacent grid-points:
    \begin{equation}
     0 \leq \varepsilon(\Delta\gamma) \leq \min \{f(\gamma_-) - f(\gamma^*), f(\gamma_+) - f(\gamma^*)\}.
    \end{equation}    
    \label{prop:error-bound-on-gamma-gridsearch}
\end{proposition}

\begin{proof}
    By definition of grid-search, $\gamma^*_{SIGS}$ can be expressed as 
    \begin{equation}
        \gamma^*_{SIGS} = \arg \min_{\gamma \in \Gamma} f(\gamma).
        \label{eq:gamma-SIGSforce}
    \end{equation}
    As such, 
    \begin{equation}
        f(\gamma_{SIGS}^*) \leq \min\{f(\gamma_-), f(\gamma_+)\}.
    \end{equation}
    Since $\gamma^*$ is the global minimizer, and $f(\gamma^*)\leq f(\gamma_-), f(\gamma_+)$, and the local discretization error is bounded by: 
    \begin{equation}
        0 \leq f(\gamma_{SIGS}^*) - f(\gamma^*) \leq \min\{f(\gamma_-)-f(\gamma^*), f(\gamma_+) -f(\gamma^*)\},
    \end{equation}
    \begin{equation}
        \Rightarrow 0 \leq \varepsilon(\Delta\gamma) \leq \min\{f(\gamma_-)-f(\gamma^*), f(\gamma_+) -f(\gamma^*)\}.
    \end{equation}
\end{proof}
\begin{figure}[H]
    \centering
    \begin{subfigure}[b]{0.32\linewidth}
        \centering
        \includegraphics[width=\linewidth]{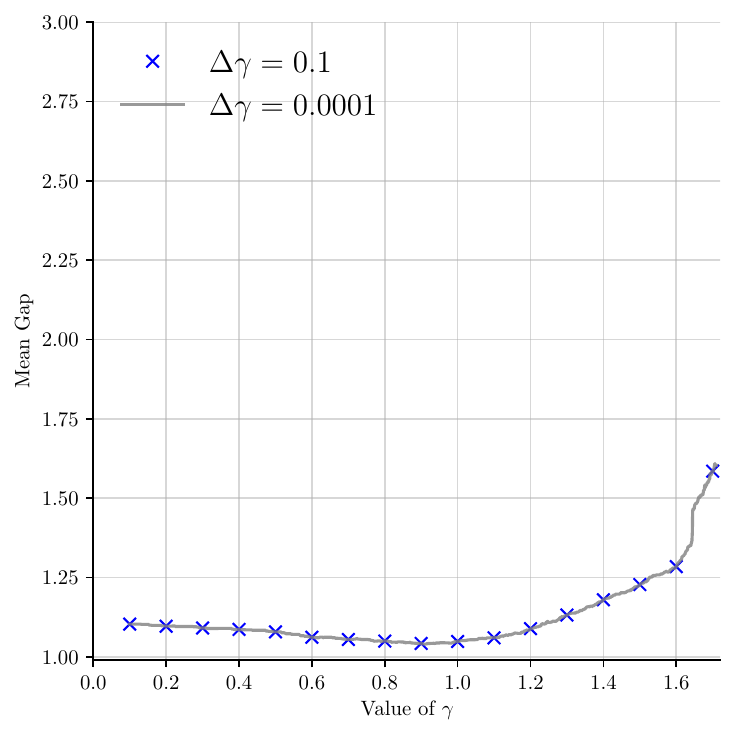}
        \caption{TSP-10}
        \label{fig:discretization-error-gamma-0.1-tsp10}
    \end{subfigure}
    \begin{subfigure}[b]{0.32\linewidth}
        \centering
        \includegraphics[width=\linewidth]{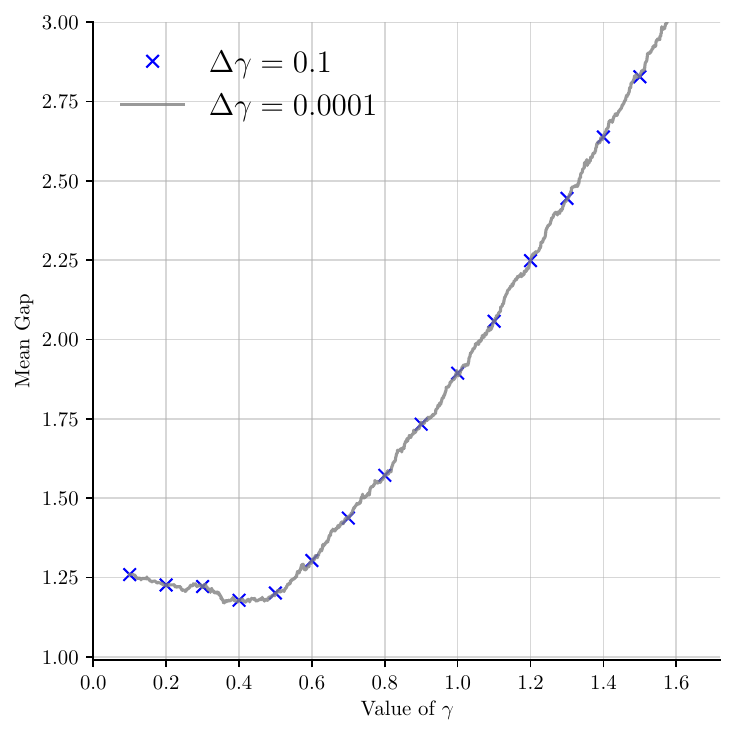}
        \caption{TSP-100}
        \label{fig:discretization-error-gamma-0.1-tsp100}
    \end{subfigure}
    \begin{subfigure}[b]{0.32\linewidth}
        \centering
        \includegraphics[width=\linewidth]{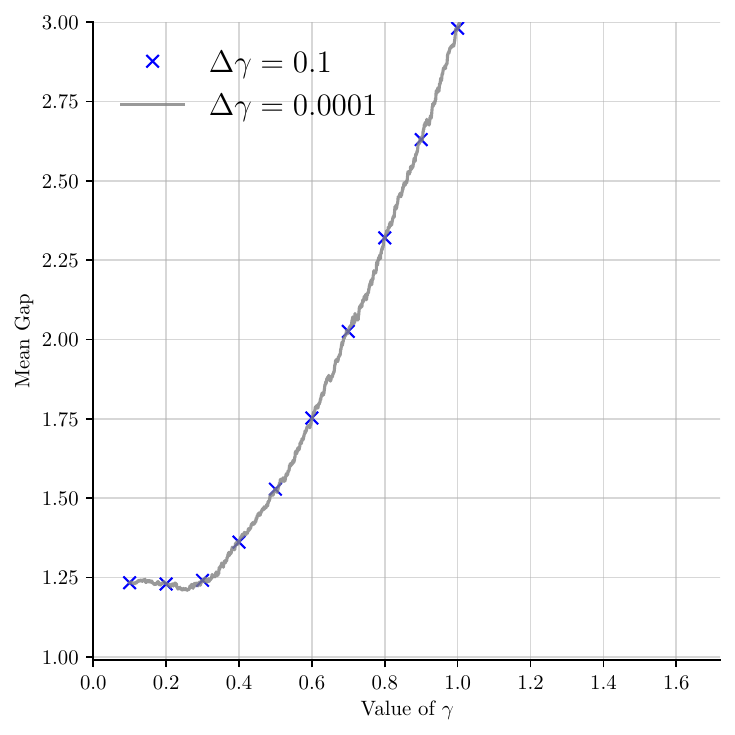}
        \caption{TSP-350}
        \label{fig:discretization-error-gamma-0.1-tsp350}
    \end{subfigure}
    \caption{Mean Gap as a function of $\gamma$, comparing a gridsearch with step-sizes $\Delta\gamma = 0.1$ (blue crosses) and $\Delta\gamma =0.0001$ (gray line). Note that the vertical axes of all plots are standardized. The Mean Gaps reported in the plot correspond to the mean gap of the validation set, since the validation set is used in SIGS to select parameter $\gamma$.}
    \label{fig:discretization-error-gamma-0.1}
\end{figure}
Proposition \ref{prop:error-bound-on-gamma-gridsearch} showed that the error introduced by gridsearch is bounded by the difference in mean gap values at the two nearest grid-points surrounding the true minimizer $\gamma^*$. This guarantee is only effective if the expression $\min\{f(\gamma_-) - f(\gamma^*), f(\gamma_+) - f(\gamma^*)\}$ is sufficiently small. To assess the practical impact of discretization, we empirically evaluate the error across different choices of $\Delta\gamma$ on TSP-10.

Table \ref{tab:discretization-error} reports results for $\Delta\gamma = 0.1, 0.01$ and $0.001$, using a finer grid $\Delta\gamma = 0.0001$ as a proxy for $\gamma^*$. Using this proxy, we determine $\gamma^* \approx \tilde{\gamma^*} = 0.9111$. At $\Delta\gamma = 0.1$, the grid-search solution is at most $0.00218$ worse in mean gap compared to $\gamma^*$, corresponding to a less than 1\% relative error, while requiring under one minute of runtime. This error increases slightly as TSP size increases: at TSP-100, we observe a discretization error of $0.00875$ and $0.0200$ at TSP-100 and TSP-350 respectively. Nevertheless, we find a discretization error at TSP-350 of $0.02$ (relative error $\approx$ 1.65\%) to be an acceptable tradeoff between wall clock efficiency and solution quality. 

Refining the step size to $\Delta\gamma=0.01$ incurs reduces the discretization error significantly at all 3 TSP-sizes: the relative error decreases by 0.15\%, 0.66\% and 1.47\% at TSP-10, TSP-100 and TSP-350 respectively. As such, SIGS could potentially be improved by doing a coarse-to-fine gridsearch: where we first use $\Delta\gamma = 0.1$ (coarse grid) to find the best $\gamma$ in the coarse grid, and search within the range of $[\gamma - \Delta\gamma, \gamma+\Delta\gamma]$ using a finer grid. We suspect that this would reduce the relative error further for simulations up to TSP-350.

However, we leave the coarse-to-fine gridsearch for future work as we suspect that eventually the coarse-grid $\Delta\gamma = 0.1$ may become irrelevant beyond TSP-350. For TSP-100 (Figure \ref{fig:discretization-error-gamma-0.1-tsp100}) and TSP-350 (Figure \ref{fig:discretization-error-gamma-0.1-tsp350}), beyond $\gamma^*$ ($\gamma^*_{\text{TSP-100}}=0.359$, $\gamma^*_{\text{TSP-350}} =0.2584$), the mean gap on the validation set rises quickly. Under our size-invariant analysis, although $\gamma$ is still in the safe region $\gamma\in (0,\gamma_{\max})$, it is possible that due to Conjecture \bref{conjecture:opt-gamma-decreases}{1}, the range of good candidate values for $\gamma$ is also narrowing. As such, a more rigorous analysis would be required to determine if there is a better strategy for analysis beyond TSP-350, which we leave for future work. 

\begin{table}[H]
    \caption{Discretization Error across Multiple Magnitudes of $\Delta\gamma$. $\tilde{\gamma^*}$ is a proxy for $\gamma^*$ using a more refined gridsearch with step size $\Delta\gamma=0.0001$. Then we compute the upper bound of $\varepsilon(\Delta\gamma)$ as per Proposition \ref{prop:error-bound-on-gamma-gridsearch}. The values reported under the $\sup \varepsilon(\Delta\gamma)$ column are mean optimality gaps on the validation set, which is used to select the best value of $\gamma$ in the SIGS procedure. Relative error is given by $\frac{\sup \varepsilon(\Delta\gamma)}{f(\tilde{\gamma^*})}$, expressed as a percentage, and the "Time" column consists of the total time for selecting the parameter $\gamma$ and testing on the held out test set. } 
    \centering
    \small
    \begin{tabular}{cccccccc}
         $\Delta\gamma$ & $\gamma_-$ & $\gamma_+$ & $f(\gamma_-)$ & $f(\gamma_+)$ & $\sup \varepsilon(\Delta\gamma)$ & Relative Error & Time (min) \\
         \midrule
         \multicolumn{8}{l}{\textbf{TSP-10}, $\tilde{\gamma^*}=0.9111$, $f(\tilde{\gamma^*}) = 1.0405$}\\
         0.1 & 0.9 & 1.0 & 1.0426 & 1.0491 & 0.00218 & 0.21\% & 0.789\\
         0.01 & 0.91 & 0.92 & 1.0410 & 1.0424 & 0.000560 & 0.054\% & 7.10\\
         0.001 & 0.911 & 0.912 & 1.0405 & 1.0405 & 0.00 & 0.00\% & 59.9 \\
         \midrule
         \multicolumn{8}{l}{\textbf{TSP-100}, $\tilde{\gamma^*}=0.3590$, $f(\tilde{\gamma^*}) = 1.1701$}\\
         0.1 & 0.3 & 0.4 & 1.2220 & 1.1788 & 0.00875 & 0.75\% & 2.73\\
         0.01 & 0.35 & 0.36 & 1.1916 & 1.1712 &  0.00107 & 0.092\% & 19.2\\
         0.001 & 0.359 & 0.360 & 1.1701 & 1.1712 & 0.00 & 0.00\% & 132 \\
         \midrule
         \multicolumn{8}{l}{\textbf{TSP-350}, $\tilde{\gamma^*}=0.2584$, $f(\tilde{\gamma^*}) = 1.2103$}\\
         0.1 & 0.2 & 0.3 & 1.2303 & 1.2303 & 0.0200 & 1.65\% & 95.8 \\
         0.01 & 0.25 & 0.26 & 1.2142 & 1.2125 & 0.00221 & 0.18\% & 695\\
         0.001 & 0.258 & 0.259 & 1.2113 & 1.2113 & 0.00101 & 0.083\% & 3680 \\
         \bottomrule
    \end{tabular}
    \label{tab:discretization-error}
\end{table}
\pagebreak

\fi

\end{document}